\documentclass[onecolumn,longbibliography,nofootinbib, superscriptaddress,10pt,aps,prx]{revtex4-2}
\usepackage[dvips]{graphicx} 
\usepackage{amsfonts}
\usepackage{amssymb}
\usepackage{amscd}
\usepackage{amsmath}    
\usepackage{amsthm}
\usepackage{appendix}
\usepackage{bbm}
\usepackage{dsfont}
\usepackage{enumerate}
\usepackage{epsfig}
\usepackage[colorlinks, linkcolor=blue, anchorcolor=blue, citecolor=blue]{hyperref}
\usepackage{makecell}
\usepackage{xcolor}		
\usepackage{physics}
\usepackage[most]{tcolorbox}
\usepackage{tabularx}
\usepackage{MnSymbol}
\usepackage{tikz}
\usepackage{tikzpeople}
\usetikzlibrary{arrows}
\usetikzlibrary{shapes,fadings,decorations}
\usetikzlibrary{decorations.pathmorphing,patterns}
\usetikzlibrary{calc}
\usetikzlibrary{positioning}

\newtheorem{theorem}{Theorem}

\newtheorem{lemma}{Lemma}
\newtheorem{corollary}{Corollary}

\newtheorem{observation}{Observation}
\newtheorem{definition}{Definition}

\newtcolorbox[auto counter]{mybox}[2][]{
	enhanced,
	colback=blue!5!white,
	colframe=blue!75!black,
	fonttitle=\bfseries,
	title=Box \thetcbcounter: #2,#1
}

\newcommand{\bE}{\mathbb{E}}

\newcommand{\TV}{\mathrm{TV}}

\begin{document}
\title{Exponential Separations between Quantum Learning with and without Purification}

\author{Zhenhuan Liu}
\thanks{ZL, WG, and ZD contributed equally to this work.\\\href{mailto:liu-zh20@mails.tsinghua.edu.cn}{liu-zh20@mails.tsinghua.edu.cn}\\ \href{mailto:wgong@g.harvard.edu}{wgong@g.harvard.edu}\\ \href{mailto:du-zy23@mails.tsinghua.edu.cn}{du-zy23@mails.tsinghua.edu.cn}\\ \href{cai.zhenyu.physics@gmail.com}{cai.zhenyu.physics@gmail.com}}
\affiliation{Center for Quantum Information, Institute for Interdisciplinary Information Sciences, Tsinghua University, Beijing 100084, China}
\author{Weiyuan Gong}
\thanks{ZL, WG, and ZD contributed equally to this work.\\\href{mailto:liu-zh20@mails.tsinghua.edu.cn}{liu-zh20@mails.tsinghua.edu.cn}\\ \href{mailto:wgong@g.harvard.edu}{wgong@g.harvard.edu}\\ \href{mailto:du-zy23@mails.tsinghua.edu.cn}{du-zy23@mails.tsinghua.edu.cn}\\ \href{cai.zhenyu.physics@gmail.com}{cai.zhenyu.physics@gmail.com}}
\affiliation{School of Engineering and Applied Sciences, Harvard University, 150 Western Ave., Allston, Massachusetts 02134, USA}
\author{Zhenyu Du}
\thanks{ZL, WG, and ZD contributed equally to this work.\\\href{mailto:liu-zh20@mails.tsinghua.edu.cn}{liu-zh20@mails.tsinghua.edu.cn}\\ \href{mailto:wgong@g.harvard.edu}{wgong@g.harvard.edu}\\ \href{mailto:du-zy23@mails.tsinghua.edu.cn}{du-zy23@mails.tsinghua.edu.cn}\\ \href{cai.zhenyu.physics@gmail.com}{cai.zhenyu.physics@gmail.com}}
\affiliation{Center for Quantum Information, Institute for Interdisciplinary Information Sciences, Tsinghua University, Beijing 100084, China}
\author{Zhenyu Cai}
\affiliation{Department of Materials, University of Oxford, Parks Road, Oxford OX1 3PH, United Kingdom}
\affiliation{Quantum Motion, 9 Sterling Way, London N7 9HJ, United Kingdom}

\begin{abstract}
In quantum learning tasks, quantum memory can offer exponential reductions in statistical complexity compared to any single-copy strategies, but this typically necessitates at least doubling the system size.
We show that such exponential reductions can also be achieved by having access to the purification of the target mixed state.
Specifically, for a low-rank mixed state, only a \emph{constant} number of ancilla qubits is needed for estimating properties related to its purity, cooled form, principal component and quantum Fisher information with \emph{constant} sample complexity, which utilizes single-copy measurements on the purification.
Without access to the purification, we prove that these tasks require \emph{exponentially} many copies of the target mixed state for any strategies utilizing a bounded number of ancilla qubits, even with the knowledge of the target state's rank.
Our findings also lead to practical applications in areas such as quantum cryptography.
With further discussions about the source and extent of the advantages brought by purification, our work uncovers a new resource with significant potential for quantum learning and other applications.
\end{abstract}

\maketitle

\vspace{2em}
\tableofcontents

\vspace{3em}
\section{Introduction}\label{sec:intro}
Methods for learning and predicting properties of unknown quantum states are key subroutines in quantum information science, facilitating the benchmarking and verification of quantum devices and the study of many-body physics~\cite{eisert2020quantum,Daley2022analog,elben2023randomized}. 
The properties of quantum systems can be much more naturally probed using techniques that leverage quantum resources than their classical counterparts. 
The recent rapid advance of quantum technologies on various fronts, such as signal sensing, computational capacity, and information storage, offers the potential for efficient implementations of these learning techniques~\cite{ludlow2015clock,Wang2017memory,Arute2019supermacy,Bluvstein2024rydberg}.

Quantum memory is one of the most crucial resources for quantum learning. 
Recent works have shown that quantum memory can reduce statistical complexity, like sample and query complexities, in learning from quantum data~\cite{anshu2024survey}. 
Specifically, polynomial statistical separations between quantum learning protocols with and without quantum memory~\cite{bubeck2020entanglement,chen2022toward,chen2022tight,chen2024optimalstate,aaronson2024quantum,fawzi2023quantum} have been proven.
The subsequent works~\cite{huang2021information,aharonov2022quantum,chen2022memory,chen2022quantum,oh2024entanglement} further reveal exponential separations between protocols with and without quantum memory in various tasks including shadow tomography, physical quantum system learning, unitary learning, circuit learning, and quantum dynamics learning. 
Moreover, fine-grained trade-offs between sample complexity and quantum memory requirements are characterized for learning Pauli observables~\cite{chen2022memory,chen2024optimal}, testing purity~\cite{chen2024optimal}, and learning Pauli channel eigenvalues~\cite{chen2022quantum,chen2023efficient,chen2024tight}.
Some of these advantages have been further experimentally demonstrated on near-term devices~\cite{Huang_2022_quantum,bluvstein2024logical,seif2024entanglement}.

In quantum state learning tasks, the statistical advantages brought by quantum memory stem from the capability for joint quantum operations.  
Such operations normally require doubling the system size and sometimes even a polynomial qubit overhead~\cite{chen2021hierarchy,huang2021information}, leading to significant hardware requirements in their implementations.
Therefore, it is natural to ask:

\begin{center}
\emph{Is there a quantum resource that can achieve an exponential separation in learning with only a constant number of ancilla qubits?}
\end{center}

\begin{figure*}[htbp]
\centering
\includegraphics[width=0.7\linewidth]{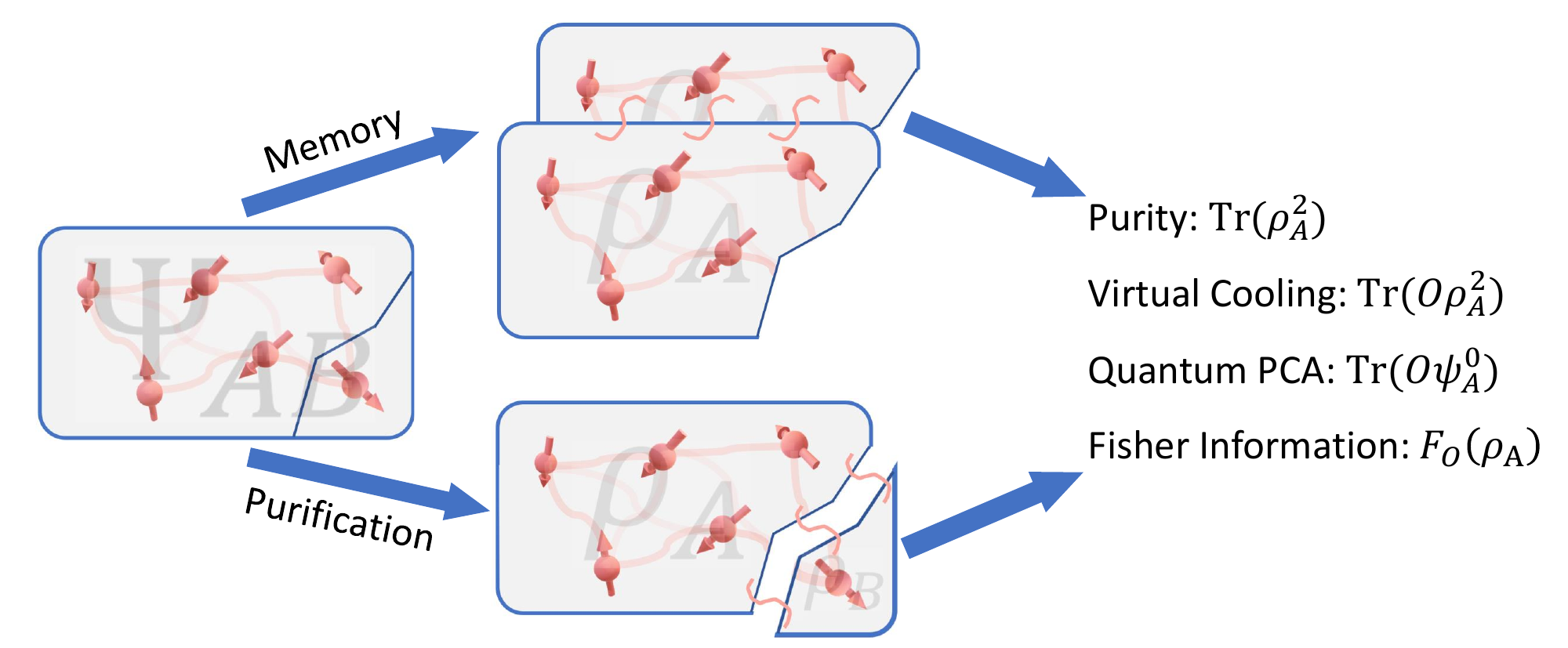}
\caption{We consider learning an unknown quantum mixed state $\rho_A$, which is the reduced density matrix of a global pure state, $\rho_A=\Tr_B(\Psi_{AB})$.
Compared with $\rho_A$, $\Psi_{AB}$ has a constant number of additional ancilla qubits with reduced density matrix being $\rho_B$. 
In certain learning tasks, such as purity estimation, quantum virtual cooling, quantum principal component analysis (quantum PCA), and quantum Fisher information estimation, joint operations among many identical copies bring exponential sample complexity advantages compared with single-copy strategies.
However, executing these tasks typically necessitates doubling the system size, which can be resource-intensive.
In this work, we prove that the ability to perform quantum operations on $\Psi_{AB}$, which only contains a constant number of additional qubits, also leads to an exponential separation in sample complexity when compared to protocols that utilize single-copy operations or even with a bounded number of memory qubits.
}
\label{fig:overview}
\end{figure*}

In this work, we answer this question affirmatively by considering the concept of \emph{purification}, which refers to a pure state $\Psi_{AB}=\ketbra{\Psi_{AB}}{\Psi_{AB}}$ whose reduced density matrix is the target mixed state we wish to probe, $\rho_A=\Tr_B(\Psi_{AB})$.
Our motivation for considering purification as a resource for quantum learning is multifaceted.
From an information-theoretic perspective, a mixed state indicates that some information from the target system leaks into the environment as a result of the entanglement between the system and the environment.
Therefore, estimating certain properties of the mixed state can be much easier with access to the global freedoms contained in the purification. 
Moreover, based on Uhlmann's theorem and transformation~\cite{uhlmann1976transition}, purification serves as a powerful theoretical tool that has been widely applied in various domains, including quantum inner product calculation, quantum Shannon entropy, quantum cryptography, quantum gravity, and quantum complexity theory~\cite{bostanci2023unitary}.
Accessing or preparing the purification of an unknown quantum state also arises as a key subroutine in a wide range of quantum tasks. 
For example, analog quantum simulation often involves preparing a global quantum many-body state and estimating its local properties like the entanglement Hamiltonian~\cite{Kokail2021entanglement}, entanglement entropy~\cite{Islam2015entropy,adam2016thermalization,brydges2019renyi}, and correlation functions~\cite{Shao2024FHM}.
In these cases, even though the properties of the local mixed states are of interest, we naturally have access to the global pure state.
Another practical scenario involving purification is quantum cryptography, including quantum key distribution~\cite{bennett1984quantum,ekert1991quantum} and quantum random number generator~\cite{herrero2017qrng}, where the complement system of the purification is at risk of being acquired by a third party seeking to steal the secret key or the random number. 
Additionally, certain quantum simulation algorithms, such as block encoding, require the purification-preparation oracle to encode a mixed state~\cite{Low2019hamiltonian,wang2024entropy,liu2024estimatingtracequantumstate}.
Furthermore, preparing and manipulating purification play vital roles in entanglement distillation~\cite{divincenzo1999entanglement} and mixed state compiling~\cite{ezzell2023quantum}.

Assume the target $n$-qubit state $\rho_A$ is low-rank and can be purified to the pure state $\Psi_{AB}$ with a constant number of ancilla qubits.
We show that many learning tasks can be achieved with constant sample complexity using $\Psi_{AB}$, which is established by designing new purification-assisted quantum learning protocols.
\begin{theorem}[Informal]
Given an $n$-qubit state $\rho_A$ that can be purified to $\Psi_{AB}$ with a constant number of ancilla qubits in system $B$, the quantities of $\Tr(\rho_A^t)$, $\Tr(O\rho_A^t)$, $\Tr(O\psi_A^0)$, and $F_O(\rho_A)$ can be accurately estimated with sample complexity $\mathcal{O}(1)$ by single-copy operations on $\Psi_{AB}$. 
Here $O$ represents some observable with bounded norm, $t\ge2$ is an integer, $\psi_A^0=\ketbra{\psi_A^0}{\psi_A^0}$ is the eigenstate of $\rho_A$ corresponding to its largest eigenvalue, and $F_O(\rho_A)$ denotes the quantum Fisher information.
\end{theorem}
\noindent These four quantities play important roles in quantum entropy (purity) estimation~\cite{Islam2015entropy,adam2016thermalization,brydges2019renyi,zhang2021shadow,Shaw2024benchmarking,gong2024sample}, quantum virtual cooling~\cite{cotler2019cooling}, quantum error mitigation~\cite{cai2023qem,koczorExponentialErrorSuppression2021,hugginsVirtualDistillationQuantum2021}, quantum principal component analysis~\cite{Lloyd2014qpca,Kimmel2017hamiltonian}, and quantum metrology~\cite{braunstein1994fisher,Giovannetti2011metrology}.
At the same time, we prove that these learning tasks are exponentially hard without purification, even when prior knowledge of the rank is available.
\begin{theorem}[Informal]
Any protocol capable of accurately estimating any one of $\Tr(\rho_A^t)$, $\Tr(O\rho_A^t)$, $\Tr(O\psi_A^0)$, or $F_O(\rho_A)$ using only single-copy operations on $\rho_A$ requires at least $\Omega(2^{n/2})$ sample complexity, even if the rank of $\rho_A$ is a constant and is known. 
Furthermore, even if one has access to $k\le n$ qubits of quantum memory and can interact with $\rho_A$ fewer than twice in a single experiment, a sample complexity at least $\Omega(\min\{2^{n/2},2^{n-k}\})$ is necessary.
\end{theorem}
\noindent Our results demonstrate that using a constant number of purification qubits leads to an exponential separation in sample complexity for these quantum learning tasks, compared to protocols that solely rely on single-copy measurements of $\rho_A$.
Furthermore, the distinction between learning with a constant number of purification qubits and memory qubits underscores the crucial role of purification in achieving an exponential learning advantage. 
An intuitive comparison between purification-assisted protocols and those utilizing quantum memory is illustrated in Fig.~\ref{fig:overview}.

Beyond quantum state learning, we explore the utility of purification in other contexts. 
We design new purification-assisted channel learning protocols and show the constant sample complexity in key tasks like unitarity estimation~\cite{montanaro2013survey,chen2023unitarity}, virtual channel purification~\cite{liu2024virtual}, and channel principal component analysis.
The sample complexity separations in state learning also lead to applications in quantum cryptography.
We design a protocol that enables the verification of a $2n$-qubit quantum computer using an $n$-qubit quantum computer.
In addition, we have constructed a blind observable estimation protocol that permits a client to measure an observable using a server, without revealing the correct expectation value to the server.


The remaining parts of the paper are organized as follows.
In Sec.~\ref{sec:upper_bound}, we propose purification-assisted protocols and show their constant sample complexity in purity estimation, quantum virtual cooling, quantum principal component analysis, and quantum Fisher information estimation.
These protocols are built from a series of simple yet insightful observations.
In Sec.~\ref{sec:lower_bound}, we reduce the learning tasks into carefully constructed state discrimination hard instances to prove the fundamental exponential sample complexity for protocols that only rely on single-copy operations on $\rho_A$ and less than $n$ memory qubits.
This section only contains main results and ideas, with all technical derivations left in Appendix~\ref{app:proof_hardness}.
In Sec.~\ref{sec:channel_learning}, we show how to design purification-based channel learning protocols for learning the information of unitarity, virtual channel distillation, and channel principal component analysis with constant sample complexity.
In Sec.~\ref{sec:crypto}, we discuss the applications of purification in quantum cryptography, proposing new quantum verification and blind observable estimation protocols.
In Sec.~\ref{sec:discussion}, we discuss the origin of the advantages of purification and compare the capabilities of purification and quantum memory.
In Sec.~\ref{sec:outlook}, we summarize our conclusions and propose some open problems.

\section{Constant Sample Complexity with Purification}\label{sec:upper_bound}

In this section, we focus on scenarios where the target quantum state $\rho_A$ on system $A$ can be purified to $\Psi_{AB}=\ketbra{\Psi_{AB}}$ with a constant number of ancilla qubits in system $B$ and physical operations allowed on $\Psi_{AB}$.
We present four examples to demonstrate how access to purification enables the achievement of constant sample complexity.

The sample complexity separation achieved through accessing purification can be intuitively understood through the Schmidt decomposition. 
Any bipartite pure state can be decomposed into
\begin{equation}\label{eq:schimdt}
\ket{\Psi_{AB}}=\sum_{j=0}^{2^{\abs{B}}-1}\sqrt{\lambda_j}\ket{\psi_A^j}\otimes\ket{\psi_B^j},
\end{equation}
where $\left\{\ket{\psi_A^j}\right\}_j$ and $\left\{\ket{\psi_B^j}\right\}_j$ are sets of mutually orthogonal state vectors. 
From this decomposition, the first observation is that the eigenvalue spectra of $\rho_A$ and $\rho_B$ are identical.
Therefore, the estimation of functions of $\rho_A$'s eigenvalues can be replaced by performing the same estimation on $\rho_B$.
Secondly, quantum operations on $\rho_B$ can filter the target eigenstates of $\rho_A$, enhancing the operations and measurements applicable to those eigenstates.
Therefore, if the qubit number in $\rho_B$ is significantly smaller than that in $\rho_A$, one can utilize $\rho_B$ to simplify learning tasks about $\rho_A$.

\subsection{Purity Estimation}
The most intuitive application is the purity estimation, a vital task in quantum information science finding applications in quantum benchmarking~\cite{Shaw2024benchmarking} and entanglement detection~\cite{Islam2015entropy,adam2016thermalization,brydges2019renyi,zhang2021shadow}.
The effectiveness of purification in purity estimation stems from a simple observation:
\begin{observation}\label{obs:purity}
Given a pure state $\Psi_{AB}=\ketbra{\Psi_{AB}}{\Psi_{AB}}$ and its reduced density matrices $\rho_A=\Tr_B(\Psi_{AB})$ and $\rho_B=\Tr_A(\Psi_{AB})$, we have $\Tr(\rho_A^2)=\Tr(\rho_B^2)$.
\end{observation}
\noindent This observation can be easily proved from the fact that the density matrices $\rho_A$ and $\rho_B$ have identical eigenvalue spectra.
Therefore, measuring the purity of the larger state $\rho_A$ is equivalent to measuring the purity of the small state $\rho_B$, which can be achieved with constant sample complexity.

\begin{theorem}
Given an $n$-qubit mixed state $\rho_A$, which is the reduced density matrix of a pure state $\Psi_{AB}$ with $n+\mathcal{O}(1)$ qubits, there exists a protocol based on single-copy measurements on $\Psi_{AB}$ that uses $\mathcal{O}(1)$ copies of $\Psi_{AB}$ to estimate $\Tr(\rho_A^2)$ within constant additive error.
\end{theorem}

\noindent Several single-copy purity estimation protocols have been developed. Among these, the randomized measurement protocol~\cite{elben2019statistical,elben2023randomized} reaches the optimal sample complexity concerning $n$ dependence~\cite{chen2022memory,gong2024sample}.
Specifically, the randomized measurement protocol can estimate the purity of an unknown $d$-dimension state to $\epsilon$ accuracy with sample complexity scaling as $\mathcal{O}\left(\max\left\{\frac{1}{\epsilon^2},\frac{\sqrt{d}}{\epsilon}\right\}\right)$.
Therefore, applying this protocol to estimate the purity of $\rho_B$ consisting of a constant number of qubits to constant additive error requires only $\mathcal{O}(1)$ sample complexity.
A more straightforward method is based on quantum tomography~\cite{KUENG2017tomo,chen2023adaptivity}.
In this approach, one first reconstructs an estimator for the density matrix of system $B$ and then computes the purity from this estimator.
In Appendix~\ref{app:purity_upper}, we prove the constant sample complexity of tomography-based purity estimation protocol and show its generalization to the estimation of higher-order moments $\Tr(\rho_A^t)$.

\subsection{Quantum Virtual Cooling}
The estimation of $\Tr(O\rho_A^t)$ with an integer $t\ge2$ and some observable $O$ is essential to tasks including quantum virtual cooling~\cite{cotler2019cooling} and quantum error mitigation~\cite{hugginsVirtualDistillationQuantum2021,koczorExponentialErrorSuppression2021}.
For a thermal state $\rho_A\propto e^{-\beta H}$ where $H$ is the system's Hamiltonian, estimating $\Tr(O\rho_A^t)$ allows one to extract information from the system at a lower temperature.
In the case of a pure state subjected to incoherent errors, estimating $\Tr(O\rho_A^t)$ assists in suppressing the noise rate, thereby improving the accuracy and reliability of measured properties.

When the observable $O$ has a bounded spectral norm, $\norm{O}_\infty\le\mathrm{1}$, accurately estimating the linear function $\Tr(O\rho_A)$ only requires constant sample complexity. This is because the variance of measuring $O$, given by $\Tr(O^2\rho_A)-\Tr(O\rho_A)^2\le\norm{O}_\infty^2$, is upper bounded by some constant. 
However, estimating $\Tr(O\rho_A^t)$ is more challenging due to its non-linearity. 
While numerous protocols leveraging quantum memory have been proposed \cite{cotler2019cooling,hugginsVirtualDistillationQuantum2021,koczorExponentialErrorSuppression2021,Zhou2024hybrid}, single-copy schemes are more suitable for near-term devices. 
Protocols such as classical shadow \cite{Huang2020predicting, hu2022logical, seif2023shadow_dist} have been developed in order to address this difficulty. 
However, an exponential sample complexity lower bound exists for single-copy protocols~\cite{Huang_2022_quantum}, underscoring the inherent difficulty of the task.

Here, we show that purification assists in the estimation of $\Tr(O\rho_A^t)$ based on the following observation:
\begin{observation}\label{obs:cooling}
Given a pure state $\Psi_{AB}=\ketbra{\Psi_{AB}}{\Psi_{AB}}$ and its reduced density matrices $\rho_A=\Tr_B(\Psi_{AB})$ and $\rho_B=\Tr_A(\Psi_{AB})$, we have $\rho_A^t=\Tr_B[\Psi_{AB}(\mathbb{I}_A\otimes\rho_B^{t-1})]$, where $\mathbb{I}$ denotes the identity matrix.
\end{observation}
\noindent This observation can be easily proved using the Schmidt decomposition. 
Leveraging this observation, we can design a protocol to estimate $\Tr(O\rho_A^t)$, relying only on single-copy operations on $\Psi_{AB}$.
The protocol first performs quantum tomography on system $B$ to reconstruct the classical estimator $\hat{\rho}_B$.
Next, we treat $O\otimes\hat{\rho}_B^{t-1}$ as the new observable and estimate it on state $\Psi_{AB}$, as $\Tr(O\rho_A^t)=\Tr\left[\Psi_{AB}\left(O\otimes\rho_B^{t-1}\right)\right]$.
Since system $B$ only contains a constant number of qubits and $\norm{O}_\infty\le\mathrm{1}$, both the construction of $\hat{\rho}_B$ and observable measurement require constant sample complexity, which results in the following theorem. 
The detailed proof is presented in Appendix~\ref{app:cooling_upper}. 
\begin{theorem}\label{thm:cooling}
Given an $n$-qubit mixed state $\rho_A$, which is the reduced density matrix of a pure state $\Psi_{AB}$ with $n+\mathcal{O}(1)$ qubits, there exists a protocol based on single-copy measurements on $\Psi_{AB}$ that uses $\mathcal{O}(1)$ copies of $\Psi_{AB}$ to estimate $\Tr(O\rho_A^t)$, with $\norm{O}_\infty=\mathcal{O}(1)$ and $t$ being an integer, within constant additive error.
\end{theorem}

\subsection{Quantum Principal Component Analysis}
Given a quantum state $\rho_A$, quantum principal component analysis aims to output the value of $\Tr(O\psi_A^0)=\bra{\psi_A^0}O\ket{\psi_A^0}$, where $\ket{\psi_A^0}$ is the eigenstate of $\rho_A$ corresponding to its largest eigenvalue.
Since the principal component of $\rho_A$ contains key information of the target quantum system, quantum principal component analysis finds applications in quantum error mitigation~\cite{hugginsVirtualDistillationQuantum2021,koczorExponentialErrorSuppression2021,wei2023realizing}, Hamiltonian simulation~\cite{Kimmel2017hamiltonian}, quantum machine learning~\cite{Lloyd2014qpca}, and classical data analysis~\cite{gordon2022covariance}.
Notably, the value of $\Tr(O\psi_A^0)$ can be rewritten as $\mathrm{lim}_{t\to\infty}\frac{\Tr(O\rho_A^t)}{\Tr(\rho_A^t)}$ when the largest eigenvalue of $\rho_A$ is non-degenerate.
Thus, quantum principal component analysis can be viewed as the extreme case of quantum virtual cooling and may encounter similar challenges for single-copy schemes.

Inspired by quantum steering~\cite{uola2020steering}, purification simplifies this task based on the following observation:
\begin{observation}\label{obs:pca}
Given a pure state $\Psi_{AB}=\ketbra{\Psi_{AB}}{\Psi_{AB}}$, its two reduced density matrices $\rho_A=\Tr_B(\Psi_{AB})$ and $\rho_B=\Tr_A(\Psi_{AB})$, principal components $\psi_A^0$ and $\psi_B^0$ of $\rho_A$ and $\rho_B$, respectively, and the corresponding eigenvalues $\lambda_A^0=\lambda_B^0$, we have $\psi_A^0=\frac{1}{\lambda_B^0}\Tr_B[\Psi_{AB}(\mathbb{I}_A\otimes\psi_B^0)]$.
\end{observation}

\noindent Essentially, this observation highlights the utility of purification in filtering out eigenstates of the target mixed state. Similar to quantum virtual cooling, Observation~\ref{obs:pca} implies that we can first reconstruct estimators $\hat{\psi}_B^0$ and $\hat{\lambda}_B^0$ using tomography and then estimate the observable $\frac{1}{\hat{\lambda}^0_B}O\otimes\hat{\psi}_B^0$ on $\Psi_{AB}$ to get the value of $\Tr(O\psi_A^0)$, thereby achieving a constant sample complexity. A caveat is that when two eigenvalues are close, accurately determining the corresponding eigenstates becomes difficult. To ensure $ \hat{\psi}_B^0 $ can be constructed accurately with a low sample complexity, we require the difference between the largest and second-largest eigenvalues of $\rho_B$ to be $\Theta(1)$. Under this assumption, we can design an efficient purification-assisted quantum principal component analysis protocol with constant sample complexity, as detailed in Appendix~\ref{app:qpca_upper}.

\begin{theorem}\label{thm:pca}
Assume that the difference between the largest and the second largest eigenvalues of the target $n$-qubit mixed state $\rho_A$ is $\Theta(1)$.
Given that $\rho_A$ is the reduced matrix of a pure state $\Psi_{AB}$ with $n+\mathcal{O}(1)$ qubits, there exists a protocol based on single-copy measurements on $\Psi_{AB}$, which uses $\mathcal{O}(1)$ copies of $\Psi_{AB}$ to estimate $\Tr(O\psi_A^0)$ within constant additive error, with $\norm{O}_\infty\le\mathcal{O}(1)$ and $\ket{\psi_A^0}$ being the eigenstate of $\rho_A$ with the largest eigenvalue.
\end{theorem}

\subsection{Quantum Fisher Information Estimation}\label{sec:q_fish}
Quantum Fisher information lies at the center of quantum metrology~\cite{braunstein1994fisher,Giovannetti2011metrology} and plays an important role in entanglement detection~\cite{hyllus2012fisher,toth2012fisher,li2013fisher,toth2014fisher,ren2021metrology}.
The challenge in estimating the value of quantum Fisher information stems from its definition.
Given an observable $O$ and a quantum state with the spectral decomposition $\rho_A=\sum_j\lambda_A^j\ketbra{\psi_A^j}{\psi_A^j}$, the Fisher information is defined as 
\begin{equation}\label{eq:QFI}
F_O(\rho_A)=2\sum_{j,k}\frac{(\lambda_A^j-\lambda_A^k)^2}{\lambda_A^j+\lambda_A^k}\abs{\bra{\psi_A^j}O\ket{\psi_A^k}}^2.
\end{equation}
Therefore, to estimate quantum Fisher information, one needs to estimate $\abs{\bra{\psi_A^j}O\ket{\psi_A^k}}^2$ for $j\neq k$, which is not straightforward without the spectral information about the target state $\rho_A$.
Conventional methods normally rely on statistical errors and state moments to compute the value of quantum Fisher information~\cite{yu2021fisher,rath2021fisher}.
Instead of calculating these values, we show that purification helps to estimate quantum Fisher information accurately.

The following observation indicates that the purification can effectively filter out the outer product of two different eigenstates of $\rho_A$.
\begin{observation}\label{obs:fisher}
Given a pure state $\Psi_{AB}=\ketbra{\Psi_{AB}}{\Psi_{AB}}$, its two reduced density matrices $\rho_A=\Tr_B(\Psi_{AB})$ and $\rho_B=\Tr_A(\Psi_{AB})$, two eigenvalues $\lambda_A^j=\lambda_B^j$ and $\lambda_A^k=\lambda_B^k$, and corresponding eigenstates $\ket{\psi_A^j}$, $\ket{\psi_A^k}$, $\ket{\psi_B^j}$, and $\ket{\psi_B^k}$, we have $\frac{1}{\sqrt{\lambda_B^j\lambda_B^k}}\Tr_B\left[\Psi_{AB}\left(\mathbb{I}_A\otimes\ketbra{\psi_B^j}{\psi_B^k}\right)\right]=\ketbra{\psi_A^k}{\psi_A^j}$.
\end{observation}
\noindent Using this observation, we can design a protocol to accurately estimate quantum Fisher information by using purification.
We show that $\abs{\bra{\psi_A^j}O\ket{\psi_A^k}}^2=\frac{1}{2\lambda_B^j\lambda_B^k}\Tr\left[\Psi_{AB}^{\otimes 2}\left(O^{\otimes 2}\otimes P_{B}^{jk}\right)\right]$, where $P_{B}^{jk}=\ketbra{\psi_B^j}{\psi_B^k}\otimes\ketbra{\psi_B^k}{\psi_B^j}+h.c.$. 
Thus, we can reformulate the quantum Fisher information as 
\begin{equation}\label{eq:new_fisher}
F_O(\rho_A)=\sum_{j,k}\frac{(\lambda_B^j-\lambda_B^k)^2}{\lambda_B^j\lambda_B^k(\lambda_B^j+\lambda_B^k)}\Tr\left[\Psi_{AB}^{\otimes 2}\left(O^{\otimes 2}\otimes P_{B}^{jk}\right)\right].
\end{equation}
To estimate quantum Fisher information, we first construct $P_B^{jk}$ using tomography on system $B$, followed by measuring the observable $O^{\otimes 2}\otimes P_{B}^{jk}$ on $\Psi_{AB}^{\otimes 2}$.
Although this remains a nonlinear function of $\Psi_{AB}$, $P_B^{jk}$ can be decomposed into a summation of tensor product forms,
\begin{equation}
  \hat{P}_B^{jk}=\frac{1}{2}\left(\hat{P}_{B+}^{jk}\otimes\hat{P}_{B+}^{jk}+\hat{P}_{B-}^{jk}\otimes\hat{P}_{B-}^{jk}\right)  
\end{equation}
with $\hat{P}_{B+}^{jk}=\ketbra{\hat{\psi}_B^j}{\hat{\psi}_B^k}+h.c.$ and $\hat{P}_{B-}^{jk}=i\ketbra{\hat{\psi}_B^j}{\hat{\psi}_B^k}+h.c.$, where $i$ is the unit imaginary number. 
Such tensor decomposition allows efficient estimation of the target expectation value via single-copy measurements and post-processing. 
In Appendix~\ref{app:fisher_upper}, we prove the following theorem:
\begin{theorem}
Assume that all nonzero eigenvalues and the absolute value of the difference between arbitrary two eigenvalues of the target $n$-qubit mixed state $\rho_A$ are $\Theta(1)$. 
Given that $\rho_A$ is the reduced density matrix of a pure state $\Psi_{AB}$ with $n+\mathcal{O}(1)$ qubits, there exists a protocol based on single-copy measurements on $\Psi_{AB}$ that uses $\mathcal{O}(1)$ copies of $\Psi_{AB}$ to estimate $F_O(\rho_A)$, with $\norm{O}_\infty\le\mathcal{O}(1)$, within constant additive error.
\end{theorem}

\section{Exponential Overhead without Purification}\label{sec:lower_bound}
After showing the four important quantum state learning tasks can be accomplished with constant sample complexity through purification, we delve into the challenge of achieving these same tasks via single-copy measurements without purification. 
Our analysis reveals that these tasks necessitate an exponential sample complexity without purification qubits.
This underscores the significant advantage offered by a constant number of purification qubits in quantum learning tasks.

Our proof reduces the task of learning the properties of states to distinguishing between two state ensembles~\cite{Huang_2022_quantum, chen2022memory}. 
The core logic behind this reduction is that, if an algorithm can accurately learn a particular quantum property, it should also be capable of successfully distinguishing between ensembles that exhibit different values of that property. 
Consequently, the difficulty of distinguishing tasks implies the fundamental challenges inherent in quantum learning tasks.
A key contribution of our proof is demonstrating that certain learning tasks require exponential sample complexity, even when the targeted mixed states have a \emph{constant known rank}. 

We further strengthen our results concerning learning with bounded additional memory qubits and limited quantum capabilities. 
Specifically, we consider the scenario where the quantum memory interacts with the input state $\rho$ twice, as depicted in Fig.~\ref{fig:bounded_memory}.
Initially, a positive operator-valued measure (POVM) is applied to the main register (containing the target state) and the ancillary register (containing the memory qubits). 
Following this, the ancillary register retains its state while the main register is reset to the target state. 
Subsequently, a second POVM is applied to both registers.
The measurement outcomes of these two POVMs can be utilized to predict quantum properties and address state distinguishing tasks. 
Note that when $k=0$, this algorithm reduces to all protocols utilizing only single-copy operations and measurements.
While $k=n$, any algorithms using two-copy measurements and operations are included.
Therefore, the learning model with $k$ qubits of quantum memory depicts a smooth transition between algorithms using single-copy and two-copy measurements.
We prove these learning tasks require exponential sample complexity even with bounded quantum memory, highlighting the significant role of a constant number of additional purification qubits in these tasks. 
The details of formalization and proof are provided in Appendix \ref{app:proof_hardness}.

\begin{figure}
\centering
\includegraphics[width=0.3\linewidth]{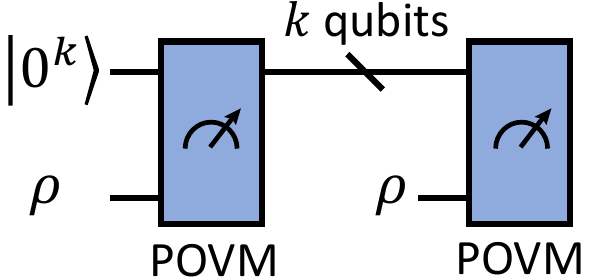}
\caption{The quantum state learning algorithm with bounded memory qubits and twice interactions with the target state $\rho$. When the number of additional memory qubits is zero, $k=0$, this algorithm reduces to any single-copy learning algorithm. When $k=n$, this algorithm includes all two-copy learning algorithms.}
\label{fig:bounded_memory}
\end{figure}

\subsection{Hardness of Purity Estimation}\label{sec:lower_bound_purity}
Our first result establishes the hardness of predicting the purity $\Tr(\rho^2)$ of a low-rank input state $\rho$ using single-copy measurements with bounded quantum memory.
We consider two state ensembles:
\begin{itemize}
\item Ensemble $\mathcal{S}_1$: $\rho=0.9U\ket{0}\bra{0}U^\dagger+0.1V\ket{0}\bra{0}V^\dagger$ is a rank-$2$ random state with Haar randomly chosen $U$ and $V$.
\item Ensemble $\mathcal{S}_2$: $\rho=\frac{1}{2}(U\ket{0}\bra{0}U^\dagger+V\ket{0}\bra{0}V^\dagger)$ is a rank-$2$ random state with Haar randomly chosen $U$ and $V$.
\end{itemize}
\noindent Suppose a random state $\rho$ is selected from either $\mathcal{S}_1$ or $\mathcal{S}_2$ with an equal probability $1/2$, and copies of the state are provided to an algorithm. 
The input state $\rho$ has a known constant rank and exhibits different values of purity with high probability depending on the chosen ensemble. 
Therefore, if the algorithm can estimate the purity within a constant additive error, it can distinguish between the two ensembles.
In Appendix~\ref{app:purity_lower}, we formalize the algorithm based on single-copy measurements via a learning tree representation \cite{Huang_2022_quantum, chen2022memory}.
We show that the total variation distance of the probability distributions at the leaf nodes of the learning tree is exponentially small for the two ensembles unless the depth of the learning tree is exponential with the qubit number. 
Consequently, we prove that distinguishing between $\mathcal{S}_1$ and $\mathcal{S}_2$ requires exponentially many single-copy measurements, which also serves as the sample complexity lower bound for purity estimation.

\begin{theorem}[Exponential overhead of purity estimation, informal, see Theorem~\ref{thm:purity_lower} and Theorem~\ref{thm:purity_lower_mem}] \label{thm:purity_informal}
Given an unknown $n$-qubit state $\rho$ and its rank (which can be constant), predicting its purity $\Tr(\rho^2)$ within constant additive error requires $\Omega(\min\{2^{n/2},2^{n-k}\})$ sample complexity for protocols using $0\leq k\leq n$ additional memory qubits and the circuit shown in Fig.~\ref{fig:bounded_memory}. 
\end{theorem}

\noindent Combining the results in Sec.~\ref{sec:upper_bound}, we establish an exponential separation between purity estimation with and without purification. 
Furthermore, this separation between learning with a constant number of additional purification qubits and memory qubits highlights the essential role of purification in achieving an exponential learning advantage.
It is easy to generalize this result from purity to higher-order moment estimation, as the states from different ensembles also have different values of $\Tr(\rho_A^t)$.

\subsection{Hardness of Quantum Virtual Cooling and Principal Component Analysis}

Our next result establishes the hardness of quantum virtual cooling, which involves predicting the expectation value $\Tr(O\rho^2)$ of an observable $O$, as well as the difficulty of quantum principal component analysis, which involves predicting the expectation value $\Tr(O\ketbra{\psi_0})$ without purification.
Here, $\ketbra{\psi_0}$ represents the principal component, i.e., the eigenstate corresponding to the largest eigenvalue of the input state $\rho$. 
As discussed in Sec.~\ref{sec:upper_bound}, both tasks aim to estimate an expectation values of observables based on higher-order powers of the input state $\rho$. 
Based on the similar state distinguishing technique, we prove that these tasks require exponential sample complexity, even when the algorithm is assisted by fewer than $n$ memory qubits and the knowledge of the eigenvalues of the input state. 

\begin{theorem}[Exponential overhead of quantum virtual cooling and principal component analysis, informal, see Theorem~\ref{thm:vc_pca_lower} and Theorem~\ref{thm:vc_pca_lower_mem}]
Given an unknown $n$-qubit state $\rho$, its rank (which can be constant) or even its eigenvalues, and observable $O$ satisfying $\norm{O}_{\infty} = 1$, predicting $\Tr(O\rho^2)$ or $\Tr(O\ketbra{\psi_0})$ within constant additive error requires at least $\Omega(\min\{2^{n/2},2^{n-k}\})$ samples for protocols using $0\leq k\leq n$ additional memory qubits and the circuit shown in Fig.~\ref{fig:bounded_memory}. 
Here, $\ketbra{\psi_0}$ is the principal component of $\rho$.
\end{theorem}

\noindent The proof is detailed in Appendix~\ref{app:cooling_pca_lower} and follows a similar argument to Theorem~\ref{thm:purity_informal}. 
We remark that a polynomial-time quantum principal component analysis algorithm has been proposed previously~\cite{Lloyd2014qpca}, which has a polynomial overhead on memory qubits and has a provable quantum advantage over any algorithm based on single-copy measurements~\cite{Huang_2022_quantum}. 
However, it was previously unknown whether this advantage persists when the rank of the target state is a constant and known. 
Our results confirm that the advantage remains and that a constant number of purification qubits suffices to achieve an exponential advantage over bounded quantum memories, further highlighting the necessity of purification for obtaining advantage.

\subsection{Hardness of Quantum Fisher Information Estimation}
Our fourth learning task considers estimating the quantum Fisher information of the input state $\rho$, a valuable component in many quantum information tasks. 
The quantum Fisher information is a highly nonlinear property of the input state, making its prediction challenging. 
Here, we prove that this task requires exponentially many single-copy measurements when performed without purification and with a bounded number of memory qubits, implying the advantage of having purification qubits in estimating quantum Fisher information.

\begin{theorem}[Exponential overhead of quantum Fisher information estimation, informal, see Theorem~\ref{thm:fisher_lower} and Theorem~\ref{thm:fisher_lower_mem}]
Given an unknown $n$-qubit state $\rho$, its rank (which can be constant) and observable $O$ satisfying $\norm{O}_{\infty} = 1$, predicting the quantum Fisher information $F_O(\rho)$ within constant additive error requires at least $\Omega(\min\{2^{n/2},2^{n-k}\})$ sample complexity for protocols using $0\leq k\leq n$ additional memory qubits and the circuit shown in Fig.~\ref{fig:bounded_memory}. 
\end{theorem}

\noindent The proof also shares similarities as in Theorem~\ref{thm:purity_informal}, which is detailed in Appendix \ref{app:fisher_lower}.
Notably, this is the first establishment of the hardness in estimating quantum Fisher information.
The established lower bound indicates that learning quantum Fisher information remains challenging even when the quantum memory is relatively large, specifically for $k = n - \omega(\log n)$.
In Sec.~\ref{sec:discussion}, we show that the quantum Fisher information for low-rank quantum states can be estimated using the quantum principal component analysis algorithm with a polynomial number of additional memory qubits.

\section{Quantum Cryptography}\label{sec:crypto}
We have demonstrated that purification qubits provide advantages in certain learning tasks compared to memory qubits. 
Here, we leverage the power of purification in quantum cryptography by proposing two protocols. 
In the first protocol, a client uses mixed states to verify the server's computational capabilities. 
In the second protocol, a client requests the server to perform complex observable measurements while ensuring the expectation value remains private.

\subsection{Quantum Verification}
In the first application, a client aims to use a less powerful quantum computer to verify the capabilities of a more powerful server-side quantum computer, as depicted in Fig.~\ref{fig:verification}. 
The basic idea involves sending copies of a state to the serve, asking it to estimate specific properties, and verifying the results using purification. 
The hardness results from Sec.~\ref{sec:lower_bound} suggest that without sufficient computational power, the server cannot predict certain kinds of properties. 
Conversely, accurate estimations would demonstrate the server's quantum computation power. 
To verify the correctness of these estimations, the client uses the pure state and applies the estimation protocols discussed in Sec.~\ref{sec:upper_bound} to achieve constant sample complexity. 
Ultimately, the client can assess the server's computational power by comparing their estimations.

\begin{figure}[htbp]
\centering
\includegraphics[width=0.7\linewidth]{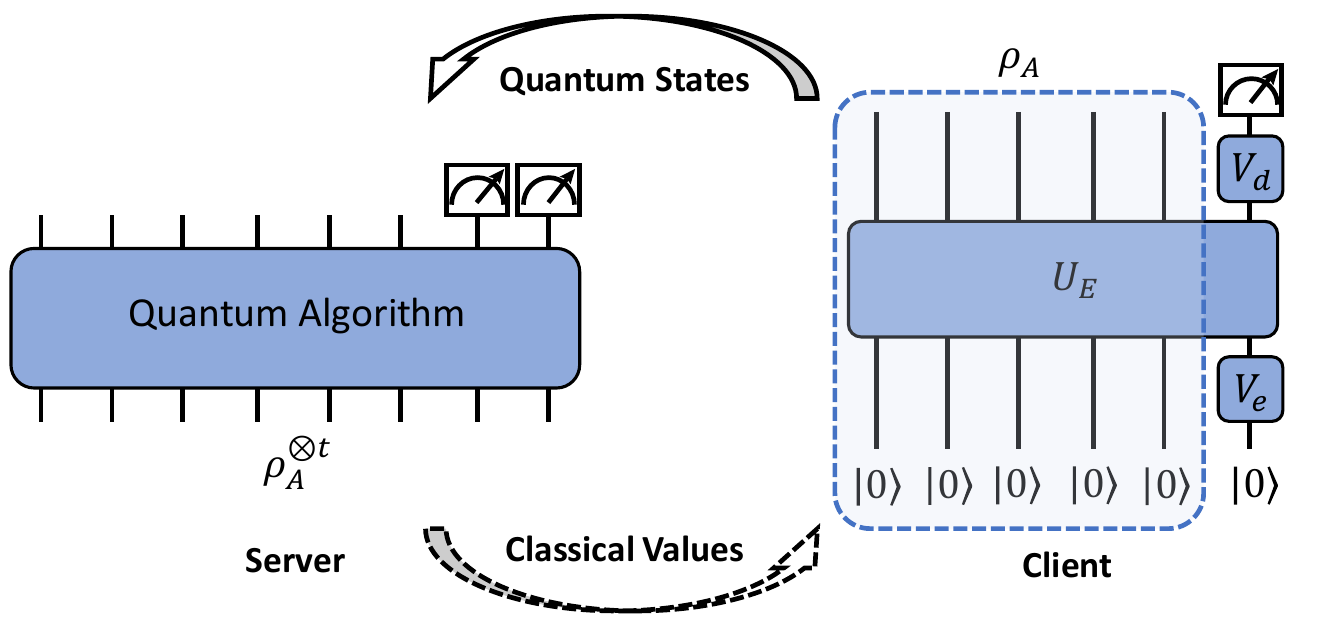}
\caption{Protocol for quantum verification and blind observable estimation. 
The client prepares the state $\Psi_{AB} = UV_e\ketbra{0}_{AB}V_e^{\dagger}U^{\dagger}$ and sends a constant number $t$ of copies of $\rho_A = \Tr_B(\Psi_{AB})$ to the server. 
The server is then asked to estimate certain properties of $\rho_A$ and report the classical values. 
Meanwhile, the client can measure system $A$ and the purification system $B$ on a rotated basis defined by $V_d$ to verify the server's computational power.
In addition, by post-selecting the data from the server using measurement results on system $B$, the client can achieve blind observable estimation.
}
\label{fig:verification}
\end{figure}

Specifically, assume the client possesses a quantum computer with systems $A$ and $B$, where $A$ is an $n$-qubit system and $B$ is an $\mathcal{O}(1)$-qubit system. 
The client seeks to verify whether the server can reliably manipulate a larger quantum computer with at least $2n$ qubits or has a powerful quantum memory capable of interacting with the input state more than twice. 
First, the client uses a small encoding unitary $V_e$ to encode certain information in system $B$ and then evolves the joint system using an entangling gate $U_E$. 
The client repeatedly prepares the state $\Psi_{AB} = U_EV_e\ketbra{0}_{AB}V_e^{\dagger}U_E^{\dagger}$ and sends the $n$-qubit state $\rho_A=\Tr_B(U_EV_e\ketbra{0}_{AB}V_e^{\dagger}U_E^{\dagger})$, which is now a mixed state, to the server for a constant times $t=\mathcal{O}(1)$.
The client then requests the server to estimate properties of $\rho_A$, such as purity or principal component information, and report the results to the client. 
The server can perform quantum operations on $\rho^{\otimes t}_A$ to extract the desired information. 
Meanwhile, since the client has access to the purification $\Psi_{AB}$, it can easily estimate the properties using its device with $n+\mathcal{O}(1)$ qubits and verify the server's estimation.

To illustrate the relationship between the quantum verification protocol and our proof of the sample complexity lower bound, we take the purity estimation as an example.
In Sec.~\ref{sec:lower_bound_purity}, we prove the difficulty of purity estimation by constructing two mixed-state ensembles, $\mathcal{S}_1$ and $\mathcal{S}_2$.
Note that, by setting $U_E=\ketbra{0}{0}\otimes U+\ketbra{1}{1}\otimes V$ and $V_e\ket{0}=\alpha\ket{0}+\sqrt{1-\alpha^2}\ket{1}$ and adjusting the value of $\alpha$, states in the two ensembles can be prepared in state $\rho_A$.
Thus, if one sets $\alpha=\sqrt{0.9}$ and $\alpha=\sqrt{0.5}$ with equal probabilities and asks the server to report the value of the purity of $\rho_A$, the function of this quantum verification protocol is guaranteed by our sample complexity lower bound analysis.
At the same time, as the client holds the state of the single-qubit state $\rho_B$, whose purity is equivalent with the purity of $\rho_A$, it is easy for the client to verify the result reported by the server.

According to the hardness results in Sec.~\ref{sec:lower_bound}, if the server can report accurate results after receiving only constant copies of $\rho_A$, it indicates that the server can effectively manipulate a large number of memory qubits or possesses a robust quantum memory capable of interacting with many input copies. 
Moreover, as the evolution $U$ and $V_e$ are not required to take on any special form, this verification protocol can be efficiently implemented via native Hamiltonian evolution or pseudo-random circuits~\cite{ji2018pseudo}, which is proved to be of $\mathrm{log}(n)$ depth~\cite{schuster2024random}. 
We summarize the quantum verification protocol using Box~\ref{box:verification_protocol}.

\subsection{Blind Observable Estimation}
Quantum computers are powerful not only for their speedups compared with classical computers, but also the privacy in certain computational tasks.
Consider the case where a client has limited quantum computation power and entrusts a server to perform some part of a quantum information processing task.
The blind quantum computation is proposed to ensure that the server reliably achieves the task and at the same time cannot obtain certain information of the task~\cite{broadbent2009bqc,stefanie2012BQC,Fitzsimons2017bqc}.
We will show that, with the assistance of purification, we can design a protocol to achieve blind observable estimation. 

\begin{mybox}[float,label={box:verification_protocol}]{Quantum verification}
    \textbf{Target: } 
    \begin{enumerate}[]
        \item
        The client uses a quantum computer with systems $A$ and $B$, where $A$ is an $n$-qubit system and $B$ is an $\mathcal{O}(1)$-qubit system, to verify whether the server can reliably manipulate a larger quantum computer with at least $2n$ qubits, or if it has a quantum memory capable of reliably interacting with the input state more than twice.
    \end{enumerate} 
    \textbf{Protocol:} 
    \begin{enumerate}[(1)]
		\item
		\emph{State preparation and transmission}: The client prepares the pure state $\Psi_{AB}$ by applying a small encoding unitary $V_e$ on system $B$ and evolve the joint system using an entangling evolution $U_E$. The client then sends a constant number of copies of the state $\rho_A$ to the server.
            \item 
		\emph{Property estimation}: The server is asked to learn properties of the state $\rho_A$ and report the results to the client. The client estimates the properties using the pure state $\Psi_{AB}$.
            \item 
		\emph{Verification}: The client verifies the server's capability by checking whether the server's estimation value is correct.
    \end{enumerate}
\end{mybox}

Consider a scenario where the client has a large analog quantum simulator adept at preparing quantum many-body states but struggles with performing measurements~\cite{Daley2022analog,tran2023measuring}. 
The client aims to predict the expectation value $\bra{\psi} O \ket{\psi}$ of a complicated observable $O$ on a target state $\ket{\psi}$, where  $\ket{\psi} = U \ket{0^n}$ and $U$ is generated by a Hamiltonian evolution $U = e^{-iHt}$ with a Hamiltonian that is native to the client's simulator. 
However, the measurement can only be performed by the server, and the client wishes to keep the correct value hidden from the server.

To achieve this, the client can conceal information about the state $\ket{\psi}$ by sending a mixed state. 
As shown in Fig.~\ref{fig:verification}, the unitary $V_e$ is set to be the identity gate, and the entangling gate $U_E$ is divided into two parts.
The first part of $U_E$ is a two-qubit entangled gate acting on the single-qubit system $B$ and one qubit in system $A$, resulting in the state of $\frac{1}{\sqrt{2}}\ket{0^{n-1}}_A\left(\ket{0}_A\ket{0}_B+\ket{1}_A\ket{1}_B\right)$.
Now, the density matrix of system $A$ is $\rho_A=\frac{1}{2}\ketbra{0^{n-1}}{0^{n-1}}\otimes\left(\ketbra{0}{0}+\ketbra{1}{1}\right)$.
Then, the client evolves the state by $U=e^{-iHt}$ on system $A$, resulting in the mixed state
\begin{equation}
\rho_A= \frac{1}{2}(\ketbra{\psi}+\ketbra{\psi_\perp}),
\end{equation}
where the state $\ket{\psi_\perp}$ is orthogonal to $\ket{\psi}$. 
The client then sends copies of $\rho_A$ to the server and requests the server to perform a single-copy observable measurement with the given observable $O$ and report the measurement results. 
Since the server can only access the degenerate mixed state, it cannot determine the expectation value. 
In contrast, the client retains the purification and can measure the purification qubit in the computational basis, which collapses $\rho_A$ to $\psi$ or $\psi_\perp$ with equal probabilities. 
Then, the client only keeps the reported result if the measurement result is $0$. 
After repeating this procedure many times, the client can average the kept results to obtain an accurate estimation.

Simultaneously, because the server is unaware of the measurement result of the purification qubit, it cannot derive the correct measurement value from its data. 
Thus, the client can estimate the expectation value without revealing it to the server, effectively implementing a blind observable estimation protocol. 
We summarize this protocol in Box~\ref{box:bqc_protocol}.

\begin{mybox}[float,label={box:bqc_protocol}]{Blind observable estimation}
 \textbf{Target: } 
    \begin{enumerate}[]
        \item
        The client aims to predict the value $\bra{\psi}O\ket{\psi}$ for the state $\ket{\psi} = U\ket{0^n}$, where $O$ is a complicated observable that can only be measured by the server. However, the client does not want the server to know the value of $\bra{\psi}O\ket{\psi}$.
        
    \end{enumerate} 
    \textbf{Protocol:} 
    \begin{enumerate}[(1)]
		\item
		\emph{State preparation and transmission}: The client prepares the state $\Psi_{AB}$ by applying a small entangling gate between $A$ and the single-qubit system $B$ and evolves system $A$ with $U$. The client then sends copies of the state $\rho_A$ to the server.
            \item 
		\emph{Observable measurement}: The server measures the observable $O$ on the state $\rho_A$ and reports the results. Meanwhile, the client measures system $B$ on the computational basis. If the outcome is $0$, the client keeps the reported result from the server.
            \item 
		\emph{Expectation value estimation}: The client estimates the expectation value by averaging the kept results.
    \end{enumerate}
\end{mybox}

Unlike the quantum verification protocol discussed previously, the security of this protocol does not rely on the sample complexity separations.
Instead, it directly depends on the degeneracy of $\rho_A$, which is a mixture of two orthogonal states with equal weights.
Therefore, even if the server has an unbounded number of memory qubits and arbitrary quantum computational power, and obtains the matrix form of $\rho_A$, it still cannot extract the accurate information of $\ket{\psi}$.

Moreover, the client can verify that the server conducts observable estimation honestly. 
To achieve this, the client prepares a set of observables, including both target observables and test observables with known expectation values, such as those that commute with the Hamiltonian.
Since the server lacks knowledge of the target observables, it risks performing dishonest operations. 
Thus, the client can leverage its knowledge of the evolution Hamiltonian to verify the server's report. 
If the server provides correct measurement results on all test observables, the client can trust the measurement results on the target observables.

\subsection{Necessity of Purification}\label{sec:necessity_purification}
A natural question regarding our protocols is the necessity of purification. 
It seems feasible to design protocols with similar functionality without purification. 
For instance, instead of preparing a pure state and sending a subsystem to the server, one could directly prepare and send a mixed state. 
There are two typical methods for preparing the mixed state.
The first one involves randomly initializing some qubits as the mixed state $\alpha\ketbra{0}+(1-\alpha)\ketbra{1}$ while keeping others in the state $\ket{0}$ followed by evolving the entire system with $U$.
The number of initially mixed qubits and the value of $\alpha$ determines the purity and the principal component of the resulting state.
Another method relies on classical randomness. 
One can flip a coin to obtain some random results and send random pure states based on those results. 
If the server cannot access the classical randomness, the effective state perceived by the server will be a mixed state.

The critical issue for these protocols utilizing mixed states is security or the potential for information leakage. 
If the randomness in preparing the mixed state is generated classically, such as by flipping a coin, or if the mixed state arises from decoherence, then the purification of the mixed state may inadvertently leak into the environment.
Thus, the purification is at risk of being accessed by the server and used to cheat the client. 
On the other hand, if the mixed state is prepared by entangling it with another qubit in the client's quantum processor, this approach is equivalent to our purification-assisted protocol.

\section{Discussion}\label{sec:discussion}
\subsection{Generalization Beyond Purification}

So far we have been discussing the various tasks we can perform if we have access to the purification for the mixed state of interests $\rho_A$. 
In the global pure state $\Psi_{AB}$, subsystems $A$ and $B$ share entanglement, which brings huge advantages in many quantum learning tasks.
A natural question is, can we achieve the same advantage with only classical correlation?
An intuitive example is that two $n$-bit probability distributions, $\mathbf{p}_n$ and $\mathbf{q}_n$, are originally hard to distinguish.
However, these two distributions are marginal distributions of two $(n+1)$-bit probability distributions, $\mathbf{p}_{n+1}$ and $\mathbf{q}_{n+1}$, and their complement single-bit marginal distributions, $\mathbf{p}_1$ and $\mathbf{q}_1$ share different properties.
Therefore, one can easily distinguish $\mathbf{p}_{n+1}$ and $\mathbf{q}_{n+1}$ by estimating properties of $\mathbf{p}_1$ and $\mathbf{q}_1$, which is much easier than the original problem.
Note that this simple example is analogous to the purification-assisted purity estimation protocol shown in Observation~\ref{obs:purity}, in which one can estimate the purity of a large system by only operating a much smaller complement system.

To generalize this simple example and further explore the power of classical correlation, we consider the state of 
\begin{align}\label{eq:mixed_composite}
    \rho_{AB} = \sum_{j=0}^{2^{\abs{B}-1}} \lambda_j \ketbra{\psi_A^j}\otimes\ketbra{\psi_B^j},
\end{align}
which is the incoherent version of purification, with $\ket{\psi_A^j}$ and $\ket{\psi_B^j}$ being eigenstates of reduced density matrices $\rho_A$ and $\rho_B$.
One can check that, with this state, the observations, protocols, and theorems discussed in Sec.~\ref{sec:upper_bound} still apply to purity estimation, quantum virtual cooling, and principal component analysis, achieving the same performance as the purification. 
The tasks are those effectively trying to filter the eigenvalues and eigenstates of $\rho_A$, and then perform some measurement on them. 
In these tasks, the necessary ingredient is the classical correlations between systems $A$ and $B$, not the quantum entanglement between these states (e.g. the coherence among the Schmidt basis states in the composite system). 
For principal component analysis, we do not even need perfect classical correlation between the eigenstates of $A$ and $B$ like in Eq.~\eqref{eq:mixed_composite}, we only require the classical conditional probability for the states to satisfy that system $A$ is always in $\ket{\psi_{A}^{0}}$ when system $B$ is in $\ket{\psi_{B}^0}$, and the reverse need not be true. 

On the other hand, we need to measure the cross terms between different eigenstates for estimating Fisher information, $\abs{\bra{\psi_A^j}O\ket{\psi_A^k}}^2$, which is beyond the paradigm described above. 
In such a case, the quantum entanglement between the states in different subsystems (i.e. the coherence among the Schmidt basis states in the composite system) is essential. 
We can easily see that the protocols outlined in Sec.~\ref{sec:q_fish} do not work if we are given the incoherent state Eq.~\eqref{eq:mixed_composite} instead.
In Appendix~\ref{app:fisher_classical_corr}, we provide a rigorous proof that having classical correlation cannot improve the sample complexity for estimating quantum fisher information, which is summarized as following:
\begin{theorem}[Informal]
Given an $n$-qubit low-rank state $\rho_A$, its rank, target observable $O$ with bounded norm, and access to a constant number of ancillary qubits that are classically correlated with $\rho_A$ in the form of Eq.~\eqref{eq:mixed_composite}, any protocol that accurately estimates quantum Fisher information requires at least $\Omega(2^{n/2})$ complexity.
\end{theorem}

We would like to point out some key differences between purification and classical correlation.
Although the classical correlation given in Eq.~\eqref{eq:mixed_composite} can partially replace the role of purification, it has a very specific form: the mixture of tensor product of eigenstates of two subsystems.
This form is particularly vulnerable to unitary perturbations applied to the joint system, which can undermine the utility of the classical correlation.
However, the utility of purification is robust to arbitrary unitary operation.
Furthermore, as discussed in Sec.~\ref{sec:necessity_purification}, mixed states cannot ensure information security and thus cannot replace the role of purification in quantum cryptography.

\subsection{Comparison with Quantum Memory}
In this work, we show that in certain quantum state learning tasks, purification provides exponential sample complexity separations compared with single-copy measurements and even $n-\omega(\log n)$ memory qubits.
A natural question is whether we can further strengthen our conclusion.
Specifically, is there a quantum learning task that can be readily accomplished using purification but is exponentially challenging for any other protocol, regardless of the number of memory qubits utilized, even if unlimited?
However, we will show that such kind of advantages does not exist.

For tasks involving purity estimation and quantum virtual cooling, one can use the generalized SWAP test circuit, as shown in Fig.~\ref{fig:memory}(a).
The circuit comprises $t$ identical copies of $\rho_A$, a control qubit initialized as $\ket{+}$, a controlled $t$-th order permutation gate, and the Pauli-$X$ measurement on the control qubit and observable measurement on one copy of the state.
Then, the measurement results give the target value $\expval{X\otimes O}=\Tr(O\rho_A^t)$ and $\expval{X}=\Tr(\rho_A^t)$.
Essentially, the generalized SWAP test circuit transforms the nonlinear quantity estimation into the linear observable estimation on multiple copies of states.
Given that both the Pauli-$X$ operator and the observable $O$ possessing bounded spectral norms, the sample complexity for purity estimation and quantum virtual cooling are all upper bounded by some constant.

\begin{figure}[htbp]
\centering
\includegraphics[width=0.7\linewidth]{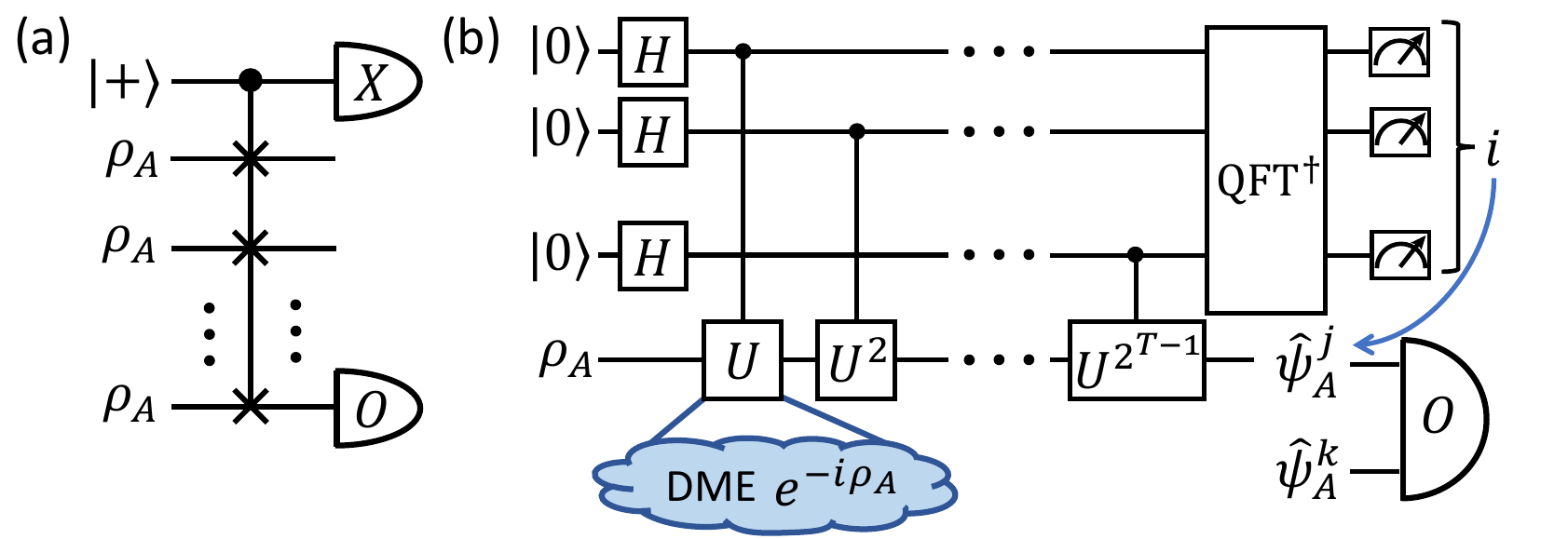}
\caption{Memory-based protocols to achieve tasks discussed in this work. (a) Utilizing the generalized SWAP test algorithm, one can accurately and efficiently estimate $\Tr(\rho_A^t)$ and $\Tr(O\rho_A^t)$. The gate depicted in this figure is the controlled $t$-th order permutation gate. (b) Combining the density matrix exponentiation (DME) and quantum phase estimation, one can prepare the eigenstates of $\rho_A$. Then, by performing joint measurements on multiple copies of eigenstates, one can estimate $\Tr(O\psi_A^0)$ and quantum Fisher information. QFT stands for the quantum Fourier transformation. }
\label{fig:memory}
\end{figure}

For tasks involving quantum principal component analysis and quantum Fisher information estimation, one can use the circuit shown in Fig.~\ref{fig:memory}(b), which is the combination of density matrix exponentiation~\cite{Lloyd2014qpca,Kimmel2017hamiltonian,kjaergaard2022DME} and quantum phase estimation algorithms.
Specifically, the quantum phase estimation algorithm is capable of estimating the eigenvalues of a target unitary and preparing the corresponding eigenstates, and the density matrix exponentiation algorithm can realize the unitary evolution $e^{-i\rho_A}$ for the given unknown state $\rho_A$.
By integrating these two algorithms, one can approximately estimate the eigenvalues $\hat{\lambda}_A^j$ according to the measurement results of control qubits and approximately prepare the eigenstates $\hat{\psi}_A^j$ of $\rho_A$, as shown in Fig.~\ref{fig:memory}(b).
To estimate $\Tr(O\psi_A^0)$, one can perform observable estimation on the output state and post-select the measurement result according to the measurement outcomes of control qubits.
To estimate $\abs{\bra{\psi_A^j}O\ket{\psi_B^k}}^2$, one can use another register to store the output state of the previous experiment and perform the SWAP test with the newly obtained state.
Subsequently, using the estimated eigenvalues and the values of $\abs{\bra{\psi_A^j}O\ket{\psi_B^k}}^2$, one can classically compute the quantum Fisher information.

This protocol can estimate $\Tr(O\psi_A^0)$ and $F_O(\rho_A)$ with constant sample complexity and $\mathcal{O}(n)$ ancillary qubits for the following reasons.
Firstly, the density matrix exponentiation algorithm achieves an $\epsilon$ accuracy in diamond distance for simulating the evolution of $e^{-i\rho_A}$ uses $\mathcal{O}(\frac{1}{\epsilon^2})$ copies of $\rho_A$, which does not result in exponential sample complexity.
Secondly, in this work, the state $\rho_A$ considered has only a constant number of nonzero eigenvalues, thus a constant number of control qubits and runs of density matrix exponentiation subroutine suffice.
In addition, the probability of getting a particular eigenstate using quantum phase estimation depends on the fidelity between the input state and the target eigenstate.
In the task of estimating $\Tr(O\psi_A^0)$, we require that the maximal eigenvalue of $\rho_A$ be $\Theta(1)$, implying the fidelity between $\rho_A$ and $\psi_A^0$ is $\Theta(1)$.
Therefore, by initializing the input state to be the target mixed state $\rho_A$, the post-selection probability is $\Theta(1)$, which does not contribute an exponential factor to the sample complexity.
In the task of quantum Fisher information estimation, we further require that all non-zero eigenvalues be $\Theta(1)$.
Thus, by similarly initializing the input state to be $\rho_A$, the probability of getting any eigenstate with nonzero eigenvalue is $\Theta(1)$, which again does not introduce an exponential factor into the sample complexity.

It is worth noting that this protocol serves as a general method for estimating functions of the eigenvalues and eigenstates of $\rho_A$, owing to its ability to prepare eigenstates and estimate eigenvalues.
Therefore, it eliminates many other scenarios where purification might offer advantages over unbounded quantum memory.
Note that Ref.~\cite{chen2024localtestunitarilyinvariant} provides a rigorous proof that, without the limitation of quantum memory, the learning capability of purification is the same as the original target state, which coincides with our observations.

\section{Summary and Outlook}\label{sec:outlook}

In this work, we explored purification as a valuable quantum resource, highlighting its potential in quantum learning. 
We identified exponential sample complexity separations between learning tasks performed with and without purification in four state learning scenarios. 
In addition, we demonstrated applications in quantum cryptography and discussed the power and underlying source of the advantages that purification provides.
To derive our main results, we proved the sample complexity lower bound for certain state learning tasks, even when constraining the rank of the state.
In addition, for the first time, we established the sample complexity lower bound for estimating quantum fisher information, which is central to quantum metrology.

In addition to scenarios mentioned in the Introduction and quantum cryptography, purification is also practically meaningful in many other physical scenarios. Our results indicate that if we have access to the collective state of the target quantum system and the environment, there are various protocols for estimating the properties of the target quantum system that are robust against the noise caused by the environment.
This observation might be useful for platforms where the environment (noise source) is clear and can be measured to a certain degree, such as the readout cavity and Purcell filter for circuit quantum electrodynamics~\cite{blais2021cqed}. 
A lot of open quantum system simulation is done by explicitly simulating the full quantum system including the target system and the environment~\cite{wang2011open,su2020open,cattaneo2021collision}. Hence, our protocols allow for possible faster subroutines for extracting properties of the target system in these simulations.
Throughout this work, we hold the assumption that the purification state $\Psi_{AB}$ only contains $n+\mathcal{O}(1)$ qubits to ensure that the purification-assisted protocols have constant sample complexity.
According to our sample complexity analysis of memory-assisted protocols, purification would show sample complexity advantage over quantum memory in a large range of qubit numbers.
For example, since the sample complexity for single-copy estimation of state purity is $\Theta(2^{n/2})$~\cite{chen2022memory}, purification will have advantages over quantum memory as long as the total qubit number of $\Psi_{AB}$ is less than $\frac{3}{2}n$.
It will be interesting to figure out the optimal sample complexity for the other tasks with the assistance of purification.

While our results showcase the significant potential of purification, they also underscore the challenges associated with preparing the purification of even a low-rank density matrix. 
If one could easily prepare this purification without doubling the system, learning the state properties would become efficient with such a small memory overhead, thus contradicting our established hardness results. 
Therefore, while purification can aid learning tasks when accessible, purifying a given mixed state presents considerable difficulties.
There exist many quantum algorithms requiring the oracle to prepare the purification of the target mixed state, such as block encoding~\cite{wang2024entropy,liu2024estimatingtracequantumstate} and quantum purified query access model~\cite{belovs2019quantum,gilyen2019distributional,gur2021sublinear}.
Our results indicate that, it is impossible to replace the purification oracle in these algorithms without at least doubling the system size.
Furthermore, from the perspective of open quantum dynamics, the decoherence of a quantum state is caused by its interaction with the environment. 
Thus, our results indicate that even the leakage of a single qubit of quantum information to the environment can pose significant challenges in quantum learning.
In addition, it is a valuable next step to find further applications in scenarios where purification naturally exists, beyond the cryptography that have been discussed.
For example, quantum many-body physics~\cite{Kokail2021entanglement}, quantum gravity~\cite{Kim2020gravity}, and quantum query models~\cite{Low2019hamiltonian} present interesting avenues for future research.

Another important direction is further exploring the power of purification in a wider range of quantum learning tasks. 
In Appendix~\ref{sec:channel_learning}, we have achieved constant sample complexity with purification in several \emph{channel} learning tasks, including the unitarity estimation, virtual channel distillation, and channel PCA.
However, the sample complexity lower bounds without purification, and the resulting exponential separation between learning with and without channel purification remains to be established. 
Another related potential direction is to explore the application of purification in Hamiltonian learning~\cite{huang2023Hamiltonian,Li2024hamiltonian,Yu2023hamiltonian,bluhm2024hamiltonian,Dutkiewicz2024hamiltonian}.
Additionally, finding new applications of purification in quantum state learning, such as estimating partial transposition moments $\Tr[(\rho_A^{\mathrm{T}_{A_1}})^t]$ with $\mathrm{T}_{A_1}$ being the partial transposition over a subsystem of $A$, is crucial, as these moments play vital roles in mixed-state entanglement detection~\cite{elben2020mixed,zhou2020single,liu2021negativity}. 
Furthermore, an intriguing question is whether we can fully characterize all functions of the target mixed state $f(\rho_A)$ that can be estimated using purification with constant sample complexity, thereby gaining a comprehensive understanding of the power of purification.

In recent years, significant advancements in quantum learning have led to the development of numerous novel protocols. 
While we primarily utilize tomography to illustrate the advantage of purification in Sec.~\ref{sec:upper_bound}, it would be valuable to integrate other quantum learning protocols, such as randomized measurements~\cite{Huang2020predicting,elben2023randomized,rovvie2025triple}, with purification and explore applications in broader contexts. 
Meanwhile, exploring other resources, in addition to quantum memory, purification, and conjugate states~\cite{robbie2024conjugate}, that would benefit quantum learning tasks is also a meaningful future direction.

\section*{Acknowledgements}
We thank Hong-Ye Hu and Andreas Elben for illuminating suggestions and comments throughout this research. 
We acknowledge the valuable discussion with Senrui Chen, Jens Eisert, Chao Yin, Weikang Li, Bartek Czech, Yunchao Liu, and Qi Ye. 
ZL and ZD acknowledge the support from the National Natural Science Foundation of China Grant No.~12174216 and the Innovation Program for Quantum Science and Technology Grant No.~2021ZD0300804 and No.~2021ZD0300702. ZC acknowledges support from the EPSRC QCS Hub EP/T001062/1, EPSRC projects Robust and Reliable Quantum Computing (RoaRQ, EP/W032635/1), Software Enabling Early Quantum Advantage (SEEQA, EP/Y004655/1) and the Junior Research Fellowship from St John’s College, Oxford.

\appendix

\section{Additional Lemmas}
In this section, we provide some lemmas that are necessary for our subsequent proofs.

\begin{lemma}\label{lemma:pca}
Given a positive matrix $M$, with the difference between the largest and the second largest eigenvalues equaling $\Delta$, another Hermitian matrix $M^\prime$ with $\norm{M-M^\prime}_{\mathrm{tr}}\le\epsilon$, their largest eigenvalues $\lambda$ and $\lambda^\prime$ satisfy $\abs{\lambda-\lambda^\prime}\le\epsilon\le\frac{\Delta}{2}$.
Given their eigenstates corresponding to their largest eigenvalues, $\ket{\psi}$ and $\ket{\psi^\prime}$, we have $\norm{\ketbra{\psi}{\psi}-\ketbra{\psi^\prime}{\psi^\prime}}\le2\sqrt{\frac{2\epsilon}{\Delta}}$.
\end{lemma}
\begin{proof}
Labeling $M^\prime=M+c$, it is easy to prove that $\lambda^\prime=\bra{\psi^\prime}(M+c)\ket{\psi^\prime}=\bra{\psi^\prime}M\ket{\psi^\prime}+\bra{\psi^\prime}c\ket{\psi^\prime}\le\lambda+\epsilon$ due to the fact that $\norm{\epsilon}_{\mathrm{tr}}\le\epsilon$. 
Besides, as $\bra{\psi}(M+c)\ket{\psi}\ge\lambda-\epsilon$, we also have $\lambda^\prime\ge\lambda-\epsilon$, which finalizes the proof of the eigenvalues.

Based on this result, we have 
\begin{equation}\label{eq:fid1}
\bra{\psi^\prime}M\ket{\psi^\prime}=\bra{\psi^\prime}M^\prime-c\ket{\psi^\prime}\ge\lambda-2\epsilon.
\end{equation}
Note that $\bra{\psi^\prime}M\ket{\psi^\prime}=\lambda\abs{\braket{\psi^\prime}{\psi}}^2+\sum_j\lambda_j\abs{\braket{\psi^\prime}{\psi_j}}^2$, where $\lambda_j$ and $\ket{\psi_j}$ denote eigenvalues and eigenstates of $M$ except for the principal component.
Given the spectral gap of $M$, labeled by $\Delta$, and the normalization condition $\abs{\braket{\psi^\prime}{\psi}}^2+\sum_j\abs{\braket{\psi^\prime}{\psi_j}}^2=1$, we have
\begin{equation}\label{eq:fid2}
\begin{aligned}
\bra{\psi^\prime}M\ket{\psi^\prime}\le&\lambda\abs{\braket{\psi^\prime}{\psi}}^2+(\lambda-\Delta)\sum_j\abs{\braket{\psi^\prime}{\psi_j}}^2\\
=&\lambda-\Delta+\Delta\abs{\braket{\psi^\prime}{\psi}}^2.
\end{aligned}
\end{equation}
Combining Eq.~\eqref{eq:fid1} and Eq.~\eqref{eq:fid2}, we have $\abs{\braket{\psi^\prime}{\psi}}^2\ge1-\frac{2\epsilon}{\Delta}$.
Using the relationship between fidelity and trace distance, we can upper bound the trace distance between $\psi$ and $\psi^\prime$ as $\norm{\psi-\psi^\prime}_{\mathrm{tr}}\le 2\sqrt{\frac{2\epsilon}{\Delta}}$.
\end{proof}

\begin{corollary}\label{cor:fisher}
Given a $d$-dimensional positive non-degenerate matrix $M$ which has non-zero eigenvalues $\{\lambda_j\}_{j=0}^{r-1}$ with $\lambda_j>\lambda_{j+1}$, minimal spectral distance $\Delta_{\mathrm{min}}=\min\left\{\min_{j=0}^{r-2}\abs{\lambda_j-\lambda_{j+1}},\lambda_{r-1}\right\}$, and eigenstates $\{\psi_j\}_{i=0}^{r-1}$, a $d$-dimensional Hermitian matrix $M^\prime=M+c$ with $\norm{c}_{\mathrm{tr}}\le\epsilon$, which has eigenvalues $\{\lambda_j^\prime\}_{j=0}^{d-1}$ and eigenstates $\{\psi_j^\prime\}_{j=0}^{d-1}$, it satisfies $\abs{\lambda_j-\lambda_j^\prime}\le\epsilon$ and $\norm{\psi_j-\psi_j^\prime}\le\frac{16}{\Delta_\mathrm{min}}\epsilon^{\frac{1}{2r}}$ for all $j\le r-1$ when $\epsilon<\frac{\Delta_\mathrm{min}^{4r}}{(8\sqrt{2})^{2r}}$.
\end{corollary}
\begin{proof}
We first prove the conclusion for eigenvalues.
By definition, we have
\begin{equation}
\begin{aligned}
\lambda_j=&\min_{\mathcal{H}_{d-j}}\max_{\ket{\psi}\in\mathcal{H}_{d-j}}\bra{\psi}M\ket{\psi}\\
\le&\min_{\mathcal{H}_{d-j}}\left(\max_{\ket{\psi}\in\mathcal{H}_{d-j}}\bra{\psi}M^\prime\ket{\psi}+\max_{\ket{\psi}\in\mathcal{H}_{d-j}}\bra{\psi}c\ket{\psi}\right)\\
\le& \min_{\mathcal{H}_{d-j}}\max_{\ket{\psi}\in\mathcal{H}_{d-j}}\bra{\psi}M^\prime\ket{\psi}+\epsilon=\lambda_j^\prime+\epsilon,
\end{aligned}
\end{equation}
where $d$ is the dimension of $M$ and $M^\prime$ and $\mathcal{H}_{d-j}$ is a $d-j$-dimensional Hilbert space.
Similarly, we can also prove that $\lambda_j^\prime\le\lambda_j+\epsilon$.
Thus, we get the conclusion that $\abs{\lambda_j-\lambda_j^\prime}\le\epsilon$ for all $j$.

Now we try to prove the conclusion for eigenstates.
According to Lemma~\ref{lemma:pca}, the distance of principal components is bounded by $\norm{\psi_0-\psi_0^\prime}_\mathrm{tr}\le2\sqrt{\frac{2\epsilon}{\Delta_{\mathrm{min}}}}$.
Substituting the conclusion of eigenvalues, we have $\norm{\lambda_0\psi_0-\lambda_0^\prime\psi_0^\prime}_\mathrm{tr}\le\lambda_0\norm{\psi_0-\psi_0^\prime}_\mathrm{tr}+\abs{\lambda_0-\lambda_0^\prime}\norm{\psi_0^\prime}_\mathrm{tr}\le2\sqrt{\frac{2\epsilon}{\Delta_{\mathrm{min}}}}+\epsilon\le3\sqrt{\frac{2\epsilon}{\Delta_{\mathrm{min}}}}$, where we use the fact that $2\epsilon<\Delta_\mathrm{min}<1$.
Define new matrix $M_1=M-\lambda_0\psi_0$ and $M_1^\prime=M^\prime-\lambda_0^\prime\psi_0^\prime$, we have
\begin{equation}
\norm{M_1-M_1^\prime}_\mathrm{tr}\le\norm{M-M^\prime}_\mathrm{tr}+\norm{\lambda_0\psi_0-\lambda_0^\prime\psi_0^\prime}_\mathrm{tr}\le4\sqrt{\frac{2\epsilon}{\Delta_{\mathrm{min}}}}.
\end{equation}
As $\ket{\psi_1}$ and $\ket{\psi_1^\prime}$ are principal eigenstates of $M_1$ and $M_1^\prime$, we can similarly employ the conclusion made in Lemma~\ref{lemma:pca} to prove that
\begin{equation}
\norm{\psi_1-\psi_1^\prime}_\mathrm{tr}\le2\sqrt{\frac{4}{\Delta_\mathrm{min}}\sqrt{\frac{2\epsilon}{\Delta_\mathrm{min}}}}.
\end{equation}
Following the similar rule, we have $\norm{\psi_{i}-\psi_{i}^\prime}_\mathrm{tr}\le\norm{\psi_{r-1}-\psi_{r-1}^\prime}_\mathrm{tr}\le2\sqrt{2}\left(\frac{4\sqrt{2}}{\Delta_{\mathrm{min}}}\right)^{1-\frac{1}{2r}}\epsilon^{\frac{1}{2r}}<\frac{16}{\Delta_\mathrm{min}}\epsilon^{\frac{1}{2r}}$.

\end{proof}

\begin{lemma}\label{lemma:fisher2}
Given $\norm{\ketbra{\psi_1}{\psi_1}-\ketbra{\psi_1^\prime}{\psi_1^\prime}}_\mathrm{tr}\le\epsilon_1$ and $\norm{\ketbra{\psi_2}{\psi_2}-\ketbra{\psi_2^\prime}{\psi_2^\prime}}_\mathrm{tr}\le\epsilon_2$, we have $\norm{P_{12}-P_{12}^\prime}_\mathrm{tr}\le 2(\epsilon_1+\epsilon_2)$, where $P_{12}=\ketbra{\psi_1}{\psi_2}\otimes \ketbra{\psi_2}{\psi_1}+\ketbra{\psi_2}{\psi_1}\otimes \ketbra{\psi_1}{\psi_2}$ and $P_{12}^\prime=\ketbra{\psi_1^\prime}{\psi_2^\prime}\otimes\ketbra{\psi_2^\prime}{\psi_1^\prime}+\ketbra{\psi_2^\prime}{\psi_1^\prime}\otimes\ketbra{\psi_1^\prime}{\psi_2^\prime}$.
\end{lemma}
\begin{proof}

According to the triangle inequality, we know that 
\begin{equation}
\begin{aligned}
&\norm{P_{12}-P_{12}^\prime}\\
\le&2\norm{\ketbra{\psi_1}{\psi_2}\otimes\ketbra{\psi_2}{\psi_1}-\ketbra{\psi_1^\prime}{\psi_2^\prime}\otimes\ketbra{\psi_2^\prime}{\psi_1^\prime}}_\mathrm{tr}\\
=&2\norm{\left(\ketbra{\psi_1}{\psi_1}\otimes\ketbra{\psi_2}{\psi_2}-\ketbra{\psi_1^\prime}{\psi_1^\prime}\otimes\ketbra{\psi_2^\prime}{\psi_2^\prime}\right)\mathrm{SWAP}}_\mathrm{tr}\\
=&2\norm{\ketbra{\psi_1}{\psi_1}\otimes\ketbra{\psi_2}{\psi_2}-\ketbra{\psi_1^\prime}{\psi_1^\prime}\otimes\ketbra{\psi_2^\prime}{\psi_2^\prime}}_\mathrm{tr}\\
\le&2\norm{\ketbra{\psi_1}{\psi_1}\otimes(\ketbra{\psi_2}{\psi_2}-\ketbra{\psi_2^\prime}{\psi_2^\prime})}_{\mathrm{tr}}\\
&+2\norm{(\ketbra{\psi_1^\prime}{\psi_1^\prime}-\ketbra{\psi_1}{\psi_1})\otimes\ketbra{\psi_2^\prime}{\psi_2^\prime}}_\mathrm{tr}\\
=&2(\epsilon_1+\epsilon_2),
\end{aligned}
\end{equation}
which concludes our proof.

\end{proof}

\begin{lemma}[{\cite[Lemma 2]{Huang_2022_quantum}}]\label{lem:high_moment_haar}
For any $m$-qubit pure states $\ket{\phi_1},\ket{\phi_2},\cdots,\ket{\phi_K}$ and an $m$-qubit Haar random state $\ket{\psi}$, we have
\begin{equation}
    \bE_{\ket{\psi}} \prod_{j=1}^K\abs{\langle \psi \ket{\phi_j}}^2 \ge \frac{1}{(2^m + K - 1)\cdots(2^m+1)2^m}
\end{equation}
\end{lemma}

\begin{lemma}\label{lem:eigenvalue}
    For two state $\ket{u}, \ket{v}$, the matrix $\alpha \ketbra{u} + \beta \ketbra{v}$ has at most two eigenvalues:
    \begin{equation}
        \frac{\alpha + \beta}{2} \pm \frac{1}{2} \sqrt{(\alpha - \beta)^2 + 4 \alpha\beta \abs{\bra{u}\ket{v}}^2}
    \end{equation}
\end{lemma}

\begin{lemma}\label{lem:innerproduct}
    For two $m$-qubit Haar random state $\ket{u},\ket{v}$, their inner product $a=\abs{\bra{u}\ket{v}}^2$ satisfies:
    \begin{equation}
        \Pr(a \ge 2^{-m}+\epsilon ) \le 2 \exp(-\frac{\epsilon^2 d}{100\pi^2})
    \end{equation}
\end{lemma}
\begin{proof}
    For an $m$-qubit unitary $U$, define $f(U) = \Tr(\ketbra{v} U \ketbra{0}U^{\dagger})$. For a given $\ket{v}$, $f(U)$ has the same distribution as $a$. By utilizing the fact that $\bE_{U} U \ketbra{0} U^{\dagger} = \frac{I}{d}$, the mean value of $f(U)$ is $\bE f(U) = \frac{1}{d}$, where $d = 2^m$ is the dimension of the system.

    By Levy's lemma, we have {\cite[Theorem 7.37]{watrous2018theory}}:
    \begin{equation}\label{eq:levy_inner}
        \Pr(f(U) - \bE f(U) \ge \epsilon) \le 2 \exp(-\frac{\epsilon^2 d}{25\pi \kappa^2}),
    \end{equation}
    where $\kappa$ is the Lipshitz constant  of the function $f(U)$. To analyze the Lipschitz constant of the function $f$, consider two unitary $U$ and $V$, we have:
    \begin{equation}
    \begin{split}
        &\abs{f(U) - f(V)} \\
        &= \big{|}\Tr( (\ketbra{v} \otimes \ketbra{0}) (U \otimes U^{\dagger}) S)\\
        & \ \ - \Tr( (\ketbra{v} \otimes \ketbra{0}) (V \otimes V^{\dagger}) S)\big{|} \\
        &= \abs{\Tr( (\ketbra{v} \otimes \ketbra{0}) (U \otimes U^{\dagger} - V\otimes V^{\dagger}) S))}\\
        & \le \norm{\ketbra{v} \otimes \ketbra{0}}_1 \norm{(U \otimes U^{\dagger} - V\otimes V^{\dagger}) S}_{\infty}\\
        &= \norm{(U \otimes U^{\dagger} - V\otimes V^{\dagger})}_{\infty}\\
        & \le \norm{(U \otimes U^{\dagger} - V\otimes U^{\dagger})}_{\infty} + \norm{(V \otimes U^{\dagger} - V\otimes V^{\dagger})}_{\infty}\\
        & = \norm{U - V}_{\infty} + \norm{ U^{\dagger} - V^{\dagger}}_{\infty} \\
        & \le 2 \norm{U-V}_2
    \end{split}
    \end{equation}
    That is, $\kappa = 2$. Substituting $\kappa=2$ into Eq.~\eqref{eq:levy_inner} leads to the result.
\end{proof}

\begin{corollary}
    Given constant $\epsilon = \mathcal{O}(1)$ and $0 < \delta < 1$, for sufficiently large $m$, we have that for two $m$-qubit Haar random states $\ket{u}$ and $\ket{v}$,
    \begin{equation}
        \Pr(\abs{\bra{u}\ket{v}}^2 \ge \epsilon) \le \delta.
    \end{equation}
\end{corollary}

\begin{corollary}\label{col:innerproduct}
    Given two $m$-qubit Haar random states $\ket{u}$ and $\ket{v}$, let $M = \frac{3}{8} \ketbra{u} + \frac{1}{8} \ketbra{v}$ and denote the largest eigenvalue of $M$ as $\lambda_1$ and the second largest eigenvalue of $M$ as $\lambda_2$, with corresponding eigenvector $\ket{\phi_1}, \ket{\phi_2}$. Given constant $\epsilon = \mathcal{O}(1)$ and $0 < \delta < 1$, for sufficiently large $m$, with probability at least $1-\delta$, we have
    \begin{equation}
    \begin{split}
        \abs{\bra{u}\ket{\phi_1}}^2 &\ge 1-\epsilon, \\
        \abs{\bra{u}\ket{\phi_2}}^2 &\le \epsilon, \\
        \lambda_1 &\le \frac{3}{8} + \frac{\epsilon}{2}, \\
        \lambda_2 &\ge \frac{1}{8} - \frac{\epsilon}{2}.
    \end{split}
    \end{equation}
\end{corollary}

\begin{proof}
    Given $\epsilon < 0.1$ and $0 < \delta < 1$, denote $\epsilon_1 = \frac{8\epsilon}{3}$. By Corollary \ref{col:innerproduct}, we have
    \begin{equation}
        \Pr(\abs{\bra{u}\ket{v}}^2 \ge \epsilon_1) \le \delta
    \end{equation}
    for sufficiently large $m$. Then, by Lemma \ref{lem:eigenvalue}, with probability $1-\delta$, the eigenvalues satisfy:
    \begin{equation}\label{eq:col_eigenvalue}
        \begin{split}
            &\lambda_1 = \frac{1}{4} + \frac{1}{2}\sqrt{1/16 + 3\epsilon_1/16} \le \frac{3}{8} + \frac{3\epsilon_1}{16}, \\
            &\lambda_2 = \frac{1}{2} - \lambda_1 \ge \frac{1}{8} -  \frac{3\epsilon_1}{16}, \\
        \end{split}
    \end{equation}
    Moreover, we have:
    \begin{equation}\label{eq:col_expectation1}
        \bra{u}M\ket{u} = \frac{3}{8} + \frac{1}{8} \abs{\bra{u}\ket{v}}^2 \ge \frac{3}{8}
    \end{equation}
    Suppose $\abs{\bra{u}\ket{\phi_1}}^2 = 1-\epsilon_2$, then $\abs{\bra{u}\ket{\phi_2}}^2 = \epsilon_2$, and the expectation value can be written as:
    \begin{equation}\label{eq:col_expectation2}
        \bra{u}M\ket{u} = \lambda_1 (1-\epsilon_2) + \lambda_2 \epsilon_2
    \end{equation}
    Combining Eq.~\eqref{eq:col_eigenvalue}, Eq.~\eqref{eq:col_expectation1} and Eq.~\eqref{eq:col_expectation2}, with probability $1-\delta$, we have:
    \begin{equation}
    \begin{split}
        \lambda_1 (1-\epsilon_2) + \lambda_2 \epsilon_2  &\ge \frac{3}{8}, \\
        \lambda_1 - \frac{3}{8} &\ge \epsilon_2 (\lambda_1 - \lambda_2), \\
        \frac{3\epsilon_1}{16(\lambda_1 - \lambda_2)} &\ge \epsilon_2.
    \end{split}   
    \end{equation}
    Here, $\lambda_1 - \lambda_2 \le \frac{1}{4} + \frac{3\epsilon_1}{8} \le \frac{1}{2}$ and thus $\epsilon_2 \le \frac{3\epsilon_1}{8} \le \epsilon$. That is, with probability $1-\delta$, 
    \begin{equation}
        \abs{\bra{u}\ket{\phi_1}}^2  = 1-\epsilon_2 \ge 1-\epsilon.
    \end{equation}
\end{proof}

\begin{corollary}\label{col:fisher}
    Consider the state ensemble $\mathcal{S}_1, \mathcal{S}_2$ defined in Eq.~\eqref{eq:fisher_info_SA}. Given $0 < \delta < 1$, for the observable $X_1 = X \otimes I^{\otimes n-1}$, for sufficiently large $n$, if we randomly choosen $\rho \in \mathcal{S}_1$, then with probability at least $1-\delta$, $F_{X_1} \ge 0.01$. 
\end{corollary}
\begin{proof}
    For randomly selected $\rho_A \in \mathcal{S}_1$, the three eigenvector is $\frac{1}{2} \ge \lambda_1 \ge \lambda_2$, and denote the corresponding eigenstates as $\ket{0} \otimes \ket{u}, \ket{1} \otimes \ket{\phi_1}$ and $\ket{1} \otimes \ket{\phi_2}$. Then, the Fisher information can be calculated as:
    \begin{equation}
    \begin{split}
        F_{X_1} &= 2[\frac{(\frac{1}{2}-\lambda_1)^2}{\frac{1}{2}+\lambda_1} \abs{\bra{0}X\ket{1}}^2 \abs{\bra{u}\ket{\phi_1}}^2\\
        &+ \frac{(\frac{1}{2}-\lambda_2)^2}{\frac{1}{2}+\lambda_2} \abs{\bra{0}X\ket{1}}^2 \abs{\bra{u}\ket{\phi_2}}^2] \\
        &\ge \frac{(\frac{1}{2}-\lambda_1)^2}{\frac{1}{2}+\lambda_1} \abs{\bra{u}\ket{\phi_1}}^2
    \end{split}
    \end{equation}
    By Corollary \ref{col:innerproduct}, for sufficiently small $\epsilon$  and sufficient large $n$, we have
    \begin{equation}
        F_{X_1} \ge \frac{(\frac{1}{2} - \frac{3}{8}-\epsilon)^2}{\frac{1}{2} + \frac{3}{8} +\epsilon} (1-\epsilon)^2 \ge 0.01
    \end{equation}
    with probability at least $1-\delta$.
\end{proof}

\section{Additional Proofs on the Sample Complexity of Purification-assisted Learning Tasks} \label{app:proof_upper}

\subsection{State Moment Estimation}\label{app:purity_upper}
It is known that, with sample complexity $\mathcal{O}(\frac{d^3}{\epsilon^2})$, one can generate the estimator $\hat{\rho}$ of a $d$-dimensional state $\rho$, with trace distance $\norm{\hat{\rho}-\rho}_\mathrm{tr}\le\epsilon<1$ \cite{KUENG2017tomo}.
Thus, we can rewrite $\hat{\rho}=\rho+c$ with $\norm{c}_\mathrm{tr},\norm{c}_{\infty}\le\epsilon$.
Here, $\norm{\cdot}_\infty$ denotes the maximal absolute eigenvalue of the target matrix.
Then, we have
\begin{equation}
\norm{\hat{\rho}^t-\rho^t}_\mathrm{tr}=\norm{(\rho+c)^t-\rho^t}_\mathrm{tr}=\norm{\sum_{j=1}^t C_t}_\mathrm{tr}.
\end{equation}
where $C_t = \sum_j M_j^{(t)}$ is the sum of all matrices with $t$ terms of $c$ and $(b-t)$ terms of $\rho$.  Using the fact that $\norm{AB}_\mathrm{tr}\le\norm{A}_\infty\norm{B}_\mathrm{tr}$ and $\norm{\rho}_\infty,\norm{\rho}_\mathrm{tr}\le 1$, for arbitrary $M_j^{(t)}$, we have $\norm{M_j^{(t)}}_\mathrm{tr} \le \epsilon^t$. Then, 
\begin{equation}
    \norm{C_t}_\mathrm{tr} \le \binom{n}{t} \epsilon^t,
\end{equation}
and for $\epsilon < \frac{1}{t}$, we have
\begin{equation}     
    \abs{\Tr(\hat{\rho}^t)-\Tr(\rho^t)} \le \norm{\hat{\rho}^t-\rho^t}_\mathrm{tr} \le \sum_{t=1}^n \norm{C_t}_\mathrm{tr} \le (1+\epsilon)^t - 1 \le 2\epsilon t.
\end{equation}

Therefore, by setting the sample complexity to be $\mathcal{O}(\frac{t^2d^3}{\epsilon^2})$, one can estimate $\Tr(\rho^t)$ to $\epsilon$ accuracy.
Considering that $t$ and the dimension of $\rho_B$ are all constants, we can estimate $\Tr(\rho_A^t)$ to $\epsilon$ accuracy with a constant number of copies of the purification $\Psi_{AB}$.

\subsection{Quantum Virtual Cooling}\label{app:cooling_upper}
We first prove Observation~\ref{obs:cooling}.
Using the Schmidt decomposition,  $\Psi_{AB}=\sum_{j,k}\sqrt{\lambda_j\lambda_k}\ketbra{\psi_A^j,\psi_B^j}{\psi_A^k,\psi_B^k}$, $\rho_A=\sum_j\lambda_A^j\ketbra{\psi_A^j}{\psi_A^j}$, and $\rho_B=\sum_j\lambda_B^j\ketbra{\psi_B^j}{\psi_B^j}$, we have
\begin{equation}
\begin{aligned}
&\Tr_B[\Psi_{AB}(\mathbb{I}_A\otimes\rho_B^{t-1})]\\
=&\sum_{j,k,l}\sqrt{\lambda_j\lambda_k}\lambda_l^{t-1}\Tr_B[\ketbra{\psi_A^j,\psi_B^j}{\psi_A^k,\psi_B^k}(\mathbb{I}_A\otimes\ketbra{\psi_B^l}{\psi_B^l})]\\
=&\sum_{j,k,l}\sqrt{\lambda_j\lambda_k}\lambda_l^{t-1}\ketbra{\psi_A^j}{\psi_A^k}\braket{\psi_B^l}{\psi_B^j}\braket{\psi_B^k}{\psi_B^l}\\
=&\sum_{j,k,l}\sqrt{\lambda_j\lambda_k}\lambda_l^{t-1}\ketbra{\psi_A^j}{\psi_A^k}\delta_{j,l}\delta_{k,l}\\
=&\sum_l\lambda_l^t\ketbra{\psi_A^l}{\psi_A^l}=\rho_A^t,
\end{aligned}
\end{equation}
which finishes the proof. 

Based on Observation~\ref{obs:cooling}, we can rewrite $\Tr(O\rho_A^t)=\Tr[\Psi_{AB}(O\otimes\rho_B^{t-1})]$.
Thus, a straightforward protocol to estimate $\Tr(O\rho_A^t)$ is first reconstructing the estimator $\hat{\rho}_B$ for $\rho_B$ and then estimating the expectation value of observable $O\otimes\hat{\rho}_B^{t-1}$ on state $\Psi_{AB}$.
Denoting $\hat{M}$ to be the estimator of $\Tr[\Psi_{AB}(O\otimes\hat{\rho}_B^{t-1})]$, we have
\begin{equation}
\begin{aligned}
&\abs{\Tr[\Psi_{AB}(O\otimes\rho_B^{t-1})]-\hat{M}}\\
\le&\abs{\Tr[\Psi_{AB}(O\otimes\rho_B^{t-1})]-\Tr[\Psi_{AB}(O\otimes\hat{\rho}_B^{t-1})]}\\
&+\abs{\Tr[\Psi_{AB}(O\otimes\hat{\rho}_B^{t-1})]-\hat{M}}.
\end{aligned}
\end{equation}
Thus, if we can make sure of $\abs{\Tr[\Psi_{AB}(O\otimes\rho_B^{t-1})]-\Tr[\Psi_{AB}(O\otimes\hat{\rho}_B^{t-1})]}\le\epsilon_1$ and $\abs{\Tr[\Psi_{AB}(O\otimes\hat{\rho}_B^{t-1})]-\hat{M}}\le\epsilon_2$ with $0\le\epsilon_1,\epsilon_2$ and $\epsilon_1+\epsilon_2=\epsilon$, we can estimate $\Tr(O\rho_A^t)$ to $\epsilon$ accuracy.
We will show that these two conditions can be easily satisfied with constant sample complexity.

Firstly, as shown in the previous section, with constant sample complexity, one can make sure that $\norm{\rho_B^{t-1}-\hat{\rho}_B^{t-1}}_\infty\le\frac{\epsilon_1}{\norm{O}_\infty}$ and thus $\norm{O\otimes(\rho_B^{t-1}-\hat{\rho}_B^{t-1})}_\infty\le\epsilon_1$ using single-copy state tomography.
Therefore, with constant sample complexity, we can make sure $\abs{\Tr\{\Psi_{AB}[O\otimes(\rho_B^{t-1}-\hat{\rho}_B^{t-1})]\}}=\abs{\bra{\Psi_{AB}}O\otimes(\rho_B^{t-1}-\hat{\rho}_B^{t-1})\ket{}\Psi_{AB}}\le\norm{O\otimes(\rho_B^{t-1}-\hat{\rho}_B^{t-1})}_\infty\le\epsilon_1$.
After obtaining the classical description of $\hat{\rho}_B^{t-1}$, the estimation of $\hat{M}$ is achieved by normal expectation value estimation protocols.
For example, one can first rotate $\Psi_{AB}$ to the eigenbasis of $O\otimes\hat{\rho}_B^{t-1}$ and then perform computational basis measurements.
Then, the variance is bounded by $\mathrm{Var}(\hat{M})=\Tr[\Psi_{AB}(O^2\otimes\hat{\rho}_B^{2t-2})]-\Tr[\Psi_{AB}(O\otimes\hat{\rho}_B^{t-1})]^2\le\norm{O}_\infty^2\norm{\hat{\rho}_B^{t-1}}_\infty^2$.
According to the result of tomography, we can make sure that $\norm{\hat{\rho}_B^{t-1}}_\infty\le 1+\frac{\epsilon_1}{\norm{O}_\infty}$.
Therefore, one can prove that the variance of $\hat{M}$ is also bounded by some constant.
Subsequently, one can estimate the value of $\Tr[\Psi_{AB}(O\otimes\hat{\rho}_B^{t-1})]$ to $\epsilon_2$ accuracy with constant sample complexity, which summarizes the proof.

In practical scenarios, we are normally required to estimate the value of $\frac{\Tr(O\rho_A^t)}{\Tr(\rho_A^t)}$, where $\Tr(\rho_A^t)$ in the denominator represents the normalization factor.
When considering the case where $\rho_A$ can be purified with a constant number of ancilla qubits, the denominator, $\Tr(\rho_A^t)=\Tr(\rho_B^t)$, is a constant.
Thus, the extra denominator would not cause a big increase in the sample complexity.

\subsection{Quantum Principal Component Analysis}\label{app:qpca_upper}

To prove Observation~\ref{obs:pca}, we first rewrite the Schmidt decomposition as $\ket{\Psi_{AB}}=\sqrt{\lambda_0}\ket{\psi_A^0}\otimes\ket{\psi_B^0}+\sum_{j\neq0}\sqrt{\lambda_j}\ket{\psi_A^j}\otimes\ket{\psi_B^j}$, with $\lambda_0>\lambda_j$ for all $j$, $\lambda_0$ is the largest eigenvalue of $\rho_A$ and $\rho_B$ and $\ket{\psi_A^0}$ and $\ket{\psi_B^0}$ are principal components of $\rho_A$ and $\rho_B$, respectively.
Due to the orthogonality between $\ket{\psi_B^0}$ and $\ket{\psi_B^j}$, we have
\begin{equation}
\begin{aligned}
&\frac{1}{\lambda_0}\bra{\psi_B^0}\Psi_{AB}\ket{\psi_B^0}\\
=&\frac{1}{\lambda_0}\bra{\psi_B}(\lambda_0\ketbra{\psi_A^0}{\psi_A^0}\otimes\ketbra{\psi_B^0}{\psi_B^0})\ketbra{\psi_B^0}=\ketbra{\psi_A^0}{\psi_A^0}.
\end{aligned}
\end{equation}

The proof of Theorem~\ref{thm:pca} is also based on quantum tomography.
Specifically, one can first use state tomography to reconstruct the density matrix $\rho_B$ and classically compute its biggest eigenvalue $\hat{\lambda}_B^0$ and the corresponding eigenstate $\hat{\psi}_B^0$.
Then, one estimates the expectation value of $\frac{1}{\hat{\lambda}_B^0}O\otimes\hat{\psi}_B^0$ on quantum state $\Psi_{AB}$ to get the final estimator $\hat{M}$.
The total statistical error can be written as
\begin{equation}
\begin{aligned}
&\abs{\Tr(O\psi_A^0)-\hat{M}}\\
\le&\abs{\frac{1}{\lambda_B^0}\Tr[\Psi_{AB}(O\otimes\psi_B^0)]-\frac{1}{\hat{\lambda}_B^0}\Tr[\Psi_{AB}(O\otimes\hat{\psi}_B^0)]}\\
+&\abs{\frac{1}{\hat{\lambda}_B^0}\Tr[\Psi_{AB}(O\otimes\hat{\psi}_B^0)]-\hat{M}}.
\end{aligned}
\end{equation}
Similarly to quantum virtual cooling, to make sure that the statistical error of quantum principal component analysis is low, one needs to make sure that the two terms on the right-hand side of the last equation are low enough.

We first analyze the first term.
The natural question concerning it is that, if we can learn $\rho_B$, whose difference between the largest and second largest eigenvalues is a constant, to $\epsilon$ accuracy using tomography, can we bound the difference between $\psi_B^0$ and $\hat{\psi}_B^0$?
Using conclusions made in Lemma~\ref{lemma:pca}, we can upper bound the distance between eigenvalues and eigenstates as $\abs{\lambda_B^0-\hat{\lambda}_B^0}\le\epsilon$ and $\norm{\psi_B^0-\hat{\psi}_B^0}_\mathrm{tr}\le2\sqrt{\frac{2\epsilon}{\Delta}}$.
Using the property of $\frac{1}{1-x}\le 1+2x$ and $\frac{1}{1+x}\le 1-2x$ for all $0\le x\le \frac{1}{2}$, we have $\abs{\frac{1}{\lambda_B^0}-\frac{1}{\hat{\lambda}_B^0}}\le\frac{2\epsilon}{(\lambda_B^0)^2}\le\frac{2\epsilon}{\Delta^2}$.
Therefore, we have
\begin{equation}
\begin{aligned}
&\abs{\frac{1}{\lambda_B^0}\Tr[\Psi_{AB}(O\otimes\psi_B^0)]-\frac{1}{\hat{\lambda}_B^0}\Tr[\Psi_{AB}(O\otimes\hat{\psi}_B^0)]}\\
\le&
\norm{\frac{\psi_B^0}{\lambda_B^0}-\frac{\hat{\psi}_B^0}{\hat{\lambda}_B^0}}_\mathrm{tr}\\
=&\norm{\left(\frac{\psi_B^0}{\lambda_B^0}-\frac{\hat{\psi}_B^0}{\lambda_B^0}\right)+\left(\frac{\hat{\psi}_B^0}{\lambda_B^0}-\frac{\hat{\psi}_B^0}{\hat{\lambda}_B^0}\right)}_\mathrm{tr}\\
\le&\frac{1}{\lambda_B^0}\norm{\psi_B^0-\hat{\psi}_B^0}_{\mathrm{tr}}+\abs{\frac{1}{\lambda_B^0}-\frac{1}{\hat{\lambda}_B^0}}\\
\le&\frac{2}{\Delta}\sqrt{\frac{2\epsilon}{\Delta}}+\frac{2\epsilon}{\Delta^2},
\end{aligned}
\end{equation}
where we use the fact that $\norm{\hat{\psi}_B^0}_\mathrm{tr}=1$.
Thus, if $\norm{\rho_B-\hat{\rho}_B}_{\mathrm{tr}}\le\frac{\Delta^3\epsilon_1^2}{16}$ with $0<\epsilon_1\le 1$, one can make sure that $\abs{\frac{1}{\lambda_B^0}\Tr[\Psi_{AB}(O\otimes\psi_B^0)]-\frac{1}{\hat{\lambda}_B^0}\Tr[\Psi_{AB}(O\otimes\hat{\psi}_B^0)]}\le\epsilon_1$.
To achieve this, the sample complexity of tomography based on single-copy measurements on $\rho_B$ is upper bounded by $\mathcal{O}(\frac{256d_B^3}{\Delta^6\epsilon_1^4})$, which is not related with the qubit number of $\rho_A$ according to the requirements raised in Theorem~\ref{thm:pca}.

It is also easy to show that one can make sure $\abs{\frac{1}{\hat{\lambda}_B^0}\Tr[\Psi_{AB}(O\otimes\hat{\psi}_B^0)]-\hat{M}}\le\epsilon_2$ also with constant sample complexity.
The estimator $\hat{M}_2$ is actually constructed by measuring the observable $\frac{1}{\hat{\lambda}_B^0}O\otimes\hat{\psi}_B^0$ on the quantum state $\Psi_{AB}$, which is similarly achieved by first rotating $\Psi_{AB}$ to the eigenbasis of $O\otimes\hat{\psi}_B^0$ and performing computational basis measurements.
The sample complexity for accurately estimating it also depends on the spectral norm.
As $\frac{\Delta^3\epsilon_1^2}{16}<\frac{\Delta}{2}$ and $\norm{\hat{\psi}_B^0}_{\infty}=1$, we have $\hat{\lambda}_B^0\ge\lambda_B^0-\frac{\Delta^3\epsilon_1^2}{16}\ge\frac{\Delta}{2}$ and $\norm{\frac{1}{\hat{\lambda}_B^0}O\otimes\hat{\psi}_B^0}_\infty\le\frac{2\norm{O}_\infty}{\Delta}$.
Therefore, to make sure $\abs{\frac{1}{\hat{\lambda}_B^0}\Tr[\Psi_{AB}(O\otimes\hat{\psi}_B^0)]-\hat{M}}\le\epsilon_2$, the sample complexity is upper bounded by $\mathcal{O}(\frac{\norm{O}_\infty^2}{\Delta^2\epsilon_2^2})$, which is also irrelevant with the qubit number of $\rho_A$.

\subsection{Quantum Fisher Information}\label{app:fisher_upper}
According to the definition of quantum Fisher information, Eq.~\eqref{eq:new_fisher}, $F_O(\rho_A)$ is the summation over a $\mathcal{O}(1)$ number of terms when the rank of $\rho_A$ is upper bounded by $\mathcal{O}(1)$.
Therefore, if one can accurately estimate each term with $\mathcal{O}(1)$ sample complexity, one can also accurately estimate $F_O(\rho_A)$ with $\mathcal{O}(1)$ sample complexity.
Each term is composed of two factors, the eigenvalue factor $\frac{(\lambda_B^j-\lambda_B^k)^2}{\lambda_B^j\lambda_B^k(\lambda_B^j+\lambda_B^k)}$ and the eigenstate factor $\Tr\left[\Psi_{AB}^{\otimes 2}\left(O^{\otimes 2}\otimes P_{B}^{jk}\right)\right]$.
Due to the assumptions that $\lambda_A^j,\lambda_A^k=\Theta(1)$, $\abs{\lambda_A^j-\lambda_A^k}=\Theta(1)$, and $\norm{O}_\infty=\mathcal{O}(1)$, the values of these two terms can be both bounded by $\mathcal{O}(1)$.
Therefore, if these two terms can be both accurately estimated, the product of them can be accurately estimated.

Similar with all the previous tasks, the protocol to estimate quantum Fisher information is also based on tomography of $\rho_B$.
As the dimension of system $B$ is a constant, with constant sample complexity, one can accurately estimate its matrix form and make sure that $\norm{\rho_B-\hat{\rho}_B}_{\mathrm{tr}}\le\epsilon$ with arbitrarily small $\epsilon$.
After that, one classically calculates all its nonzero eigenvalues and eigenstates, $\{\hat{\lambda}_B^j\}_j$ and $\{\hat{\psi}_B^j\}_j$, and then uses them to calculate the eigenvalue factor and construct the observable $\hat{P}_B^{jk}=\ketbra{\hat{\psi}_B^j}{\hat{\psi}_B^k}\otimes\ketbra{\hat{\psi}_B^k}{\hat{\psi}_B^j}+h.c.$.
According to Lemma~\ref{cor:fisher}, the accuracy of eigenvalues is bounded by the accuracy of tomography, $\abs{\lambda_B^i-\hat{\lambda}_B^i}\le\epsilon$.
Therefore, the error of the eigenvalue factor can be bounded.

To estimate the eigenstate factor, we can treat $O^{\otimes 2}\otimes \hat{P}_B^{jk}$ as the observable and estimate its expectation value on $\Psi_{AB}^{\otimes 2}$. 
Denoting the estimator to be $\hat{M}$, the total error is upper bounded by
\begin{equation}
\begin{aligned}
e_1+e_2=&\abs{\Tr\left[\Psi_{AB}^{\otimes 2}\left(O^{\otimes 2}\otimes P_{B}^{jk}\right)\right]-\Tr\left[\Psi_{AB}^{\otimes 2}\left(O^{\otimes 2}\otimes \hat{P}_B^{jk}\right)\right]}+\abs{\hat{M}-\Tr\left[\Psi_{AB}^{\otimes 2}\left(O^{\otimes 2}\otimes \hat{P}_B^{jk}\right)\right]}.
\end{aligned}
\end{equation}
According to results about eigenstates made in Lemma~\ref{cor:fisher} and Lemma~\ref{lemma:fisher2}, $e_1$ can be bounded as the distance $\norm{p_{\mathcal{S}_2}^{jk}-\hat{P}_B^{jk}}_\mathrm{tr}$ is bounded by the accuracy of tomography.
Notice that only single-copy measurements on $\Psi_{AB}$ is allowed, one cannot directly estimate $O^{\otimes 2}\otimes \hat{P}_B^{jk}$ on two copies of $\Psi_{AB}$.
Therefore, it is not straightforward to use the matrix norm of $O^{\otimes 2}\otimes \hat{P}_B^{jk}$ to analyze the sample complexity for reducing $e_2$.
Note that $\hat{P}_B^{jk}$ can be written as the summation of two observables with tensor structure, $\hat{P}_B^{jk}=\frac{1}{2}\left(\hat{P}_{B+}^{jk}\otimes\hat{P}_{B+}^{jk}+\hat{P}_{B-}^{jk}\otimes\hat{P}_{B-}^{jk}\right)$ with $\hat{P}_{B+}^{jk}=\ketbra{\hat{\psi}_B^j}{\hat{\psi}_B^k}+h.c.$ and $\hat{P}_{B-}^{jk}=i\ketbra{\hat{\psi}_B^j}{\hat{\psi}_B^k}+h.c.$, where $i$ is the unit imaginary number.
Therefore, one can estimate the expectation value of $O\otimes\hat{P}_{B+}^{jk}$ or $O\otimes\hat{P}_{B-}^{jk}$ on single-copies of $\Psi_{AB}$ and multiply them to get the final estimator $\hat{M}=\frac{1}{2}\left(\hat{M}_+\times\hat{M}_++\hat{M}_-\times\hat{M}_-\right)$.
The error $e_2$ now can be bounded as
\begin{equation}
\begin{aligned}
e_2\le&\frac{1}{2}\abs{\hat{M}_+\times\hat{M}_+-\Tr\left[\Psi_{AB}\left(O\otimes\hat{P}_{B+}^{jk}\right)\right]^2}+\frac{1}{2}\abs{\hat{M}_-\times\hat{M}_--\Tr\left[\Psi_{AB}\left(O\otimes\hat{P}_{B-}^{jk}\right)\right]^2}\\
\le&\frac{1}{2}\abs{\hat{M}_++\Tr\left[\Psi_{AB}\left(O\otimes\hat{P}_{B+}^{jk}\right)\right]}\times\abs{\hat{M}_+-\Tr\left[\Psi_{AB}\left(O\otimes\hat{P}_{B+}^{jk}\right)\right]}+\frac{1}{2}\abs{\hat{M}_-+\Tr\left[\Psi_{AB}\left(O\otimes\hat{P}_{B-}^{jk}\right)\right]}\times\abs{\hat{M}_--\Tr\left[\Psi_{AB}\left(O\otimes\hat{P}_{B-}^{jk}\right)\right]}.
\end{aligned}
\end{equation}
Now, as norms of $O\otimes\hat{P}_{B+}^{jk}$ and $O\otimes\hat{P}_{B-}^{jk}$ are all bounded, we can use constant sample complexity to bound the value of $e_2$.
This concludes our proof as now errors of both the eigenvalue factor and eigenstate factor can be bounded using tomography on $\rho_B$ with constant sample complexity.

\section{Additional Proofs on the Hardness of Learning Tasks} \label{app:proof_hardness}
\subsection{Learning Tree Formalization}\label{app:learning_tree}

In this subsection, we recap the learning tree formalization introduced in~\cite{aharonov1997fault,bubeck2020entanglement,chen2022memory,huang2021information,chen2023complexity,chen2024optimal}, which is the standard reasoning about adaptive protocols for quantum learning. This formulation will be the key subroutine in proving all the sample lower bounds in the rest of this section.

An arbitrary (adaptive) protocol for learning a quantum state $\rho$ using measurements from a set $\mathcal{M}$ on $c$ copies of $\rho$ at a time can be formulated as follows. Starting from the root node, we select $c$ copies of $\rho$, perform some POVM chosen from $\mathcal{M}$ on $\rho^{\otimes c}$, and move to its child node corresponding to the outcome. In this work, we focus on $c=1$ for single-copy measurements and $c=2$ for quantum measurements with $k\leq n$ qubits of quantum memory (see Definition~\ref{def:learning_k}). Quantitatively, we define the tree representation as follows:
\begin{definition}[Tree representation learning protocols~\cite{chen2024optimal}]
Given an unknown $n$-qubit quantum state $\rho$, any learning protocol can be represented as a rooted tree $\mathcal{T}$ of depth $T$ with every node on the tree recording the measurement outcomes so far. Below are the properties of $\mathcal{T}$:
\begin{itemize}
    \item We associate a probability $p^\rho(u)$ with any node $u$ on the tree $\mathcal{T}$.
    \item The probability associated with the root $r$ of the tree is $p^\rho(r)=1$.
    \item At each non-leaf node $u$, we measure a batch of $\rho^{\otimes c}$ using an adaptively chosen POVM $M_u=\{F_{s}^u\}_s\in \mathcal{M}$, and obtain a classical outcome $s$. The child node $v$ associated with classical outcome $s$ of the node $u$ is connected via the edge $e_{u,s}$.
    \item If $v$ is the child of $u$ via the edge $e_{u,s}$, the probability associated with $v$ is
    \begin{align*}
    p^\rho(v)=p^\rho(u)\Tr(F_s^u \rho^{\otimes c}).
    \end{align*}
    \item Each root-to-leaf path is of length $T$. Note that for a leaf node $\ell$, $p^\rho(\ell)$ is the probability of being in state $\ell$ at the end of the learning protocol. We further denote the set of leaves of $\mathcal{T}$ by $\mathrm{leaf}(\mathcal{T})$.
\end{itemize}
\end{definition}

To prove lower bounds for learning quantum data, we usually consider the reduction from the distinguishing problem between two ensembles of quantum states. For the distinguishing problem in this paper, we can further reduce it to a distinguishing task of the following form between the maximally mixed state and an ensemble: Given access to copies of an unknown state $\rho$, we are to distinguish the following two cases:
\begin{itemize}
\item $\rho$ is the maximally mixed state $\rho_m=\mathbb{I}/2^n$.
\item $\rho$ is sampled from a known ensemble $\mathcal{S}$ over $n$-qubit mixed states.
\end{itemize}
This prototypical distinguishing task is also known as the many-versus-one distinguishing problem. 

There are a line of tools for proving lower bound of a many-versus-one distinguishing problem. Suppose $\mathcal{T}$ be the tree representation of a learning protocol for the many-versus-one distinguishing task that solves it with probability $p_{\text{succ}}$. Le Cam's two-point method~\cite{yu1997assouad} indicates that 
\begin{align}
p_{\text{succ}}&\leq d_{\mathrm{TV}}(\mathbb{E}_{\rho\sim \mathcal{S}}p^\rho, p^{\rho_m})\nonumber\\
&\equiv \frac12\sum_{\ell\in \mathrm{leaf}(\mathcal{T})}\abs{\mathbb{E}_{\rho\sim \mathcal{S}}p^\rho(\ell)- p^{\rho_m}(\ell)},
\end{align}
where $d_{\mathrm{TV}}$ is the total variation distance between two probability distributions.
For each leaf $\ell\in \mathrm{leaf}(\mathcal{T})$, we define the likelihood ratio as
\begin{align}
L(\ell)\equiv \frac{\mathbb{E}_{\rho\sim \mathcal{S}}p^{\rho}(\ell)}{p^{\rho_m}(\ell)}.
\end{align}
We also define the likelihood ratio for each edge $e_{u,s}$ and each state $\rho$:
\begin{align}
L_\rho(u, s)\equiv \frac{p^{\rho}(s|u)}{p^{\rho_m}(s|u)},\quad L_\rho(\ell)\equiv \frac{p^{\rho}(\ell)}{p^{\rho_m}(\ell)}.
\end{align}

Based on the standard Le Cam's method, we can further have a one-side likelihood ratio argument~\cite{chen2022memory}: we have
\begin{align}
&d_{\mathrm{TV}}(\mathbb{E}_{\rho\sim \mathcal{S}}p^\rho, p^{\rho_m})\leq\Pr_{\ell\sim p^{\rho_m}(\ell), \rho\sim S}[L_\rho(\ell)\leq \alpha] + 1-\alpha,\\
&d_{\mathrm{TV}}(\mathbb{E}_{\rho\sim \mathcal{S}}p^\rho, p^{\rho_m})\leq \Pr_{\ell\sim p^{\rho_m}(\ell)}[L(\ell)\leq \alpha] + 1-\alpha,
\end{align}
for any $0<\alpha<1$. Based on this one-side likelihood ratio argument, Ref.~\cite{chen2023complexity} proposed a martingale argument, which is later extended to the context of learning from quantum data~\cite{chen2024optimal}.

\begin{lemma}[Martingale argument, {\cite[Lemma 7]{chen2024optimal}}]\label{lem:martingale}
Suppose there exists is a $\delta>0$ such that, for every node $u$ we have 
\begin{align}
\mathbb{E}_{\rho\sim \mathcal{S}}\mathbb{E}_{s\sim p^{\rho_m}(s|u)}[(L_\rho(u, s)-1)^2]\leq \delta.
\end{align}
We then have
\begin{align}
\Pr_{\ell\sim p^{\rho_m}(\ell), \rho\sim \mathcal{S}}[L_\rho(\ell)\leq 0.9]\leq 0.1 + c\delta T
\end{align}
for some constant $c$.
\end{lemma}

\subsection{Purity Estimation}\label{app:purity_lower}
\subsubsection{Without quantum memory}
In this part, we consider the task of estimating purity $\Tr(\rho^2)$ given an unknown quantum state $\rho$, which can be regarded as computing state moments $\Tr(\rho^k)$ at $k=2$ or a special case of estimating nonlinear quantum observables $\Tr(O\rho^2)$ at $O=I$. We show that even purity estimation for states with known constant rank requires exponential sample complexity. Quantitatively:

\begin{theorem}[Exponential overhead of purity estimation without purification]\label{thm:purity_lower}
Given an $n$-qubit low-rank state $\rho$ with $n\geq 2$ and its rank (even constant), any protocol that predicts $\Tr(\rho^2)$ within constant additive error for arbitrary $\rho$ requires at least $\Omega(2^{n/2})$ sample complexity.
\end{theorem}
\begin{proof}
We consider the following distinguish task. One is given an unknown quantum state $\rho$ randomly drawn from one of the following two state ensembles $\mathcal{S}_1$ and $\mathcal{S}_2$ that only contains pure states or rank-$2$ mixed states:
\begin{itemize}
    \item Ensemble $\mathcal{S}_1$: $\rho=0.9U\ket{0}\bra{0}U^\dagger+0.1V\ket{0}\bra{0}V^\dagger$ is rank-$2$ random state with Haar randomly chosen $U$ and $V$.
    \item Ensemble $\mathcal{S}_2$: $\rho=\frac{1}{2}(U\ket{0}\bra{0}U^\dagger+V\ket{0}\bra{0}V^\dagger)$ is rank-$2$ random state with Haar randomly chosen $U$ and $V$.
\end{itemize}

Here, we show that distinguishing these two ensembles requires at least exponentially many single-copy measurements. We adapt the learning tree scheme in Appendix~\ref{app:learning_tree} for learning with single-copy measurements in the above lower bound. The most general single-copy measurement scheme can be realized by performing rank-$1$ POVMs adaptively on each copy of $\rho$~\cite{chen2022memory}:
\begin{enumerate}
    \item We start from the root $r$. 
    \item At each node $u$, we perform rank-1 POVM $\{2^n w^u_s \ketbra{\psi^u_s}\}$. After obtaining a measurement result $s$, go to the son node $s$ of $u$.
    \item Finally, we stop at a leaf node $l$.
\end{enumerate}
We assume that there are totally $T+1$ nodes, $u_1 = r, u_{T+1} = l$.
For a given state $\rho$, denote $p_{\rho}$ as the probability distribution of reaching leaf nodes. The probability of reaching one leaf node $l$ is denoted by $p_{\rho}(l)$.
 
Suppose we have an algorithm that, given an unknown state $\rho$ that is guaranteed to have rank at most $2$, outputs the purity within constant additive error $\epsilon\leq \mathcal{O}(1)$ with probability at least $5/6$. In $\mathcal{S}_1$, $\Tr(\rho^2)=0.82+0.18\Tr(U\ket{0}\bra{0}U^\dagger V\ket{0}\bra{0}V^\dagger)$. In $\mathcal{S}_2$, we have $\Tr(\rho^2)=\frac12+\frac12\Tr(U\ket{0}\bra{0}U^\dagger V\ket{0}\bra{0}V^\dagger)$. We have
\begin{align}
\mathbb{E}_{\rho\sim \mathcal{S}_1}[\Tr(\rho^2)]&=0.82+0.18\Tr(U\ket{0}\bra{0}U^\dagger V\ket{0}\bra{0}V^\dagger)\nonumber\\
&=0.82+\frac{0.09}{d}.
\end{align}
Given than $n\geq 2$ and thus $d\geq 4$, we have $\mathbb{E}_{\rho\sim \mathcal{S}_1}[\Tr(\rho^2)]\leq 5/8$. By Markov inequality, we have with high probability, a random state $\rho$ in $\mathcal{S}_2$ that has purity at most $4/5$. Thus, we can use our purity estimation algorithm to distinguish these two cases with high probability for some constant additive $\epsilon$.

By Le Cam's two-point method~\cite{yu1997assouad}, to successfully distinguish these two state ensembles with high probability, we require
\begin{align}
    \TV(\bE_{\mathcal{S}_1} p_{\rho}, \bE_{\mathcal{S}_2} p_{\rho}) \ge \Theta(1).
\end{align}
We can verify that $\TV(p_{\rho_m},\bE_{\mathcal{S}_1} p_{\rho})\leq \TV(p_{\rho_m},\bE_C p_{\rho}) \leq 1 -  (1+\frac{T}{2^{n-2}})^{-T}$ and $\TV(p_{\rho_m},\bE_{\mathcal{S}_2} p_{\rho})\leq \TV(p_{\rho_m},\bE_C p_{\rho}) \leq 1 -  (1+\frac{T}{2^{n-2}})^{-T}$, where $S_C$ is the ensemble containing Haar random states. Here, the second inequality hold because you can write $\mathbb{E}_{U,V}[(\frac 12(U\ket{0}\bra{0}U^\dagger+V\ket{0}\bra{0}V^\dagger))^k]$ into $2^k$ terms and each term is bounded above by $\mathbb{E}_U[\frac{1}{2^k}(U\ket{0}\bra{0}U^\dagger)^k]$ (and similar for the first). Thus we have
\begin{align}
\TV(\bE_{\mathcal{S}_1} p_{\rho}, \bE_{\mathcal{S}_2} p_{\rho})\leq 2\left(1 -  \left(1+\frac{T}{2^{n-2}}\right)^{-T}\right), 
\end{align}
which combined with the previous equation gives $T\geq\Omega(2^{n/2})$.
\end{proof}

\subsubsection{With bounded quantum memory}\label{app:purity_memory_lower}
In this part, we show that given a low-rank quantum state and its rank (even constant), estimating the purity with a $k$-qubit quantum memory to a constant additive error requires at least $\Omega(\min\{2^{n-k},2^{n/2}\})$ copies of $\rho$. We still consider distinguishing between the same ensembles:
\begin{itemize}
    \item Ensemble $\mathcal{S}_1$: $\rho=0.9U\ket{0}\bra{0}U^\dagger+0.1V\ket{0}\bra{0}V^\dagger$ is rank-$2$ random state with Haar randomly chosen $U$ and $V$.
    \item Ensemble $\mathcal{S}_2$: $\rho=\frac{1}{2}(U\ket{0}\bra{0}U^\dagger+V\ket{0}\bra{0}V^\dagger)$ is rank-$2$ random state with Haar randomly chosen $U$ and $V$.
\end{itemize}

We formalize the model of learning with $k\leq n$ qubits of quantum memory~\cite{chen2022memory,chen2024optimal} (see Fig.~\ref{fig:bounded_memory}).

\begin{definition}[Estimating purity with $k\leq n$ qubits of quantum memory]\label{def:learning_k}
The algorithm uses $N/2$ rounds and each round input two copies of $\rho$. In each round, the algorithm initializes a $k$-qubit quantum memory $\varrho$. We first get a new copy of $\rho$, select an $(n\to k)$ POVM $\{M_s^\dagger M_s\}_s$ to measure $\rho$, obtain the outcome $s$ w.p. $\Tr(\rho\cdot M_s^\dagger M_s)$, and store the $k$-qubit post-measurement state $\varrho_{s,k}=M_s\rho M_s^\dagger/\Tr(\rho M_s^\dagger M_s)$ into the $k$-qubit quantum memory. We then take another copy of $\rho$, and select an $(n+k\to 0)$ POVM $\{M_{s'}^\dagger M_{s'}\}_{s'}$ to measure the new copy and the joint state $\varrho_{s,k}\otimes\rho$, obtain the outcome $s'$ w.p. $\Tr((\varrho_{s,k}\otimes\rho)\cdot M_{s'}^\dagger M_{s'})$. After $N/2$ rounds, the algorithm predicts purity $\Tr(\rho^2)$ of $\rho$ based on the classical output. The total sample complexity of the protocol is $N$.
\end{definition}

We note that the definition here for learning with $k$ qubits of quantum memory is different from the original one~\cite{chen2022memory}: the latter one is able to extract information from unlimited copies of $\rho$ before each measurement while the definition here only allows extracting information from two copies before each measurement. Our definition can be regarded as learning with $2$-copy measurements and $k$ qubits of quantum memory~\cite{chen2024optimal}.

The learning tree scheme now becomes the following scheme.
\begin{itemize}
    \item We start from the root $r$;
    \item At each node $u$, we perform a two-copy rank-$1$ POVM $\{F_u^s=2^nw^s_u\ket{\psi^s_u}\bra{\psi^s_u}\}_s$ that satisfies Definition~\ref{def:learning_k}. After obtaining a measurement result $s$, we go to the son node $v$ of $u$.
    \item Finally, we reach a leaf $l$. We assume that the nodes and results corresponding to $l$ are $\{(u_t,s_t)\}$ and $u_1=r$.
\end{itemize}

Similar to the previous argument using Le Cam's two-point method~\cite{yu1997assouad}, we now need to compute $\TV(\bE_{\mathcal{S}_1} p_\rho,\bE_{\mathcal{S}_2} p_\rho)$. To compute a bound for $\TV(\bE_{\mathcal{S}_1} p_\rho,\bE_{\mathcal{S}_2} p_\rho)$, we only need to compute $\TV(\bE_{\mathcal{S}_1} p_\rho,p_{\rho_m})$ and $\TV(\bE_{\mathcal{S}_2} p_\rho,p_{\rho_m})$, and then apply the triangle inequality. We now consider the following values:
\begin{align}\label{eq:TV_rhoq_rhom}
\TV(\bE_Q p_\rho,p_{\rho_m}),
\end{align}
where $Q$ corresponding to the ensemble $S_Q$ that contains $\rho=qU\ket{0}\bra{0}U^\dagger+(1-q)V\ket{0}\bra{0}V^\dagger$, which are rank-$2$ random states with Haar randomly chosen $U$ and $V$, and some constant $q>0$. We can further write \eqref{eq:TV_rhoq_rhom} as:
\begin{align}
\begin{split}
\TV(\bE_Q p_\rho,p_{\rho_m})&=\frac{1}{2}\sum_l\abs{p_{\rho_m}(l)-E_Qp_\rho(l)}\\
&=\sum_l\max\{0,p_{\rho_m}(l)-E_Qp_\rho(l)\}\\
&=\bE_{l\sim p_{\rho_m}}\max\{0,1-E_QL_{\rho}(l)\}\\
&\leq\bE_{l\sim p_{\rho_m},Q}\max\{0,1-L_{\rho}(l)\},\label{eq:tv_ratio}
\end{split}
\end{align}
where $L_\rho(l)=p_\rho(l)/p_{\rho_m}(l)$ is the likelihood ratio between the two cases of reaching leaf $l$, and the last line follows from the convexity of $\max\{0,x\}$ function. We now compute $L_\rho(l)$ and $\bE_Q L_\rho(l)$:
\begin{align}
L_\rho(l)=\prod_{t=1}^T\frac{\Tr(F_{u_t}^{s_t}\rho^{\otimes2})}{\Tr(F_{u_t}^{s_t}\rho_m^{\otimes2})}=\frac{\Tr(\bigotimes_{t=1}^TF_{u_t}^{s_t}\rho^{\otimes2T})}{\Tr(\bigotimes_{t=1}^TF_{u_t}^{s_t}\rho_m^{\otimes2T})},
\end{align}
where we denote $\rho_m=I/2^n$ for simplicity. Average over $\rho$, we have two terms every $\rho$. Regarding these terms, we can split the $T$ pairs ($2T$ copies) of states into $T_{00},T_{01},T_{10},T_{11}$ that contains the terms of $U\ket{0}\bra{0}U^\dagger\otimes U\ket{0}\bra{0}U^\dagger$, $U\ket{0}\bra{0}U^\dagger\otimes V\ket{0}\bra{0}V^\dagger$, $V\ket{0}\bra{0}V^\dagger\otimes U\ket{0}\bra{0}U^\dagger$, and $V\ket{0}\bra{0}V^\dagger\otimes V\ket{0}\bra{0}V^\dagger$. We denote $S^U_{T_{00},T_{01},T_{10}}$ as the sum of all permutations on all two copies in $T_{00}$, the first copy in $T_{01}$, and the second copy in $T_{10}$. Similarly we can define $S^V_{T_{01},T_{10},T_{11}}$. We denote $s_U=2\abs{T_{00}}+\abs{T_{01}}+\abs{T_{10}}$ and $s_V=2T-s_V$. By summing over all possible partitions of $[T]$ into $T_{00},T_{01},T_{10},T_{11}$, we can compute the $L(l)=\bE_Q L_\rho(l)$ as
\begin{align}
\begin{split}
L(l)&=\bE_Q L_\rho(l)\\
&=\bE_{U,V} \frac{\Tr(\bigotimes_{t=1}^TF_{u_t}^{s_t}\cdot(qU\ket{0}\bra{0}U^\dagger+(1-q)V\ket{0}\bra{0}V^\dagger)^{\otimes2T})}{\Tr(\bigotimes_{t=1}^TF_{u_t}^{s_t}\cdot\rho_m^{\otimes2T})}\\
&=\sum_{T_{00},T_{01},T_{10},T_{11}}\frac{q^{s_U}(1-q)^{s_V}(2^n)^{2T}}{2^n...(2^n+s_U-1)2^n...(2^n+s_V-1)}\cdot\frac{\Tr(\bigotimes_{t=1}^TF_{u_t}^{s_t}\cdot S^U_{T_{00},T_{01},T_{10}}\otimes S^V_{T_{01},T_{10},T_{11}})}{\Tr(\bigotimes_{t=1}^TF_{u_t}^{s_t})}\\
&\geq\left(1-\frac{4T^2}{2^n}\right)\sum_{T_{00},T_{01},T_{10},T_{11}}q^{s_U}(1-q)^{s_V}\cdot\frac{\Tr(\bigotimes_{t=1}^TF_{u_t}^{s_t}\cdot S^U_{T_{00},T_{01},T_{10}}\otimes S^V_{T_{01},T_{10},T_{11}})}{\prod_{t=1}^T\Tr(F_{u_t}^{s_t})}.
\end{split}
\end{align}
We want to show that with high probability $0.99-\gamma 2^{k-n} T$, each individual term
\begin{align}
\frac{\Tr(\bigotimes_{t=1}^TF_{u_t}^{s_t}\cdot S^U_{T_{00},T_{01},T_{10}}\otimes S^V_{T_{01},T_{10},T_{11}})}{\prod_{t=1}^T\Tr(F_{u_t}^{s_t})}
\end{align}
is bounded below by $0.99$ for some constant $\gamma>0$. Suppose we can prove this, we can then show that with high probability $0.9\cdot\gamma'2^{k-n} T$, each $L_\rho(l)$ is bounded below by $0.9-(1-\frac{4T^2}{2^n})$ for some constant $\gamma'>0$. By plugging this into \eqref{eq:tv_ratio}, we can bound the TV distance to prove the lower bound of $T=\Omega(\min\{2^{n-k},2^{n/2}\})$.

Now, we consider how to show each individual term is bounded below by $0.99$ with high probability $0.99-\gamma 2^{k-n} T$ for some constant $\gamma>0$. We will need the following lemma:
\begin{lemma}[Lemma 16,~\cite{chen2024optimal}]\label{lem:split}
Given two positive numbers $x,y$, suppose $\rho_x$ and $\rho_y$ are two positive semi-definite operators on $x$ and $y$ qubits each of $2^n$-dimensional. Suppose $S_x$, $S_y$, and $S_{x+y}$ are the sum of all permutations on the $x$, $y$, and $x+y$ qubits, we have
\begin{align}
\Tr(\rho_x\otimes \rho_y S_{x+y})\geq\Tr(\rho_x S_x)\Tr(\rho_y S_y).
\end{align}
\end{lemma}

Via this lemma, we have each individual term satisfies
\begin{align}
&\quad\frac{\Tr(\bigotimes_{t=1}^TF_{u_t}^{s_t}\cdot S^U_{T_{00},T_{01},T_{10}}\otimes S^V_{T_{01},T_{10},T_{11}})}{\prod_{t=1}^T\Tr(F_{u_t}^{s_t})}\nonumber\\
&\geq \frac{\Tr(\bigotimes_{t\in T_{00}}F_{u_t}^{s_t}\cdot S^U_{T_{00},0,0})\cdot\Tr(\bigotimes_{t\in T_{11}}F_{u_t}^{s_t} S^V_{0,0,T_{11}})}{\prod_{t\in T_{00}}\Tr(F_{u_t}^{s_t})\cdot\prod_{t\in T_{11}}\Tr(F_{u_t}^{s_t})\cdot\prod_{t\in T_{01}\cup T_{10}}\Tr(F_{u_t}^{s_t})}\cdot\nonumber\\
&\Tr(\bigotimes_{t\in T_{01}\cup T_{10}}F_{u_t}^{s_t}\cdot S^U_{0,T_{01},T_{10}}\otimes S^V_{T_{01},T_{10},0}).
\end{align}

Note that $F_{u_t}^{s_t}=2^nw_{u_t}^{s_t}\ket{\psi_{u_t}^{s_t}}\bra{\psi_{u_t}^{s_t}}$, we have:
\begin{align}
&\frac{\Tr(\bigotimes_{t\in S_{01}\cup S_{10}}F_{u_t}^{s_t}\cdot S^U_{0,T_{01},T_{10}}\otimes S^V_{T_{01},T_{10},0})}{\prod_{t\in T_{01}\cup T_{10}}\Tr(F_{u_t}^{s_t})}\nonumber\\
=&\Tr(\bigotimes_{t\in T_{01}\cup T_{10}}\ket{\psi_{u_t}^{s_t}}\bra{\psi_{u_t}^{s_t}}\cdot S^U_{0,T_{01},T_{10}}\otimes S^V_{T_{01},T_{10},0}).
\end{align}
In the above equation, we are actually computing the trace of $\abs{T_{01}}+\abs{T_{10}}$ pairs of two-copy state $\ket{\psi_{u_t}^{s_t}}\bra{\psi_{u_t}^{s_t}}$ after the operator that sums the permutations over all the first copy of each $\ket{\psi_{u_t}^{s_t}}\bra{\psi_{u_t}^{s_t}}$ Kronecker times the sum of the permutations over all the second copy. By the symmetry invariant property of $S^V_{T_{01},T_{10},0}$, we further have
\begin{align}
\begin{split}
&\Tr(\bigotimes_{t\in T_{01}\cup T_{10}}\ket{\psi_{u_t}^{s_t}}\bra{\psi_{u_t}^{s_t}}\cdot S^U_{0,T_{01},T_{10}}\otimes S^V_{T_{01},T_{10},0})\\
=&\sum_{\pi\in S^U_{0,T_{01},T_{10}}}\Tr(\bigotimes_{t\in T_{01}\cup T_{10}}\ket{\psi_{u_t}^{s_t}}\bra{\psi_{u_t}^{s_t}}\cdot \pi\otimes S^V_{T_{01},T_{10},0})\\
=&\sum_{\pi\in S^U_{0,T_{01},T_{10}}}\Tr(\bigotimes_{t\in T_{01}\cup T_{10}}\ket{\psi_{u_t}^{s_t}}\bra{\psi_{u_t}^{s_t}}\cdot I_{d}^{\otimes \abs{T_{01}}+\abs{T_{10}} }\otimes S^V_{T_{01},T_{10},0}).
\end{split}
\end{align}
We now use the following lemma:
\begin{lemma}[Lemma 5.12,~\cite{chen2022memory}]
For any collection of pure states $\ket{\psi_1},,,.\ket{\psi_T}$, we have
\begin{align}
\sum_{\pi\in S_T}\Tr(\pi\cdot\bigotimes_{i=1}^T\ket{\psi_i}\bra{\psi_i})\geq 1.
\end{align}
\end{lemma}
Plugging into the above equation, for at least one non-empty $S_{01}$ and $S_{10}$ we have
\begin{align}
\begin{split}
&\Tr(\bigotimes_{t\in T_{01}\cup T_{10}}\ket{\psi_{u_t}^{s_t}}\bra{\psi_{u_t}^{s_t}}\cdot S^U_{0,T_{01},T_{10}}\otimes S^V_{T_{01},T_{10},0})\\
=&\sum_{\pi\in S^U_{0,T_{01},T_{10}}}\Tr(\bigotimes_{t\in T_{01}\cup T_{10}}\ket{\psi_{u_t}^{s_t}}\bra{\psi_{u_t}^{s_t}}\cdot I_{d}^{\otimes \abs{T_{01}}+\abs{T_{10}} }\otimes S^V_{T_{01},T_{10},0})\\
\geq& \abs{S^U_{0,T_{01},T_{10}}}\\
\geq& 1.
\end{split}
\end{align}

We thus have for each individual term:
\begin{align}
\begin{split}
&\frac{\Tr(\bigotimes_{t=1}^TF_{u_t}^{s_t}\cdot S^U_{T_{00},T_{01},T_{10}}\otimes S^V_{T_{01},T_{10},T_{11}})}{\prod_{t=1}^T\Tr(F_{u_t}^{s_t})}\\
\geq& \frac{\Tr(\bigotimes_{t\in T_{00}}F_{u_t}^{s_t}\cdot S^U_{T_{00},0,0})\cdot\Tr(\bigotimes_{t\in T_{11}}F_{u_t}^{s_t} S^V_{0,0,T_{11}})}{\prod_{t\in T_{00}}\Tr(F_{u_t}^{s_t})\cdot\prod_{t\in T_{11}}\Tr(F_{u_t}^{s_t})}\\
\geq& \frac{\prod_{t\in T_{00}\cup T_{11}}\Tr(F_{u_t}^{s_t}\cdot \mathcal{S}_2)}{\prod_{t\in T_{00}\cup T_{11}}\Tr(F_{u_t}^{s_t})},
\end{split}
\end{align}
where the second line follows again from Lemma~\ref{lem:split} and $\mathcal{S}_2$ is the sum of two-element permutations ($I\otimes I+\mathrm{SWAP}$). 

We can thus reach the desired argument of each individual term bounded below by $0.99$ with high probability $0.99-\gamma 2^{k-n} T$ for some constant $\gamma>0$ by combining a standard martingale argument in Lemma~\ref{lem:martingale} and the following bound on the second moment:
\begin{lemma}[Lemma 17,~\cite{chen2024optimal}]\label{lem:swap_second}
Suppose $F_{u_t}^{s_t}$ is a POVM that corresponding to the POVMs considered in Definition~\ref{def:learning_k} (learning $k$-qubit quantum memory), we have for any such $F_{u_t}^{s_t}$:
\begin{align}
\sum_{s_t}\frac{\Tr(F_{u_t}^{s_t}\mathrm{SWAP})}{\Tr(F_{u_t}^{s_t})}\leq 2^{k+n}.
\end{align}
\end{lemma}

In conclusion, we formally have the following theorem:
\begin{theorem}[Exponential overhead of purity estimation with bounded quantum memory without purification]\label{thm:purity_lower_mem}
Given an $n$-qubit low-rank state $\rho$ and its rank (even constant), predicting $\Tr(\rho^2)$ within constant additive error requires at least $\Omega(\min\{2^{n/2},2^{n-k}\})$ sample complexity for protocols with $k$ qubits of quantum memory and twice interactions with $\rho$.
\end{theorem}

\subsection{Quantum Virtual Cooling and Principle Component Analysis}\label{app:cooling_pca_lower}
\subsubsection{Without quantum memory}
Here, we consider the task of virtual cooling and quantum principle analysis, that is, predicting the expectation value $\Tr(\rho^2 O)$ of an observable $O$, as well as the difficulty of quantum principal component analysis, which involves predicting the expectation value $\Tr(\ketbra{\psi_0} O)$ without purification. Here, $\ketbra{\psi_0}$ represents the principal component, the eigenstate corresponding to the largest eigenvalue of the input state $\rho$. We show that estimating such non-linear quantity for constant low-rank states requires exponential sample complexity. 

\begin{theorem}[Exponential overhead of quantum virtual cooling and principal component analysis without purification]\label{thm:vc_pca_lower}
Given an $n$-qubit low-rank state $\rho$ and its rank or eigenvalues, any protocol that predicts $\Tr(O\rho^2)$ or $\Tr(O\ketbra{\psi_1})$ for arbitrary observable $O$ with $\norm{O}_\infty=1$ within constant additive error requires at least $\Omega(2^{n/2})$ sample complexity. Here, $\ketbra{\psi_1}$ is the principal component of $\rho$.
\end{theorem}

\begin{proof}
    
Consider the following distinguishing task: A state $\rho$ is randomly drawn from one of the following two state ensembles $\mathcal{S}_1$ and $\mathcal{S}_2$ that only contains low-rank mixed states:
\begin{enumerate}
    \item {Ensemble $\mathcal{S}_1$: 
    \begin{align}
        \rho = \frac{1}{2} \ketbra{0} \otimes (U_1 \ketbra{0}^{\otimes n-1} U_1^{\dagger}) + \frac{1}{4} [\ketbra{1} \otimes (\ketbra{0} \otimes U_2 \ketbra{0}^{\otimes n-2} U_2^{\dagger} + \ketbra{1} \otimes U_3 \ketbra{0}^{\otimes n-2} U_3^{\dagger})].
    \end{align}
    }
    \item {Ensemble $\mathcal{S}_2$:
        \begin{align}
            \rho = \frac{1}{2} \ketbra{1} \otimes (U_1 \ketbra{0}^{\otimes n-1} U_1^{\dagger}) + \frac{1}{4} [\ketbra{0} \otimes (\ketbra{0} \otimes U_2 \ketbra{0}^{\otimes n-2} U_2^{\dagger} + \ketbra{1} \otimes U_3 \ketbra{0}^{\otimes n-2} U_3^{\dagger})].
        \end{align}
    }
\end{enumerate}
In both state ensembles, $U_1$ is Haar-random unitary on $n-1$ qubits, and $U_2, U_3$ are independent Haar-random unitaries on $n-2$ qubits. Our task is to distinguish whether $\rho$ is drawn from state ensemble $\mathcal{S}_1$ or ensemble $\mathcal{S}_2$.

We consider the following properties:
\begin{enumerate}
    \item {Quantum principal component analysis. That is, given a state $\rho = \sum_j \lambda_j \ketbra{\psi_j}$ which has a single maximal eigenvalue $\lambda_1$, and given an observable $O$, output $\Tr(O\ketbra{\psi_1})$.}
    \item {Nonlinear high-order observable. Given a state $\rho$ and an observable $O$, output $\Tr(O \rho^2)$.}
\end{enumerate}

If we are able to predict these properties, we can distinguish whether $\rho$ is drawn from $\mathcal{S}_1$ or from $\mathcal{S}_2$. To see this, take nonlinear observable as an example. We have:
\begin{align}
    \Tr(\rho^2 Z_1) =  \begin{cases}
        \frac{1}{8}, & \rho \in \mathcal{S}_1, \\
        -\frac{1}{8}, & \rho \in \mathcal{S}_2.
    \end{cases} 
\end{align}

Then, we prove that the distinguishing task requires exponentially many single-copy measurements, thus proving the hardness of predicting these properties by single-copy measurements without purifications.

Here, we show that distinguishing these two ensembles requires at least exponentially many single-copy measurements. The most general single-copy measurement scheme can be realized by performing rank-1 POVMs adaptively on each copy of $\rho$~\cite{chen2022memory}:
\begin{enumerate}
    \item We start from the root $r$. 
    \item At each node $u$, we perform rank-1 POVM $\{2^n w^u_s \ketbra{\psi^u_s}\}$. After obtaining a measurement result $s$, go to the son node $s$ of $u$.
    \item Finally, we stop at a leaf node $l$.
\end{enumerate}

Suppose there are totally $T+1$ nodes, $u_1 = r, u_{T+1} = l$.
For a given state $\rho$, denote $p_{\rho}$ as the probability distribution of reaching leaf nodes. The probability of reaching one leaf node $l$ is given by
\begin{align}
    p_{\rho}(l) = \prod_{i=1}^{T} 2^n w^u_s \Tr(\ketbra{\psi_{u_{i+1}}^{u_i}}\rho).   
\end{align}

For the maximally mixed state, we have
\begin{align}
    p_{\rho_m}(l) = \prod_{i=1}^{T} w^u_s.
\end{align}

We will show later that on average, the total variation distance between $p_{\rho_m}$ and $p_{\rho}$ can be bounded.
\begin{lemma}\label{lem:pca}
    \begin{align}
        \TV(p_{\rho_m}, \bE_{\mathcal{S}_1} p_{\rho}) \le 1 -  (1+\frac{T}{2^{n-2}})^{-T}.
    \end{align}
\end{lemma}

Similarly, we have $\TV(p_{\rho_m},\bE_{\mathcal{S}_2} p_{\rho}) \le 1 -  (1+\frac{T}{2^{n-2}})^{-T}$. From triangular inequality, we obtain
\begin{align}
\begin{split}
    \TV(\bE_{\mathcal{S}_1} p_{\rho}, \bE_{\mathcal{S}_2} p_{\rho}) &\le 2 -  2(1+\frac{T}{2^{n-2}})^{-T} \\
    &=2-2\exp(-T\ln(1+T/2^{n-2}))  \\
    &\le 2 - 2\exp(-T^2/2^{n-2}) \\
    &\le 2-2(1-T^2/2^{n-2}) \\
    &= \frac{T^2}{2^{n-3}}.
\end{split}
\end{align}

By Le Cam's two-point method~\cite{yu1997assouad}, to successfully distinguish these two state ensembles with probability greater than $\frac{5}{6}$, the total variation distance must be greater than a constant, which means $T = \Omega(2^{n/2})$.
\end{proof}

\begin{proof}[Proof of Lemma \ref{lem:pca}]
\begin{align}
    \TV(p_{\mathcal{S}_1}, p_{\mathcal{S}_2}) &= \frac{1}{2} \sum_l |p_{\mathcal{S}_1}(l) - p_{\mathcal{S}_2}(l)|\nonumber\\
    &= \sum_l \max(0, p_{\mathcal{S}_1}(l) - p_{\mathcal{S}_2}(l)) \nonumber\\
    &= \sum_l p_{\mathcal{S}_1}(l) \max(0, 1 - \frac{p_{\mathcal{S}_2}(l)}{p_{\mathcal{S}_1}(l)}).
\end{align}
\begin{align}
    &\TV(p_{\rho_m}, \bE_{\mathcal{S}_1} p_{\rho}) \nonumber\\\label{eq:distance_A_mixed}
    =& \sum_l p_{\rho_m}(l) \max\left(0, 1- \bE \prod_{i=1}^{T} 2^n \Tr(\ketbra{\psi_{u_{i+1}}^{u_i}}\rho) \right).
\end{align}

To bound $\bE \prod_{i=1}^{T} \Tr(\ketbra{\psi_{u_{i+1}}^{u_i}}\rho)$, we decompose $\ket{\psi_{u_{i+1}}^{u_i}}$ into
\begin{align}
    \ket{\psi_{u_{i+1}}^{u_i}} =& \alpha_{u_{i+1}}^{u_i} \ket{0} \otimes \ket{\phi_{u_{i+1}}^{u_i}(0)} + \beta_{u_{i+1}}^{u_i} \ket{10} \otimes \ket{\phi_{u_{i+1}}^{u_i}(10)} +\nonumber\\
    &\gamma_{u_{i+1}}^{u_i} \ket{11} \otimes \ket{\phi_{u_{i+1}}^{u_i}(11)},
\end{align}
where $|\alpha_{u_{i+1}}^{u_i}|^2 + |\beta_{u_{i+1}}^{u_i}|^2 + |\gamma_{u_{i+1}}^{u_i}|^2 = 1$. And denote $\rho$ as 
\begin{align}
    \rho = &\frac{1}{2} \ketbra{0} \otimes \ketbra{\psi_0} + \frac{1}{4} \ketbra{1} \otimes (\ketbra{0} \otimes \ketbra{\psi_{10}} + \nonumber\\
    &\ketbra{1} \otimes \ketbra{\psi_{11}}),
\end{align}

Then we have
\begin{align}\label{eq:moments}
\begin{split}
    &\bE \prod_{i=1}^{T} 2^n \Tr(\ketbra{\psi_{u_{i+1}}^{u_i}}\rho) \\
    = &\bE \prod_{i=1}^{T} (2^{n-1}|\alpha_{u_{i+1}}^{u_i}|^2 |\bra{\phi_{u_{i+1}}^{u_i}(0)} \ket{\psi_0}|^2 \\
    &+ 2^{n-2}|\beta_{u_{i+1}}^{u_i}|^2 |\bra{\phi_{u_{i+1}}^{u_i}(10)} \ket{\psi_{10}}|^2 \\
    &+ 2^{n-2}|\gamma_{u_{i+1}}^{u_i}|^2 |\bra{\phi_{u_{i+1}}^{u_i}(11)} \ket{\psi_{11}}|^2) \\
    = & \sum_{S_0,S_{10},S_{11}} \bE \prod_{i \in S_0} 2^{n-1}|\alpha_{u_{i+1}}^{u_i}|^2 |\bra{\phi_{u_{i+1}}^{u_i}(0)} \ket{\psi_0}|^2\cdot\\
    &\prod_{i \in S_{10}} 2^{n-2}|\beta_{u_{i+1}}^{u_i}|^2 |\bra{\phi_{u_{i+1}}^{u_i}(10)} \ket{\psi_{10}}|^2\cdot\\
    &\prod_{i \in S_{11}} 2^{n-2}|\gamma_{u_{i+1}}^{u_i}|^2 |\bra{\phi_{u_{i+1}}^{u_i}(11)} \ket{\psi_{11}}|^2,
    \end{split}
\end{align}
where $\sum_{S_0,S_{10},S_{11}}$ denotes $\sum_{S_0,S_{10},S_{11}: \cup S_i = [T], S_i \cap S_j=\emptyset}$. From Lemma \ref{lem:high_moment_haar}, we have
\begin{align}
\begin{split}
    \bE \prod_{i \in S_0} 2^{n-1} |\bra{\phi_{u_{i+1}}^{u_i}(0)} \ket{\psi_0}|^2 & \ge \frac{1}{(1 + \frac{|S_0|-1}{2^{n-1}}) \cdots (1+\frac{1}{2^{n-1}})} \\
    & \ge (1+\frac{T}{2^{n-1}})^{-|S_0|} \\
    &\ge (1+\frac{T}{2^{n-2}})^{-|S_0|}
\end{split}
\end{align}

Similarly, we have
\begin{align}
\begin{split}
    \bE \prod_{i \in S_{10}} 2^{n-2} |\bra{\phi_{u_{i+1}}^{u_i}(10)} \ket{\psi_{10}}|^2 \ge (1+\frac{T}{2^{n-2}})^{-|S_{10}|},\\
    \bE \prod_{i \in S_{11}} 2^{n-2} |\bra{\phi_{u_{i+1}}^{u_i}(11)} \ket{\psi_{11}}|^2 \ge (1+\frac{T}{2^{n-2}})^{-|S_{11}|}.
\end{split}
\end{align}
Then, we can bound Eq.~\eqref{eq:moments}:
\begin{align}
\begin{split}
    &\bE \prod_{i=1}^{T} 2^n \Tr(\ketbra{\psi_{u_{i+1}}^{u_i}}\rho) \\ 
    \ge & \sum_{S_0,S_{10},S_{11}} \prod_{i \in S_0} |\alpha_{u_{i+1}}^{u_i}|^2 \prod_{j \in S_{10}}|\beta_{u_{i+1}}^{u_i}|^2 \cdot\\
    &\prod_{k \in S_{11}} |\gamma_{u_{i+1}}^{u_i}|^2 (1+\frac{T}{2^{n-2}})^{-|S_0|+|S_{10}|+|S_{11}|} \\
    = & \prod_{i=1}^T (|\alpha_{u_{i+1}}^{u_i}|^2 + |\beta_{u_{i+1}}^{u_i}|^2 + |\gamma_{u_{i+1}}^{u_i}|^2) (1+\frac{T}{2^{n-2}})^{-T} \\
    = & (1+\frac{T}{2^{n-2}})^{-T}.
\end{split}
\end{align}

Combining with Eq.~\eqref{eq:distance_A_mixed}, we have
\begin{align}
\begin{split}
    \TV(p_{\rho_m},\bE_{\mathcal{S}_1} p_{\rho}) &\le \sum_l p_{\rho_m}(l) \max\left(0, 1- (1+\frac{T}{2^{n-2}})^{-T} \right)\\
    &= 1 -  (1+\frac{T}{2^{n-2}})^{-T}.
\end{split}
\end{align}
\end{proof}

\subsubsection{With bounded quantum memory}\label{app:vc_pca_lower_mem}

Here, we prove the hardness of quantum virtual cooling and principal component analysis of low-rank quantum states with bounded quantum memory.
\begin{theorem}[Exponential overhead of quantum virtual cooling and principal component analysis with bounded quantum memory without purification]\label{thm:vc_pca_lower_mem}
Given an unknown $n$-qubit low-rank state $\rho$ and its rank or eigenvalues, predicting $\Tr(O\rho^2)$ or $\Tr(O\ketbra{\psi_1})$ within constant additive error requires at least $\Omega(\min\{2^{n/2},2^{n-k}\})$ sample complexity for protocols with $k$ qubits of quantum memory. Here, $\ketbra{\psi_1}$ is the principal component of $\rho$.
\end{theorem}

\begin{proof}
We consider the hard instance used in the single-copy hardness proof: a state $\rho$ is randomly drawn from one of the following two state ensembles $\mathcal{S}_1$ and $\mathcal{S}_2$ that only contains low-rank mixed states:
\begin{enumerate}
    \item {Ensemble $\mathcal{S}_1$: 
    \begin{align}
        \rho = \frac{1}{2} \ketbra{0} \otimes (U_1 \ketbra{0}^{\otimes n-1} U_1^{\dagger}) + \frac{1}{4} [\ketbra{1} \otimes (\ketbra{0} \otimes U_2 \ketbra{0}^{\otimes n-2} U_2^{\dagger} + \ketbra{1} \otimes U_3 \ketbra{0}^{\otimes n-2} U_3^{\dagger})].
    \end{align}
    }
    \item {Ensemble $\mathcal{S}_2$:
        \begin{align}
            \rho = \frac{1}{2} \ketbra{1} \otimes (U_1 \ketbra{0}^{\otimes n-1} U_1^{\dagger}) + \frac{1}{4} [\ketbra{0} \otimes (\ketbra{0} \otimes U_2 \ketbra{0}^{\otimes n-2} U_2^{\dagger} + \ketbra{1} \otimes U_3 \ketbra{0}^{\otimes n-2} U_3^{\dagger})].
        \end{align}
    }
\end{enumerate}
In both state ensembles, $U_1$ is Haar-random unitary on $n-1$ qubits, and $U_2, U_3$ are Haar-random unitaries on $n-2$ qubits. Our task is to distinguish whether $\rho$ is drawn from state ensemble $\mathcal{S}_1$ or ensemble $\mathcal{S}_2$.

Similar to the argument in Appendix~\ref{app:purity_memory_lower}, we only need to compute the total variation distance
\begin{align}
\TV(\bE_Q p_\rho,p_{\rho_m}),
\end{align}
where $Q$ corresponding to the ensemble $S_Q$ that contains
\begin{align}
&q_1\ketbra{0} \otimes (U_1 \ketbra{0}^{\otimes n-1} U_1^{\dagger}) +\nonumber\\
&q_2 \ketbra{1} \otimes \ketbra{0} \otimes U_2 \ketbra{0}^{\otimes n-2} U_2^{\dagger} +\nonumber\\
&(1-q_1-q_2)\ketbra{1} \otimes\ketbra{1} \otimes U_3 \ketbra{0}^{\otimes n-2} U_3^{\dagger},
\end{align}
which are rank-$3$ random states with Haar randomly chosen $U_1$, $U_2$, and $U_3$, and some constant $0\leq q_1,q_2\leq1$ with $q_1+q_2\leq1$. By symmetry, we can also bound
\begin{align}
\TV(\bE_{Q'} p_\rho,p_{\rho_m}),
\end{align}
where $Q'$ corresponding to the ensemble $S_{Q'}$ that contains
\begin{align}
&q_1\ketbra{1} \otimes (U_1 \ketbra{0}^{\otimes n-1} U_1^{\dagger})\\
+& q_2 \ketbra{0} \otimes \ketbra{0} \otimes U_2 \ketbra{0}^{\otimes n-2} U_2^{\dagger}\nonumber\\
+& (1-q_1-q_2)\ketbra{0} \otimes\ketbra{1} \otimes U_3 \ketbra{0}^{\otimes n-2} U_3^{\dagger},
\end{align}
which are rank-$3$ random states with Haar randomly chosen $U_1$, $U_2$, and $U_3$, and some constant $0\leq q_1,q_2\leq1$ with $q_1+q_2\leq1$. By triangle inequality and choosing $q_1=1/2$ and $q_2=1/4$, we can compute a bound for
\begin{align}
\TV(\bE_{\mathcal{S}_2} p_\rho,\bE_{\mathcal{S}_2} p_\rho),
\end{align}
which leads to a sample complexity lower bound for the distinguishing problem by Le Cam's method~\cite{yu1997assouad}.

We thus only need to consider the likelihood ratio:
\begin{align}
L_\rho(l)=\prod_{t=1}^T\frac{\Tr(F_{u_t}^{s_t}\rho^{\otimes2})}{\Tr(F_{u_t}^{s_t}\rho_m^{\otimes2})}=\frac{\Tr(\bigotimes_{t=1}^TF_{u_t}^{s_t}\rho^{\otimes2T})}{\Tr(\bigotimes_{t=1}^TF_{u_t}^{s_t}\rho_m^{\otimes2T})},
\end{align}
where we denote $\rho_m=\mathbb{I}/2^n$ for simplicity. Average over $\rho$, we have three terms every $\rho$. Regarding these terms, we can split the $T$ pairs ($2T$ copies) of states into $9$ sets corresponding to the terms containing the Haar random part $U_1$, $U_2$, and $U_3$:
\begin{align}
\begin{split}
&T_{11}:\ U_1\ket{0}\bra{0}U_1^\dagger\otimes U_1\ket{0}\bra{0}U_1^\dagger,\\
&T_{22}:\ U_2\ket{0}\bra{0}U_2^\dagger\otimes U_2\ket{0}\bra{0}U_2^\dagger,\\
&T_{33}:\ U_3\ket{0}\bra{0}U_3^\dagger\otimes U_3\ket{0}\bra{0}U_3^\dagger,\\\
&T_{12}:\ U_1\ket{0}\bra{0}U_1^\dagger\otimes U_2\ket{0}\bra{0}U_2^\dagger,\\
&T_{21}:\ U_2\ket{0}\bra{0}U_2^\dagger\otimes U_1\ket{0}\bra{0}U_1^\dagger,\\
&T_{13}:\ U_1\ket{0}\bra{0}U_1^\dagger\otimes U_3\ket{0}\bra{0}U_3^\dagger,\\
&T_{31}:\ U_3\ket{0}\bra{0}U_3^\dagger\otimes U_1\ket{0}\bra{0}U_1^\dagger,\\
&T_{23}:\ U_2\ket{0}\bra{0}U_2^\dagger\otimes U_3\ket{0}\bra{0}U_3^\dagger,\\
&T_{32}:\ U_3\ket{0}\bra{0}U_3^\dagger\otimes U_2\ket{0}\bra{0}U_2^\dagger,
\end{split}
\end{align}
where $\ket{0}$ is $(n-1)$-qubit all zero state for $U_1$ and $(n-2)$-qubit all zero state for $U_2$ and $U_3$. We denote $S^{U_1}_{T_{11},T_{12},T_{13},T_{21},T_{31}}$ as the sum of all permutations on all two copies of $U_1$ in $T_{11}$, the first copy corresponding to $U_1$ in $T_{12}$ and $T_{13}$, and the second copy corresponding to $U_1$ in $T_{21}$ and $T_{31}$. Similarly we can define $S^{U_2}_{T_{22},T_{21},T_{23},T_{12},T_{32}}$ and $S^{U_3}_{T_{33},T_{31},T_{32},T_{13},T_{23}}$. We denote $s_{U_1}=2\abs{T_{11}}+\abs{T_{12}}+\abs{T_{13}}+\abs{T_{21}}+\abs{T_{31}}$, $s_{U_2}=2\abs{T_{22}}+\abs{T_{21}}+\abs{T_{23}}+\abs{T_{12}}+\abs{T_{32}}$, and $s_{U_3}=2T-s_{U_1}-s_{U_2}$. By summing over all possible partitions of $[T]$ into the $9$ sets, we can compute the $L(l)=\bE_Q L_\rho(l)$ as
\begin{align}
&L(l)=\bE_Q L_\rho(l)\nonumber\\
&=\bE_{U_1,U_2,U_3} \frac{\Tr(\bigotimes_{t=1}^TF_{u_t}^{s_t}\cdot(q_1\ket{0}\bra{0}\otimes U_1\ket{0}\bra{0}U_1^\dagger+q_2\ket{10}\bra{10}\otimes U_3\ket{0}\bra{0}U_3^\dagger+(1-q_1-q_2)\ket{11}\bra{11}\otimes U_3\ket{0}\bra{0}U_3^\dagger)^{\otimes2T})}{\Tr(\bigotimes_{t=1}^TF_{u_t}^{s_t}\cdot\rho_m^{\otimes2T})}\nonumber\\
&=\sum_{T_{ij}}\frac{q_1^{s_{U_1}}q_2^{s_{U_2}}(1-q_1-q_2)^{s_{U_3}}(2^n)^{2T}}{2^n...(2^n+s_{U_1}-1)2^n...(2^n+s_{U_2}-1)2^n...(2^n+s_{U_3}-1)}\cdot\nonumber\\
&\ \frac{\Tr(\bigotimes_{t=1}^TF_{u_t}^{s_t}\cdot \ket{0}\bra{0}^{\otimes s_{U_1}}\otimes S^{U_1}_{T_{11},T_{12},T_{13},T_{21},T_{31}}\otimes \ket{10}\bra{10}^{\otimes s_{U_2}}\otimes S^{U_2}_{T_{22},T_{21},T_{23},T_{12},T_{32}}\otimes \ket{11}\bra{11}^{\otimes s_{U_3}}\otimes S^{U_3}_{T_{33},T_{13},T_{23},T_{31},T_{32}})}{\prod_{t=1}^T\Tr(F_{u_t}^{s_t})}\nonumber\\
&\geq\left(1-\frac{4T^2}{2^n}\right)\sum_{T_{ij}}q_1^{s_{U_1}}q_2^{s_{U_2}}(1-q_1-q_2)^{s_{U_3}}\cdot\nonumber\\
&\ \frac{\Tr(\bigotimes_{t=1}^TF_{u_t}^{s_t}\cdot \ket{0}\bra{0}^{\otimes s_{U_1}}\otimes S^{U_1}_{T_{11},T_{12},T_{13},T_{21},T_{31}}\otimes \ket{10}\bra{10}^{\otimes s_{U_2}}\otimes S^{U_2}_{T_{22},T_{21},T_{23},T_{12},T_{32}}\otimes \ket{11}\bra{11}^{\otimes s_{U_3}}\otimes S^{U_3}_{T_{33},T_{13},T_{23},T_{31},T_{32}})}{\prod_{t=1}^T\Tr(F_{u_t}^{s_t})},
\end{align}
where the summation in the third and the fourth line works for $i,j=1,2,3$. Similar to Appendix~\ref{app:purity_memory_lower}, we only need to argue that each individual term 
\begin{align}
\frac{\Tr(\bigotimes_{t=1}^TF_{u_t}^{s_t}\cdot \ket{0}\bra{0}^{\otimes s_{U_1}}\otimes S^{U_1}_{T_{11},T_{12},T_{13},T_{21},T_{31}}\otimes \ket{10}\bra{10}^{\otimes s_{U_2}}\otimes S^{U_2}_{T_{22},T_{21},T_{23},T_{12},T_{32}}\otimes \ket{11}\bra{11}^{\otimes s_{U_3}}\otimes S^{U_3}_{T_{33},T_{13},T_{23},T_{31},T_{32}})}{\prod_{t=1}^T\Tr(F_{u_t}^{s_t})}
\end{align}
is bounded below by $0.99$ with high probability $0.99-\gamma 2^{k-n}T$ for some constant $\gamma$. We can now use Lemma~\ref{lem:split} similar to Appendix~\ref{app:purity_memory_lower} to split $T_{11}$ from $T_{12},T_{21},T_{13},T_{31}$, and similar for $T_{22}$ and $T_{33}$. Therefore, the term is lower bounded by:
\begin{align}
\frac{\Tr(\bigotimes_{t\in T_{11}\cup T_{22}\cup T_{33}}F_{u_t}^{s_t}\cdot \ket{0}\bra{0}^{\otimes \abs{T_{11}}}\otimes S^{U_1}_{T_{11},0,0,0,0}\otimes \ket{10}\bra{10}^{\otimes \abs{T_{22}}}\otimes S^{U_2}_{T_{22},0,0,0,0}\otimes \ket{11}\bra{11}^{\otimes T_{33}}\otimes S^{U_3}_{T_{33},0,0,0,0})}{\prod_{t\in T_{11}\cup T_{22}\cup T_{33}}\Tr(F_{u_t}^{s_t})}.
\end{align}
Following a same proof of Lemma~\ref{lem:swap_second} (see Lemma 17 of~\cite{chen2024optimal}), we can similarly prove the following corollary
\begin{corollary}\label{cor:single_step_two_moment}
Suppose $F_{u_t}^{s_t}$ is a POVM that corresponding to the POVMs considered in Definition~\ref{def:learning_k} (learning $k$-qubit quantum memory), we have for any such $F_{u_t}^{s_t}$:
\begin{align}
\sum_{s_t}\frac{\Tr(F_{u_t}^{s_t}\rho^{(1)}\otimes\mathrm{SWAP}_{n-1})}{\Tr(F_{u_t}^{s_t})}\leq 2^{k+n},\\
\sum_{s_t}\frac{\Tr(F_{u_t}^{s_t}\rho^{(2)}\otimes\mathrm{SWAP}_{n-2})}{\Tr(F_{u_t}^{s_t})}\leq 2^{k+n},
\end{align}
for arbitrary single-qubit state $\rho^{(1)}$ and two-qubit state $\rho^{(2)}$
\end{corollary}
Plugging this corollary into each term using the standard martingale argument in Lemma~\ref{lem:martingale}, we obtain the theorem.
\end{proof}

\subsection{Quantum Fisher information estimation}\label{app:fisher_lower}
\subsubsection{Without quantum memory}
Here, we prove the hardness of estimating quantum Fisher information without purification and quantum memory.
\begin{theorem}[Exponential overhead of estimation Fisher information without purification]\label{thm:fisher_lower}
Given an $n$-qubit low-rank state $\rho$, its rank and observable $O$ satisfying $\norm{O}_{\infty} = 1$, any protocol that predicts the quantum Fisher information $F_O(\rho)$ within constant additive error for arbitrary $\rho$ requires at least $\Omega(2^{n/2})$ sample complexity. 
\end{theorem}
\begin{proof}
Consider the following distinguishing task: A state $\rho$ is randomly drawn from one of the following two state ensembles $\mathcal{S}_1$ and $\mathcal{S}_2$ that only contains low-rank mixed states:
\begin{enumerate}
    \item {Ensemble $\mathcal{S}_1$: 
    \begin{align}\label{eq:fisher_info_SA}
        \rho = \frac{1}{2} \ketbra{0} \otimes U \ketbra{0}^{\otimes n-1} U^{\dagger} + \frac{3}{8} \ketbra{1} \otimes U \ketbra{0}^{\otimes n-1} U^{\dagger} + \frac{1}{8}\ketbra{1} \otimes V \ketbra{0}^{\otimes n-1} V^{\dagger}.
    \end{align}
    }
    \item {Ensemble $\mathcal{S}_2$:
        \begin{align}\label{eq:fisher_info_SB}
            \rho = \frac{1}{2} \ketbra{+} \otimes U \ketbra{0}^{\otimes n-1} U^{\dagger} + \frac{3}{8} \ketbra{-} \otimes U \ketbra{0}^{\otimes n-1} U^{\dagger} + \frac{1}{8}\ketbra{-} \otimes V \ketbra{0}^{\otimes n-1} V^{\dagger}.
        \end{align}
    }
\end{enumerate}

In both state ensembles, $U, V$ are Haar-random unitary on $n-1$ qubits. For the observable $X_1 = X \otimes \mathbb{I}^{\otimes n-1}$, it is straightforward to check that $F_{X_1} = 0$ for $\rho \in \mathcal{S}_2$. By Corollary \ref{col:fisher}, given $0 < \delta < 1$, for sufficiently large $n$, if we randomly choose $\rho \in \mathcal{S}_1$, then with probability $1-\delta$, $F_{X_1} \ge 0.01$. 

Suppose we have an algorithm that, given an unknown state $\rho$ that is guaranteed to have rank at most 3 and an observable $O$ satisfies $\norm{O}_{\infty} \le 1$, outputs the Fisher information within constant additive error $\epsilon \le \mathcal{O}(1)$ with probability at least $5/6$. Then, this algorithm can distinguish these two state ensembles with probability $5/6$ for sufficiently large n.

We can show that, on average, the total variation distance between $\rho_A$ and $\rho_B$ can be small, which implies the algorithm has exponential sample complexity overhead.

\begin{lemma}\label{lem:fisher}
    \begin{align}
        \TV(p_{\rho_m}, \bE_{\mathcal{S}_1} p_{\rho}) \le 1 -  (1+\frac{T}{2^{n-2}})^{-T}.
    \end{align}
\end{lemma}

Similarly, we have $\TV(p_{\rho_m},\bE_{\mathcal{S}_2} p_{\rho}) \le 1 -  (1+\frac{T}{2^{n-1}})^{-T}$. From triangular inequality, we obtain
\begin{align}
\begin{split}
    \TV(\bE_{\mathcal{S}_1} p_{\rho}, \bE_{\mathcal{S}_2} p_{\rho}) &\le 2 -  2(1+\frac{T}{2^{n-1}})^{-T} \\
    &=2-2\exp(-T\ln(1+T/2^{n-1}))  \\
    &\le 2 - 2\exp(-T^2/2^{n-1}) \\
    &\le 2-2(1-T^2/2^{n-1})\\
    &= \frac{T^2}{2^{n-2}}.
\end{split}
\end{align}

By Le Cam's two-point method~\cite{yu1997assouad}, to successfully distinguish these two state ensembles with probability greater than $\frac{5}{6}$, the total variation distance must be greater than a constant, which means $T = \Omega(2^{n/2})$.
\end{proof}

\begin{proof}[Proof of Lemma \ref{lem:fisher}]
The proof is similar to the proof of Lemma \ref{lem:pca}. 
\begin{align}
\begin{split}
 \TV(p_{\mathcal{S}_1}, p_{\mathcal{S}_2})& = \frac{1}{2} \sum_l |p_{\mathcal{S}_1}(l) - p_{\mathcal{S}_2}(l)|\\
&= \sum_l \max(0, p_{\mathcal{S}_1}(l) - p_{\mathcal{S}_2}(l)) \\
&= \sum_l p_{\mathcal{S}_1}(l) \max(0, 1 - \frac{p_{\mathcal{S}_2}(l)}{p_{\mathcal{S}_1}(l)}).
\end{split}
\end{align}
\begin{align}
    &\TV(p_{\rho_m}, \bE_{\mathcal{S}_1} p_{\rho}) = \nonumber\\\label{eq:distance_A_mixed2}
    &\sum_l p_{\rho_m}(l) \max\left(0, 1- \bE \prod_{i=1}^{T} 2^n \Tr(\ketbra{\psi_{u_{i+1}}^{u_i}}\rho) \right).
\end{align}

To bound $\bE \prod_{i=1}^{T} \Tr(\ketbra{\psi_{u_{i+1}}^{u_i}}\rho)$, we decompose $\ket{\psi_{u_{i+1}}^{u_i}}$ into
\begin{align}
    \ket{\psi_{u_{i+1}}^{u_i}} = \alpha_{u_{i+1}}^{u_i} \ket{0} \otimes \ket{\phi_{u_{i+1}}^{u_i}(0)} +\beta_{u_{i+1}}^{u_i} \ket{1} \otimes \ket{\phi_{u_{i+1}}^{u_i}(1)}.
\end{align}
where $|\alpha_{u_{i+1}}^{u_i}|^2 + |\beta_{u_{i+1}}^{u_i}|^2= 1$. And denote $\rho$ as 
\begin{align}
    \rho = &\frac{1}{2} \ketbra{0} \otimes \ketbra{\psi_0} + \frac{3}{8} \ketbra{1} \otimes \ketbra{\psi_0} + \nonumber\\ 
    &\frac{1}{8}\ketbra{1} \otimes \ketbra{\psi_{1}}.
\end{align}

Then we have
\begin{align}\label{eq:moments_fisher}
\begin{split}
    &\bE \prod_{i=1}^{T} 2^n \Tr(\ketbra{\psi_{u_{i+1}}^{u_i}}\rho) \\
    = &\bE \prod_{i=1}^{T} (2^{n-1}|\alpha_{u_{i+1}}^{u_i}|^2 |\bra{\phi_{u_{i+1}}^{u_i}(0)} \ket{\psi_0}|^2 + \\
    &(2^{n-1}\times \frac{3}{4})|\beta_{u_{i+1}}^{u_i}|^2 |\bra{\phi_{u_{i+1}}^{u_i}(1)} \ket{\psi_{0}}|^2 + \\
    &(2^{n-1} \times \frac{1}{4})|\beta_{u_{i+1}}^{u_i}|^2 |\bra{\phi_{u_{i+1}}^{u_i}(1)} \ket{\psi_{1}}|^2) \\
    = & 2^{(n-1)T}\sum_{S_0,S_{10},\atop S_{11}} (\frac{3}{4})^{|S_{10}|} (\frac{1}{4})^{|S_{11}|}\bE \prod_{i \in S_0} |\alpha_{u_{i+1}}^{u_i}|^2 |\bra{\phi_{u_{i+1}}^{u_i}(0)} \ket{\psi_0}|^2\\
    &\prod_{i \in S_{10}} |\beta_{u_{i+1}}^{u_i}|^2 |\bra{\phi_{u_{i+1}}^{u_i}(1)} \ket{\psi_{0}}|^2 \prod_{i \in S_{11}} |\beta_{u_{i+1}}^{u_i}|^2 |\bra{\phi_{u_{i+1}}^{u_i}(1)} \ket{\psi_{1}}|^2.
    \end{split}
\end{align}
where $\sum_{S_0,S_{10},S_{11}}$ denotes $\sum_{S_0,S_{10},S_{11}: \cup S_i = [T], S_i \cap S_j=\emptyset}$. From Lemma \ref{lem:high_moment_haar}, we have
\begin{align}
\begin{split}
    &\bE \left[ 2^{(n-1)(|S_0| + |S_{10}|)}\prod_{i \in S_0} |\bra{\phi_{u_{i+1}}^{u_i}(0)} \ket{\psi_0}|^2 \prod_{i \in S_{10}} |\bra{\phi_{u_{i+1}}^{u_i}(1)} \ket{\psi_{0}}|^2\right] \\
    &\ge [ (1 + \frac{|S_0|+|S_{10}|-1}{2^{n-1}}) \cdots (1+\frac{1}{2^{n-1}})]^{-1}  \\
    &\ge (1+\frac{T}{2^{n-1}})^{-|S_{10}|}. 
\end{split}
\end{align}
Similarly, we have:
\begin{align}
    \bE \left[ 2^{(n-1)|S_{11}|}\prod_{i \in S_{11}} |\bra{\phi_{u_{i+1}}^{u_i}(1)} \ket{\psi_1}|^2\right] \ge (1+\frac{T}{2^{n-1}})^{-|S_{11}|}.
\end{align}
Substituting this into Eq.~\eqref{eq:moments_fisher}, we have
\begin{align}
\begin{split}
    &\bE \prod_{i=1}^{T} 2^n \Tr(\ketbra{\psi_{u_{i+1}}^{u_i}}\rho) \\
    \ge & \sum_{S_0,S_{10},S_{11}} \bE \prod_{i \in S_0} |\alpha_{u_{i+1}}^{u_i}|^2 \prod_{i \in S_{10}} \frac{3}{4}|\beta_{u_{i+1}}^{u_i}|^2\\
    &\prod_{i \in S_{11}} \frac{1}{4}|\beta_{u_{i+1}}^{u_i}|^2 
    (1+\frac{T}{2^{n-1}})^{-T} \\
    = & (|\alpha_{u_{i+1}}^{u_i}|^2 + \frac{3}{4}|\beta_{u_{i+1}}^{u_i}|^2 + \frac{1}{4}|\beta_{u_{i+1}}^{u_i}|^2 ) (1+\frac{T}{2^{n-1}})^{-T} \\
    = &(1+\frac{T}{2^{n-1}})^{-T}.
    \end{split}
\end{align}

Combining with Eq.~\eqref{eq:distance_A_mixed2}, we have
\begin{align}
\begin{split}
    \TV(p_{\rho_m},\bE_{\mathcal{S}_1} p_{\rho}) &\le \sum_l p_{\rho_m}(l) \max\left(0, 1- (1+\frac{T}{2^{n-2}})^{-T} \right)\\
    &= 1 -  (1+\frac{T}{2^{n-1}})^{-T}.
\end{split}
\end{align}
\end{proof}

\subsubsection{With bounded quantum memory}

Here, we prove the hardness of estimating quantum Fisher information of low-rank quantum states with bounded memory.
\begin{theorem}[Exponential overhead of estimation Fisher information with bounded quantum memory without purification]\label{thm:fisher_lower_mem}
Given an unknown $n$-qubit low-rank state $\rho$ with $n\geq 2$ and its rank (even constant), any protocol that predicts $F_O(\rho)$ within constant additive error requires at least $\Omega(\min\{2^{n/2},2^{n-k}\})$ sample complexity with $k$ qubits of quantum memory.
\end{theorem}
\begin{proof}
We consider following the same proving procedure in Appendix~\ref{app:vc_pca_lower_mem} and the hard instance against single-copy protocols for estimating Fisher information (we rewrite it for convenience):
\begin{enumerate}
    \item {Ensemble $\mathcal{S}_1$: 
    \begin{align}
        \rho = \frac{7}{8}\cdot\left(\frac{4}{7} \ketbra{0}+\frac{3}{7}\ketbra{1}\right) \otimes U \ketbra{0}^{\otimes n-1} U^{\dagger} + \frac{1}{8}\ketbra{1} \otimes V \ketbra{0}^{\otimes n-1} V^{\dagger}
    \end{align}
    }
    \item {Ensemble $\mathcal{S}_2$:
        \begin{align}
            \rho = \frac{7}{8}\cdot\left(\frac{4}{7} \ketbra{+}+\frac{3}{7}\ketbra{-}\right) \otimes U \ketbra{0}^{\otimes n-1} U^{\dagger} + \frac{1}{8}\ketbra{-} \otimes V \ketbra{0}^{\otimes n-1} V^{\dagger}
        \end{align}
    }
\end{enumerate}
In both state ensembles, $U, V$ are Haar-random unitary on $n-1$ qubits. In the following, we write $\rho_{01}=\frac{4}{7} \ketbra{0}+\frac{3}{7}\ketbra{1}$ and $\rho_{+-}=\frac{4}{7} \ketbra{+}+\frac{3}{7}\ketbra{-}$ for convenience. Our task is to distinguish whether $\rho$ is chosen from state ensemble $\mathcal{S}_1$ or $\mathcal{S}_2$.

Similar to the argument in Appendix~\ref{app:purity_memory_lower}, we only need to compute the total variation distance
\begin{align}
\TV(\bE_Q p_\rho,p_{\rho_m}),
\end{align}
where $Q$ corresponding to the ensemble $S_Q$ that contains
\begin{align}
&q\rho_{01} \otimes (U \ketbra{0}^{\otimes n-1} U^{\dagger}) +\nonumber\\
&\quad (1-q)\ketbra{1} \otimes \ketbra{0} \otimes V \ketbra{0}^{\otimes n-2} V^{\dagger},
\end{align}
which are rank-$2$ random state with Haar randomly chosen $U$ and $V$, and some constant $0\leq q\leq1$. By symmetry, we can also bound
\begin{align}
\TV(\bE_{Q'} p_\rho,p_{\rho_m}),
\end{align}
where $Q'$ corresponding to the ensemble $S_{Q'}$ that contains
\begin{align}
&q\rho_{+-} \otimes (U \ketbra{0}^{\otimes n-1} U^{\dagger}) + \nonumber\\
&\quad (1-q)\ketbra{-} \otimes \ketbra{0} \otimes V \ketbra{0}^{\otimes n-2} V^{\dagger},
\end{align}
which are rank-$2$ random state with Haar randomly chosen $U$ and $V$, and some constant $0\leq q\leq1$. By triangle inequality and choosing $q_1=7/8$ and $q_2=1/8$, we can compute a bound for
\begin{align}
\TV(\bE_{\mathcal{S}_2} p_\rho,\bE_{\mathcal{S}_2} p_\rho),
\end{align}
which leads to a sample complexity lower bound for the distinguishing problem by Le Cam’s method~\cite{yu1997assouad}. We will follow then an exact procedure with Appendix~\ref{app:purity_memory_lower} except that the permutation operator $S_{T_{00},T_{01},T_{10}}^U$ and $S_{T_{01},T_{10},T_{11}}^V$ now acts on the last $n-1$ qubits of each copy. By using Lemma~\ref{cor:single_step_two_moment} instead of Lemma~\ref{lem:swap_second} as Appendix~\ref{app:vc_pca_lower_mem}, we can prove the desired theorem.
\end{proof}

\subsubsection{With classical correlation}\label{app:fisher_classical_corr}
Here, we prove the hardness of estimating the quantum Fisher information of low-rank quantum states, even with ancillary qubits that are classically correlated with these states.

\begin{theorem}[Exponential overhead of estimating Fisher information with classical correlation]\label{thm:fisher_lower_correlation}
Given an $n$-qubit low-rank state $\rho$, its rank $r$, observable $O$ satisfying $\norm{O}_{\infty} = 1$, and access to ancillary qubits that are classically correlated with the eigenvalues of $\rho$, the entire state can be written as
\begin{equation}
    \rho_C = \sum_{i=1}^r p_i \ketbra{\psi_i} \otimes \ketbra{i},
\end{equation}
where $\ket{\psi_i}$ are the eigenstates of $\rho$ for $1 \le i \le r$. 
Any protocol that predicts the quantum Fisher information $F_O(\rho)$ within constant additive error for arbitrary $\rho$ with a high probability requires at least $\Omega(2^{n/2})$ sample complexity. 
\end{theorem}
\begin{proof}
To prove the lower bound in the classical correlation case, we consider the following distinguishing task: A state $\rho$ is randomly drawn from one of the following two state ensembles $\mathcal{S}_1$ and $\mathcal{S}_2$ that only contains low-rank mixed states:
\begin{enumerate}
    \item {Ensemble $\mathcal{S}_1$: 
    \begin{equation}\label{eq:fisher_info_SA_2}
        \rho = \frac{1}{2} \ketbra{0} \otimes U \ketbra{0}^{\otimes n-2} U^{\dagger} + \frac{3}{8} \ketbra{1} \otimes U \ketbra{0}^{\otimes n-2} U^{\dagger} + \frac{1}{8}\ketbra{2} \otimes V \ketbra{0}^{\otimes n-2} V^{\dagger}
    \end{equation}
    }
    \item {Ensemble $\mathcal{S}_2$:
        \begin{equation}\label{eq:fisher_info_SB_2}
            \rho = \frac{1}{2} \ketbra{0} \otimes U \ketbra{0}^{\otimes n-2} U^{\dagger} + \frac{1}{8} \ketbra{1} \otimes V \ketbra{0}^{\otimes n-2} V^{\dagger} + \frac{1}{8}\ketbra{2} \otimes V \ketbra{0}^{\otimes n-2} V^{\dagger}
        \end{equation}
    }
\end{enumerate}
In both state ensembles, the first register is realized by the first two qubits and $U, V$ are Haar-random unitary on $n-2$ qubits.

We consider the observable 
\begin{equation}
    O = (\ketbra{0}{1} + \ketbra{1}{0}) \otimes \mathbb{I}^{\otimes (n-2)}.
\end{equation}
If $\rho \in \mathcal{S}_1$, then the three eigenstates are
\begin{equation}
    \ket{\psi_0} = \ket{0} \otimes U \ket{0}, \ket{\psi_1} = \ket{1} \otimes U \ket{0}, \ket{\psi_2} = \ket{2} \otimes V \ket{0}.
\end{equation}
In this scenario, the term $\abs{\bra{\psi_0} O \ket{\psi_1}}^2$ contributes non-negligibly to the Fisher information, making $F_O(\rho)$ a constant non-zero value. On the other hand, for $\rho \in \mathcal{S}_2$, the three eigenstates become
\begin{equation}
    \ket{\psi_0} = \ket{0} \otimes U \ket{0}, \ket{\psi_1} = \ket{1} \otimes V \ket{0}, \ket{\psi_2} = \ket{2} \otimes V \ket{0},
\end{equation}
and $U\ket{0}$ and $V \ket{0}$ are nearly orthogonal with high probability (by Lemma \ref{lem:innerproduct}). As a result, the Fisher information is exponentially small with probability exponentially close to 1.

Consequently, suppose we have an algorithm that, given an unknown state $\rho$ of rank at most 3 and an observable $O$ satisfying $\norm{O}_{\infty} \le \mathcal{O}(1)$, outputs $F_O(\rho)$ within constant additive error $\epsilon \le \mathcal{O}(1)$ with probability at least $5/6$. This algorithm can distinguish these two state ensembles with probability $5/6$ for sufficiently large $n$.

We now prove that even when ancillary qubits are classically correlated with the eigenstates of $\rho$, distinguishing between these two ensembles remains hard for any algorithm employing single-copy measurements. We consider the following ensemble:

\begin{enumerate}
    \item {Ensemble $\mathcal{CS}_1$: 
    \begin{equation}\label{eq:fisher_info_CSA}
        \rho = \frac{1}{2} \ketbra{00} \otimes U \ketbra{0}^{\otimes n-2} U^{\dagger} + \frac{3}{8} \ketbra{11} \otimes U \ketbra{0}^{\otimes n-2} U^{\dagger} + \frac{1}{8}\ketbra{22} \otimes V \ketbra{0}^{\otimes n-2} V^{\dagger}
    \end{equation}
    }
    \item {Ensemble $\mathcal{CS}_2$:
        \begin{equation}\label{eq:fisher_info_CSB}
            \rho = \frac{1}{2} \ketbra{00} \otimes U \ketbra{0}^{\otimes n-2} U^{\dagger} + \frac{3}{8} \ketbra{11} \otimes V \ketbra{0}^{\otimes n-2} V^{\dagger} + \frac{1}{8}\ketbra{22} \otimes V \ketbra{0}^{\otimes n-2} V^{\dagger}
        \end{equation}
    }
\end{enumerate}
Both ensembles include ancillary qubits (the second register in the first braket) that are classically correlated with the eigenstates of $\rho$ in $\mathcal{S_1}$ or $\mathcal{S_2}$
We can show that, on average, the total variation distance for the probability on leaf nodes of any algorithm with inputs from $\rho \in \mathcal{CS}_1$ and $\rho \in \mathcal{CS}_2$ can be small, which implies the algorithm has exponential sample complexity overhead. To show this, we consider the state 
\begin{equation}
    \rho_m = [\frac{1}{2}(\ketbra{00}+\frac{3}{8}\ketbra{11}+\frac{1}{8}\ketbra{22})] \otimes \mathbb{I}^{\otimes (n-2)}
\end{equation}
We prove that the two ensembles are close to $\rho_m$, as stated in the following lemma.
\begin{lemma}\label{lem:fisher_classical_correlation}
    \begin{equation}
        \TV(p_{\rho_m}, \bE_{\mathcal{CS}_1} p_{\rho}) \le 1 -  (1+\frac{T}{2^{n-2}})^{-T}
    \end{equation}
\end{lemma}

Similarly, we have $\TV(p_{\rho_m},\bE_{\mathcal{CS}_2} p_{\rho}) \le 1 -  (1+\frac{T}{2^{n-1}})^{-T}$. From triangular inequality, we obtain
\begin{equation}
\begin{split}
    \TV(\bE_{\mathcal{CS}_1} p_{\rho}, \bE_{\mathcal{CS}_2} p_{\rho}) &\le 2 -  2(1+\frac{T}{2^{n-1}})^{-T} \\
    &=2-2\exp(-T\ln(1+T/2^{n-1}))  \\
    &\le 2 - 2\exp(-T^2/2^{n-1}) \\
    &\le 2-2(1-T^2/2^{n-1}) = \frac{T^2}{2^{n-1}}
\end{split}
\end{equation}

By Le Cam's two-point method~\cite{yu1997assouad}, to successfully distinguish these two state ensembles with probability greater than $\frac{5}{6}$, the total variation distance must be greater than a constant, which means $T = \Omega(2^{n/2})$.
\end{proof}

\begin{proof}[Proof of Lemma \ref{lem:fisher_classical_correlation}]
Note that
\begin{equation}
\begin{aligned}
\TV(p_{\mathcal{S}_1}, p_{\mathcal{S}_2}) =& \frac{1}{2} \sum_l |p_{\mathcal{S}_1}(l) - p_{\mathcal{S}_2}(l)|\\
=& \sum_l \max(0, p_{\mathcal{S}_1}(l) - p_{\mathcal{S}_2}(l))\\
=& \sum_l p_{\mathcal{S}_1}(l) \max(0, 1 - \frac{p_{\mathcal{S}_2}(l)}{p_{\mathcal{S}_1}(l)})
\end{aligned}
\end{equation}
\begin{equation}\label{eq:distance_correlation_mixed2}
    \TV(p_{\rho_m}, \bE_{\mathcal{CS}_1} p_{\rho}) = \sum_l p_{\rho_m}(l) \max\left(0, 1-\frac{p_{\mathcal{CS}_1}}{p_{\rho_m}} \right) 
\end{equation}

Now we show how to bound the term $\frac{p_{\mathcal{CS}_1}}{p_{\rho_m}}$. We begin by decomposing $\ket{\psi_{u_{i+1}}^{u_i}}$ into
\begin{equation}
\begin{aligned}
\ket{\psi_{u_{i+1}}^{u_i}} =& \alpha_{u_{i+1}}^{u_i} \ket{00} \otimes \ket{\phi_{u_{i+1}}^{u_i}(0)} + \beta_{u_{i+1}}^{u_i} \ket{11} \otimes \ket{\phi_{u_{i+1}}^{u_i}(1)}\\
&+ \gamma_{u_{i+1}}^{u_i} \ket{22} \otimes \ket{\phi_{u_{i+1}}^{u_i}(2)} + \ket{\perp}
\end{aligned}
\end{equation}
where $|\alpha_{u_{i+1}}^{u_i}|^2 + |\beta_{u_{i+1}}^{u_i}|^2 + |\gamma_{u_{i+1}}^{u_i}|^2\le 1$, and $\braket{ii}{\perp} = 0, 0 \le i \le 2$. A state $\rho \in \mathcal{CS}_1$ is denoted as 
\begin{equation}
\begin{aligned}
\rho =& \frac{1}{2} \ketbra{00} \otimes \ketbra{\psi_0} + \frac{3}{8} \ketbra{11} \otimes \ketbra{\psi_0}\\ 
&+ \frac{1}{8}\ketbra{22} \otimes \ketbra{\psi_{1}}
\end{aligned}
\end{equation}
where $\psi_0 = U\ket{0}$ and $\psi_1 = V\ket{0}$. Then we have
\begin{equation}\label{eq:moments_fisher_correlation}
\begin{split}
    &p_{\mathcal{CS}_1}(l)\\
    = & \prod_{i=1}^T \Tr(\ketbra{\psi_{u_{i+1}}^{u_i}} \rho )\\
    &\bE \prod_{i=1}^{T} (\frac{1}{2}|\alpha_{u_{i+1}}^{u_i}|^2 |\bra{\phi_{u_{i+1}}^{u_i}(0)} \ket{\psi_0}|^2 + \frac{3}{8}|\beta_{u_{i+1}}^{u_i}|^2 |\bra{\phi_{u_{i+1}}^{u_i}(1)} \ket{\psi_{0}}|^2 \\
    &+ \frac{1}{8}|\gamma_{u_{i+1}}^{u_i}|^2 |\bra{\phi_{u_{i+1}}^{u_i}(2)} \ket{\psi_{1}}|^2)\\
    = & \sum_{S_0,S_{1},S_{2}} \bE \prod_{i \in S_0} \frac{1}{2}|\alpha_{u_{i+1}}^{u_i}|^2 |\bra{\phi_{u_{i+1}}^{u_i}(0)} \ket{\psi_0}|^2 \\
    &\prod_{i \in S_{1}}\frac{3}{8}|\beta_{u_{i+1}}^{u_i}|^2 |\bra{\phi_{u_{i+1}}^{u_i}(1)} \ket{\psi_{0}}|^2 \prod_{i \in S_{2}} \frac{1}{8}|\gamma_{u_{i+1}}^{u_i}|^2 |\bra{\phi_{u_{i+1}}^{u_i}(1)} \ket{\psi_{1}}|^2 
    \end{split}
\end{equation}
where $\sum_{S_0,S_{1},S_{2}}$ denotes $\sum_{S_0,S_{1},S_{2}: \cup S_i = [T], S_i \cap S_j=\emptyset}$.  From Lemma \ref{lem:high_moment_haar}, we have
\begin{equation}
\begin{split}
    &\bE \left[ 2^{(n-2)(|S_0| + |S_{1}|)}\prod_{i \in S_0} |\bra{\phi_{u_{i+1}}^{u_i}(0)} \ket{\psi_0}|^2 \prod_{i \in S_{1}} |\bra{\phi_{u_{i+1}}^{u_i}(1)} \ket{\psi_{0}}|^2\right] \\
    \ge& [ (1 + \frac{|S_0|+|S_{1}|-1}{2^{n-2}}) \cdots (1+\frac{1}{2^{n-2}})]^{-1}  \\
    \ge& (1+\frac{T}{2^{n-2}})^{-(|S_{0}|+|S_1|)} 
\end{split}
\end{equation}
Similarly, we have:
\begin{equation}
    \bE \left[ 2^{(n-2)|S_{2}|}\prod_{i \in S_{11}} |\bra{\phi_{u_{i+1}}^{u_i}(1)} \ket{\psi_1}|^2\right] \ge (1+\frac{T}{2^{n-2}})^{-|S_{2}|} 
\end{equation}
Substituting this into Eq.~\eqref{eq:moments_fisher_correlation}, we have

\begin{equation}
\begin{split}
    &p_{\mathcal{CS}_1}(l)\\
    \ge& 2^{-(n-2)T} (1+\frac{T}{2^{n-2}})^{-T}\\
    &\sum_{S_0,S_{1},S_{2}} \prod_{i \in S_0} \frac{1}{2}|\alpha_{u_{i+1}}^{u_i}|^2  \prod_{i \in S_{1}} \frac{3}{8}|\beta_{u_{i+1}}^{u_i}|^2  \prod_{i \in S_{2}} \frac{1}{8}|\gamma_{u_{i+1}}^{u_i}|^2 
    \\
    =&  2^{-(n-2)T} (1+\frac{T}{2^{n-2}})^{-T} \prod_{i=1}^T (\frac{1}{2}|\alpha_{u_{i+1}}^{u_i}|^2  + \frac{3}{8}|\beta_{u_{i+1}}^{u_i}|^2  + |\gamma_{u_{i+1}}^{u_i}|^2)
    \end{split}
\end{equation}

For $\rho_m$, we have that
\begin{equation}
\begin{split}
    &p_{\rho_m}(l)\\
    = &\bE \prod_{i=1}^{T} (\frac{1}{2} |\alpha_{u_{i+1}}^{u_i}|^2 |\bra{\psi_0} \frac{\mathbb{I}}{2^{n-2}} \ket{\psi_0}|^2\\
    &+ \frac{3}{8}|\beta_{u_{i+1}}^{u_i}|^2 |\bra{\psi_1} \frac{\mathbb{I}}{2^{n-2}} \ket{\psi_1}|^2  + \frac{1}{8} |\gamma_{u_{i+1}}^{u_i}|^2 |\bra{\psi_2} \frac{\mathbb{I}}{2^{n-2}} \ket{\psi_2}|^2 ) \\
    &= \prod_{i=1}^T2^{-(n-2)}(\frac{1}{2}|\alpha_{u_{i+1}}^{u_i}|^2 + \frac{3}{8}|\beta_{u_{i+1}}^{u_i}|^2 + \frac{1}{8}|\gamma_{u_{i+1}}^{u_i}|^2)
    \end{split}
\end{equation}
Then, we can obtain
\begin{equation}
\begin{split}
    &\frac{p_{\mathcal{CS}_1}(l)}{p_{\rho_m}(l)} = (1 + \frac{T}{2^{n-2}})^{-T}
    \end{split}
\end{equation}

Combining with Eq.~\eqref{eq:distance_correlation_mixed2}, we have
\begin{equation}
\begin{split}
    \TV(p_{\rho_m},\bE_{\mathcal{CS}_1} p_{\rho}) &\le \sum_l p_{\rho_m}(l) \max\left(0, 1- (1+\frac{T}{2^{n-2}})^{-T} \right)\\
    &= 1 -  (1+\frac{T}{2^{n-2}})^{-T}
\end{split}
\end{equation}
\end{proof}

\section{Channel Learning}\label{sec:channel_learning}
After showing that purification can help to reduce the sample complexity in many quantum state learning tasks, it is natural to wonder whether this can be generalized to channel learning tasks. 
We have listed a range of scenarios where one might have access to the purification $\Psi_{AB}$ of the target mixed state $\rho_A$ in Sec.~\ref{sec:intro}. 
In some cases, the target mixed state and its purification are obtained by putting the input \emph{pure} state $\phi_A$ through some channel and its isometry, respectively. 
Hence, we will also have access to both channel and isometry in these cases.
If we are interested in applying the channel to a range of input states beyond $\phi_A$, then we will be interested in the property of the channel instead of the output mixed states. 
This can be efficiently probed by studying the isometry, which is actually the purification of $\mathcal{E}_A$, as will be discussed below.

A quantum channel represents the interaction between the system and the environment.
Thus, the purification of a quantum channel should contain additional information of the interaction.
According to the Stinespring's dilation theorem, the action of an arbitrary channel can be represented as an isometry, $V_{\mathcal{E}}:\mathcal{H}_A\to\mathcal{H}_{AB}$, acting on the input state and tracing out part of the output system, $\mathcal{E}(\rho_A)=\Tr_B\left(V_{\mathcal{E}}\rho_AV_{\mathcal{E}}^\dagger\right)$. 
We thus define the purification of the channel $\mathcal{E}$ to be the corresponding isometry $V_{\mathcal{E}}$. 
The lowest dimension needed for the trace-out system $B$ is given by the rank of the Choi state $\rho_{\mathcal{E}}$ of the channel $\mathcal{E}$. 
Hence, as long as $\rho_{\mathcal{E}}$ has a constant rank, we can again have a system $B$ consisting a constant number of qubits like in Sec.~\ref{sec:upper_bound}. 
Under this context, we will show that if we have access to the channel purification $V_{\mathcal{E}}$, many channel learning tasks can be efficiently achieved with constant sample complexity and single-copy operations. 
Without loss of generality, hereafter we only consider the case where the input and output of the target channel have the same dimension.

Such tasks can be identified by looking for the channel analogy for the tasks in Sec.~\ref{sec:upper_bound}. One might want to apply the methods in Sec.~\ref{sec:upper_bound} directly to the Choi state $\rho_{\mathcal{E}}$ to gain the necessary advantage. 
However, preparing the Choi state of an $n$-qubit quantum channel requires $2n$-qubit quantum memory, which is much more than the $n + \order{1}$ qubit overhead needed in Sec.~\ref{sec:upper_bound}.
We will illustrate alternative ways to accomplish these channel learning tasks with the same $n + \order{1}$ qubit overhead in the following subsections.

\subsection{Unitarity Estimation}

\begin{figure}[htbp]
\centering
\includegraphics[width=0.7\linewidth]{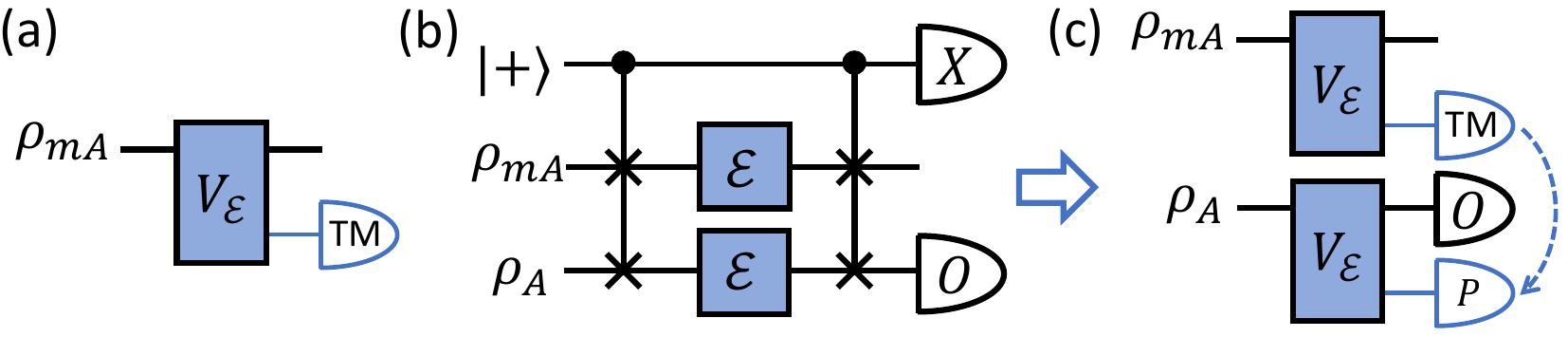}
\caption{Circuits for some channel learning tasks. (a) Purification-assisted unitarity estimation circuit. (b) The original circuit for virtual channel distillation. (c) Purification-assisted circuits for virtual channel distillation and channel principal component analysis. Here, $\rho_{mA}$ is the maximally mixed state, TM stands for tomography, and $P$ is some observable constructed based on the results of the tomography. }
\label{fig:qem}
\end{figure}
The first task we are looking at is the unitarity estimation~\cite{montanaro2013survey}.
Unitarity is defined as the purity of the Choi state of the target channel, quantifying the distance between the given channel and unitary evolution. 
It has been proven that, without the prior knowledge of the target $d_A$-dimensional channel, arbitrary protocol without quantum memory that can accurately estimate the unitarity requires $\Omega(\sqrt{d_A})$ sample complexity~\cite{chen2023unitarity}. 
With access to $V_{\mathcal{E}}$, we find that the unitarity of the channel can be efficiently estimated using the circuit in Fig.~\ref{fig:qem}(a), which is explained by the following observation:
\begin{observation}
Given a quantum channel $\mathcal{E}$ and the isometry $V_{\mathcal{E}}:\mathcal{H}_A\to\mathcal{H}_{AB}$ satisfying $\mathcal{E}(\rho_A)=\Tr_B\left(V_{\mathcal{E}}\rho_AV_{\mathcal{E}}^\dagger\right)$ for arbitrary $\rho_A$, the unitarity of the channel $\mathcal{E}$, which is also the purity of its Choi state $\rho_{\mathcal{E}}$, is given as $\Tr(\rho_{\mathcal{E}}^2)=\Tr_B\left[\Tr_A\left(V_{\mathcal{E}}\rho_{mA}V_{\mathcal{E}}^\dagger\right)^2\right]$, where $\rho_m$ is the maximally mixed state.
\end{observation}
\begin{proof}
Any given channel $\mathcal{E}$ acting on the system $A$ can always be written in the canonical Kraus representation:
\begin{align}\label{eqn:canonical_kraus}
\mathcal{E}(\rho) = \sum_{i=0}^{r-1} p_i E_{i}\rho E_{i}^\dagger
\end{align}
where $r$ is the Choi rank of the channel. The orthonormal Kraus basis here satisfies $\Tr(E_{i}^\dagger E_{j}) = \delta_{ij}d_A$, and the trace-preserving condition implies that $\sum_i p_i = 1$. Without loss of generality, we will assume all Kraus operators are ordered in descending order of $p_i$, i.e. $p_{0} \geq p_{1} \cdots p_{r-1} \geq 0$.  

The corresponding isometry of the channel can always be written in the form:
\begin{align}
V_{\mathcal{E}} = \sum_{i=0}^{r-1} \sqrt{p_i} E_{i} \otimes \ket{\psi_B^i}
\end{align}
where $\{\ket{\psi_B^i}\}$ is a set of orthonormal basis of the purification system $B$.  It is easy to verify that indeed $\mathcal{E}(\rho)=\Tr_B\left(V_{\mathcal{E}}\rho V_{\mathcal{E}}^\dagger\right)$.

Using this representation, if we input a maximally mixed state in system $A$ into the isometry and trace out system $A$, the resultant state is:
\begin{align}\label{eqn:rho_B}
\begin{split}
\rho_B &= \Tr_A\left(V_{\mathcal{E}}\rho_{mA}V_{\mathcal{E}}^\dagger\right)\\
&=   \frac{1}{d_A}\sum_{ij} \sqrt{p_ip_j} \Tr_A\left(E_{i}E_{j}^{\dagger}\right) \otimes \ketbra{\psi_B^i}{\psi_B^j} \\
&= \sum_{i=0}^{r-1} p_i \ketbra{\psi_B^i}{\psi_B^i}
\end{split}
\end{align}
In this way, the unitarity of the channel $A$, which is $\Tr(\rho_{\mathcal{E}}^2)=\sum_{i} p_i^2$, can be obtained via
\begin{align}
\sum_{i = 0}^{r-1}p_i^2 = \Tr_B[\rho_B^2].
\end{align}
\end{proof}

This proof can be generalized to prove observations in the following two sections.
According to this observation, the estimation of unitarity is equivalent to the purity estimation of the system $B$ after the isometry evolution in Fig.~\ref{fig:qem}(a). Since the qubit number of system $B$ is assumed to be a constant, we can obtain its purity via state tomography with a constant number of samples. 

\subsection{Virtual Channel Distillation}

The task of virtual channel distillation aims to estimate the value of $\Tr\left[O\mathcal{E}^{(2)}(\rho_A)\right]$ for arbitrary input state $\rho_A$  and observable $O$, with $\mathcal{E}^{(2)}$ being the (unnormalized) distilled channel whose Choi state $\rho_{\mathcal{E}^{(2)}}$ is the square of the Choi state of the original one, $\rho_{\mathcal{E}^{(2)}}=(\rho_{\mathcal{E}})^2$. The ability to perform such channel distillation enables the application of quantum error mitigation techniques~\cite{liu2024virtual} that require a minimal set of assumptions and achieve an exponential error suppression rate. As shown in Fig.~\ref{fig:qem}(b), the original circuit implementation of virtual channel distillation has removed the need for preparing the Choi states. However, it still requires $2n + \order{1}$ memory qubits to implement. Furthermore, it requires controlled-SWAP between the two registers, which means long-range connections are needed for swapping the corresponding qubits in different registers.   
We will show that the accessibility of purification removes all of these hardware requirements based on the following observation:
\begin{observation}
Given a quantum channel $\mathcal{E}$ and the isometry $V_{\mathcal{E}}$ satisfying $\mathcal{E}(\rho_A)=\Tr_B\left(V_{\mathcal{E}}\rho_AV_{\mathcal{E}}^\dagger\right)$ for arbitrary $\rho_A$, we have $\mathcal{E}^{(2)}(\rho_A)=\Tr_B\left[V_{\mathcal{E}}\rho_AV_{\mathcal{E}}^\dagger\Tr_A\left(V_{\mathcal{E}}\rho_{mA}V_{\mathcal{E}}^\dagger\right)\right]$ for arbitrary $\rho_A$, where $\rho_m$ is the maximally mixed state.
\end{observation}
\noindent This observation can be proved using Eq.~\eqref{eqn:rho_B}, as
\begin{equation}
\Tr_B\left[\rho_BV_{\mathcal{E}}\rho_AV_{\mathcal{E}}^\dagger\right] = \sum_{i=0}^{r-1} p_i^2 E_{i}\rho_A E_{i}^\dagger.
\end{equation}
Here, the Choi state of channel $\sum_{i=0}^{r-1} p_i^2 E_{i}\rho_A E_{i}^\dagger$ is $\rho_{\mathcal{E}^{(2)}}=(\rho_{\mathcal{E}})^2$.

Based on this observation, we can design a purification-assisted protocol to estimate the value of $\Tr\left[O\mathcal{E}^{(2)}(\rho_A)\right]$ using the circuit shown in Fig.~\ref{fig:qem}(c).
We first initialize the input state to be the maximally mixed state and perform tomography on the output state in system $B$, which only contains a constant number of qubits.
Based on the result of tomography, we classically reconstruct an estimator of the density matrix $\Tr_A\left(V_{\mathcal{E}}\rho_{Am}V_{\mathcal{E}}^\dagger\right)$, labeled as $\hat{\rho}_B$.
Then, we change the input state to be $\rho_A$ and measure the observable $O\otimes\hat{\rho}_B$ on the state $V_{\mathcal{E}}\rho_A V_{\mathcal{E}}^\dagger$ to get the final estimator of $\Tr\left[O\mathcal{E}^{(2)}(\rho_A)\right]$. 
Following similar arguments in quantum virtual cooling, we can easily prove that this protocol only requires constant sample complexity when $\norm{O}_{\infty}=\mathcal{O}(1)$.

\subsection{Channel Principal Component Extraction}
Using a protocol modified from the last section, we can also perform the principal component analysis for a quantum channel. 
The principle component of a given channel is defined as the Kraus component that corresponds to the leading eigenvector of its Choi state. 
More explicitly, any given channel can always be written in the canonical Kraus representation $\mathcal{E}(\rho) = \sum_j p_j E_{j}\rho E_{j}^\dagger$ with orthonormal Kraus operators $\Tr(E_{j}^\dagger E_{k}) = \delta_{j,k}d_A$ and $\sum_{j}p_j = 1$. 
Without loss of generality, we will assume $p_0$ is the largest out of all $p_i$. 
The target of channel principal component analysis is to extract the value of $\Tr\left[OE_0\rho_AE_0^\dagger\right]$, where $\rho_A$ is an arbitrary input state. 
This can be useful in practice when $\mathcal{E}$ is some noise channel and $E_0$ is the target noiseless component. 
Given the isometry $V_{\mathcal{E}}$, we have the following observation:
\begin{observation}
Given a channel $\mathcal{E}$, which has the canonical Kraus decomposition $\mathcal{E}(\cdot)=\sum_jp_jE_j\cdot E_j^\dagger$ with $\Tr(E_jE_k^\dagger)=d_A\delta_{j,k}$, and the isometry $V_{\mathcal{E}}$ satisfying $\mathcal{E}(\rho_A)=\Tr_B\left(V_{\mathcal{E}}\rho_AV_{\mathcal{E}}^\dagger\right)$ for arbitrary $\rho_A$, we have $E_0\rho_AE_0^\dagger= \frac{1}{p_0}\Tr_B\left(V_{\mathcal{E}}\rho_AV_{\mathcal{E}}^\dagger\ketbra{\psi_B^0}\right)$ for arbitrary $\rho_A$, 
where $\ket{\psi_B^0}$ is the principal eigenstate of $\Tr_A(V_{\mathcal{E}}\rho_{mA}V_{\mathcal{E}}^\dagger)$.
\end{observation}
\noindent Based on this observation, we can use the same circuit and protocol as virtual channel distillation in the last section to estimate the value of $\Tr\left[OE_0\rho_AE_0^\dagger\right]$.
The only adjustment is that the observable we measure on subsystem $B$ should be changed from $\hat{\rho}_B$ to an estimator of the principal component of $\Tr_A(V_{\mathcal{E}}\rho_AV_{\mathcal{E}}^\dagger)$.
Following similar arguments in quantum state principal component analysis, we can easily prove that this protocol also requires constant sample complexity when $\norm{O}_{\infty}=\mathcal{O}(1)$. 


\begin{thebibliography}{102}%
	\makeatletter
	\providecommand \@ifxundefined [1]{%
		\@ifx{#1\undefined}
	}%
	\providecommand \@ifnum [1]{%
		\ifnum #1\expandafter \@firstoftwo
		\else \expandafter \@secondoftwo
		\fi
	}%
	\providecommand \@ifx [1]{%
		\ifx #1\expandafter \@firstoftwo
		\else \expandafter \@secondoftwo
		\fi
	}%
	\providecommand \natexlab [1]{#1}%
	\providecommand \enquote  [1]{``#1''}%
	\providecommand \bibnamefont  [1]{#1}%
	\providecommand \bibfnamefont [1]{#1}%
	\providecommand \citenamefont [1]{#1}%
	\providecommand \href@noop [0]{\@secondoftwo}%
	\providecommand \href [0]{\begingroup \@sanitize@url \@href}%
	\providecommand \@href[1]{\@@startlink{#1}\@@href}%
	\providecommand \@@href[1]{\endgroup#1\@@endlink}%
	\providecommand \@sanitize@url [0]{\catcode `\\12\catcode `\$12\catcode
		`\&12\catcode `\#12\catcode `\^12\catcode `\_12\catcode `\%12\relax}%
	\providecommand \@@startlink[1]{}%
	\providecommand \@@endlink[0]{}%
	\providecommand \url  [0]{\begingroup\@sanitize@url \@url }%
	\providecommand \@url [1]{\endgroup\@href {#1}{\urlprefix }}%
	\providecommand \urlprefix  [0]{URL }%
	\providecommand \Eprint [0]{\href }%
	\providecommand \doibase [0]{https://doi.org/}%
	\providecommand \selectlanguage [0]{\@gobble}%
	\providecommand \bibinfo  [0]{\@secondoftwo}%
	\providecommand \bibfield  [0]{\@secondoftwo}%
	\providecommand \translation [1]{[#1]}%
	\providecommand \BibitemOpen [0]{}%
	\providecommand \bibitemStop [0]{}%
	\providecommand \bibitemNoStop [0]{.\EOS\space}%
	\providecommand \EOS [0]{\spacefactor3000\relax}%
	\providecommand \BibitemShut  [1]{\csname bibitem#1\endcsname}%
	\let\auto@bib@innerbib\@empty
	\bibitem [{\citenamefont {Eisert}\ \emph {et~al.}(2020)\citenamefont {Eisert},
		\citenamefont {Hangleiter}, \citenamefont {Walk}, \citenamefont {Roth},
		\citenamefont {Markham}, \citenamefont {Parekh}, \citenamefont {Chabaud},\
		and\ \citenamefont {Kashefi}}]{eisert2020quantum}%
	\BibitemOpen
	\bibfield  {author} {\bibinfo {author} {\bibfnamefont {J.}~\bibnamefont
			{Eisert}}, \bibinfo {author} {\bibfnamefont {D.}~\bibnamefont {Hangleiter}},
		\bibinfo {author} {\bibfnamefont {N.}~\bibnamefont {Walk}}, \bibinfo {author}
		{\bibfnamefont {I.}~\bibnamefont {Roth}}, \bibinfo {author} {\bibfnamefont
			{D.}~\bibnamefont {Markham}}, \bibinfo {author} {\bibfnamefont
			{R.}~\bibnamefont {Parekh}}, \bibinfo {author} {\bibfnamefont
			{U.}~\bibnamefont {Chabaud}},\ and\ \bibinfo {author} {\bibfnamefont
			{E.}~\bibnamefont {Kashefi}},\ }\bibfield  {title} {\bibinfo {title} {Quantum
			certification and benchmarking},\ }\href
	{https://doi.org/10.1038/s42254-020-0186-4} {\bibfield  {journal} {\bibinfo
			{journal} {Nat. Rev. Phys.}\ }\textbf {\bibinfo {volume} {2}},\ \bibinfo
		{pages} {382} (\bibinfo {year} {2020})}\BibitemShut {NoStop}%
	\bibitem [{\citenamefont {Daley}\ \emph {et~al.}(2022)\citenamefont {Daley},
		\citenamefont {Bloch}, \citenamefont {Kokail}, \citenamefont {Flannigan},
		\citenamefont {Pearson}, \citenamefont {Troyer},\ and\ \citenamefont
		{Zoller}}]{Daley2022analog}%
	\BibitemOpen
	\bibfield  {author} {\bibinfo {author} {\bibfnamefont {A.~J.}\ \bibnamefont
			{Daley}}, \bibinfo {author} {\bibfnamefont {I.}~\bibnamefont {Bloch}},
		\bibinfo {author} {\bibfnamefont {C.}~\bibnamefont {Kokail}}, \bibinfo
		{author} {\bibfnamefont {S.}~\bibnamefont {Flannigan}}, \bibinfo {author}
		{\bibfnamefont {N.}~\bibnamefont {Pearson}}, \bibinfo {author} {\bibfnamefont
			{M.}~\bibnamefont {Troyer}},\ and\ \bibinfo {author} {\bibfnamefont
			{P.}~\bibnamefont {Zoller}},\ }\bibfield  {title} {\bibinfo {title}
		{Practical quantum advantage in quantum simulation},\ }\href
	{https://doi.org/10.1038/s41586-022-04940-6} {\bibfield  {journal} {\bibinfo
			{journal} {Nature}\ }\textbf {\bibinfo {volume} {607}},\ \bibinfo {pages}
		{667} (\bibinfo {year} {2022})}\BibitemShut {NoStop}%
	\bibitem [{\citenamefont {Elben}\ \emph {et~al.}(2023)\citenamefont {Elben},
		\citenamefont {Flammia}, \citenamefont {Huang}, \citenamefont {Kueng},
		\citenamefont {Preskill}, \citenamefont {Vermersch},\ and\ \citenamefont
		{Zoller}}]{elben2023randomized}%
	\BibitemOpen
	\bibfield  {author} {\bibinfo {author} {\bibfnamefont {A.}~\bibnamefont
			{Elben}}, \bibinfo {author} {\bibfnamefont {S.~T.}\ \bibnamefont {Flammia}},
		\bibinfo {author} {\bibfnamefont {H.-Y.}\ \bibnamefont {Huang}}, \bibinfo
		{author} {\bibfnamefont {R.}~\bibnamefont {Kueng}}, \bibinfo {author}
		{\bibfnamefont {J.}~\bibnamefont {Preskill}}, \bibinfo {author}
		{\bibfnamefont {B.}~\bibnamefont {Vermersch}},\ and\ \bibinfo {author}
		{\bibfnamefont {P.}~\bibnamefont {Zoller}},\ }\bibfield  {title} {\bibinfo
		{title} {The randomized measurement toolbox},\ }\href
	{https://doi.org/10.1038/s42254-022-00535-2} {\bibfield  {journal} {\bibinfo
			{journal} {Nat. Rev. Phys.}\ }\textbf {\bibinfo {volume} {5}},\ \bibinfo
		{pages} {9} (\bibinfo {year} {2023})}\BibitemShut {NoStop}%
	\bibitem [{\citenamefont {Ludlow}\ \emph {et~al.}(2015)\citenamefont {Ludlow},
		\citenamefont {Boyd}, \citenamefont {Ye}, \citenamefont {Peik},\ and\
		\citenamefont {Schmidt}}]{ludlow2015clock}%
	\BibitemOpen
	\bibfield  {author} {\bibinfo {author} {\bibfnamefont {A.~D.}\ \bibnamefont
			{Ludlow}}, \bibinfo {author} {\bibfnamefont {M.~M.}\ \bibnamefont {Boyd}},
		\bibinfo {author} {\bibfnamefont {J.}~\bibnamefont {Ye}}, \bibinfo {author}
		{\bibfnamefont {E.}~\bibnamefont {Peik}},\ and\ \bibinfo {author}
		{\bibfnamefont {P.~O.}\ \bibnamefont {Schmidt}},\ }\bibfield  {title}
	{\bibinfo {title} {Optical atomic clocks},\ }\href
	{https://doi.org/10.1103/RevModPhys.87.637} {\bibfield  {journal} {\bibinfo
			{journal} {Rev. Mod. Phys.}\ }\textbf {\bibinfo {volume} {87}},\ \bibinfo
		{pages} {637} (\bibinfo {year} {2015})}\BibitemShut {NoStop}%
	\bibitem [{\citenamefont {Wang}\ \emph {et~al.}(2017)\citenamefont {Wang},
		\citenamefont {Um}, \citenamefont {Zhang}, \citenamefont {An}, \citenamefont
		{Lyu}, \citenamefont {Zhang}, \citenamefont {Duan}, \citenamefont {Yum},\
		and\ \citenamefont {Kim}}]{Wang2017memory}%
	\BibitemOpen
	\bibfield  {author} {\bibinfo {author} {\bibfnamefont {Y.}~\bibnamefont
			{Wang}}, \bibinfo {author} {\bibfnamefont {M.}~\bibnamefont {Um}}, \bibinfo
		{author} {\bibfnamefont {J.}~\bibnamefont {Zhang}}, \bibinfo {author}
		{\bibfnamefont {S.}~\bibnamefont {An}}, \bibinfo {author} {\bibfnamefont
			{M.}~\bibnamefont {Lyu}}, \bibinfo {author} {\bibfnamefont {J.-N.}\
			\bibnamefont {Zhang}}, \bibinfo {author} {\bibfnamefont {L.-M.}\ \bibnamefont
			{Duan}}, \bibinfo {author} {\bibfnamefont {D.}~\bibnamefont {Yum}},\ and\
		\bibinfo {author} {\bibfnamefont {K.}~\bibnamefont {Kim}},\ }\bibfield
	{title} {\bibinfo {title} {Single-qubit quantum memory exceeding ten-minute
			coherence time},\ }\href {https://doi.org/10.1038/s41566-017-0007-1}
	{\bibfield  {journal} {\bibinfo  {journal} {Nat. Photonics}\ }\textbf
		{\bibinfo {volume} {11}},\ \bibinfo {pages} {646} (\bibinfo {year}
		{2017})}\BibitemShut {NoStop}%
	\bibitem [{\citenamefont {Arute}\ \emph {et~al.}(2019)\citenamefont {Arute},
		\citenamefont {Arya}, \citenamefont {Babbush}, \citenamefont {Bacon},
		\citenamefont {Bardin}, \citenamefont {Barends}, \citenamefont {Biswas},
		\citenamefont {Boixo}, \citenamefont {Brandao}, \citenamefont {Buell},
		\citenamefont {Burkett}, \citenamefont {Chen}, \citenamefont {Chen},
		\citenamefont {Chiaro}, \citenamefont {Collins}, \citenamefont {Courtney},
		\citenamefont {Dunsworth}, \citenamefont {Farhi}, \citenamefont {Foxen},
		\citenamefont {Fowler}, \citenamefont {Gidney}, \citenamefont {Giustina},
		\citenamefont {Graff}, \citenamefont {Guerin}, \citenamefont {Habegger},
		\citenamefont {Harrigan}, \citenamefont {Hartmann}, \citenamefont {Ho},
		\citenamefont {Hoffmann}, \citenamefont {Huang}, \citenamefont {Humble},
		\citenamefont {Isakov}, \citenamefont {Jeffrey}, \citenamefont {Jiang},
		\citenamefont {Kafri}, \citenamefont {Kechedzhi}, \citenamefont {Kelly},
		\citenamefont {Klimov}, \citenamefont {Knysh}, \citenamefont {Korotkov},
		\citenamefont {Kostritsa}, \citenamefont {Landhuis}, \citenamefont
		{Lindmark}, \citenamefont {Lucero}, \citenamefont {Lyakh}, \citenamefont
		{Mandr{\`a}}, \citenamefont {McClean}, \citenamefont {McEwen}, \citenamefont
		{Megrant}, \citenamefont {Mi}, \citenamefont {Michielsen}, \citenamefont
		{Mohseni}, \citenamefont {Mutus}, \citenamefont {Naaman}, \citenamefont
		{Neeley}, \citenamefont {Neill}, \citenamefont {Niu}, \citenamefont {Ostby},
		\citenamefont {Petukhov}, \citenamefont {Platt}, \citenamefont {Quintana},
		\citenamefont {Rieffel}, \citenamefont {Roushan}, \citenamefont {Rubin},
		\citenamefont {Sank}, \citenamefont {Satzinger}, \citenamefont {Smelyanskiy},
		\citenamefont {Sung}, \citenamefont {Trevithick}, \citenamefont
		{Vainsencher}, \citenamefont {Villalonga}, \citenamefont {White},
		\citenamefont {Yao}, \citenamefont {Yeh}, \citenamefont {Zalcman},
		\citenamefont {Neven},\ and\ \citenamefont {Martinis}}]{Arute2019supermacy}%
	\BibitemOpen
	\bibfield  {author} {\bibinfo {author} {\bibfnamefont {F.}~\bibnamefont
			{Arute}}, \bibinfo {author} {\bibfnamefont {K.}~\bibnamefont {Arya}},
		\bibinfo {author} {\bibfnamefont {R.}~\bibnamefont {Babbush}}, \bibinfo
		{author} {\bibfnamefont {D.}~\bibnamefont {Bacon}}, \bibinfo {author}
		{\bibfnamefont {J.~C.}\ \bibnamefont {Bardin}}, \bibinfo {author}
		{\bibfnamefont {R.}~\bibnamefont {Barends}}, \bibinfo {author} {\bibfnamefont
			{R.}~\bibnamefont {Biswas}}, \bibinfo {author} {\bibfnamefont
			{S.}~\bibnamefont {Boixo}}, \bibinfo {author} {\bibfnamefont {F.~G. S.~L.}\
			\bibnamefont {Brandao}}, \bibinfo {author} {\bibfnamefont {D.~A.}\
			\bibnamefont {Buell}}, \bibinfo {author} {\bibfnamefont {B.}~\bibnamefont
			{Burkett}}, \bibinfo {author} {\bibfnamefont {Y.}~\bibnamefont {Chen}},
		\bibinfo {author} {\bibfnamefont {Z.}~\bibnamefont {Chen}}, \bibinfo {author}
		{\bibfnamefont {B.}~\bibnamefont {Chiaro}}, \bibinfo {author} {\bibfnamefont
			{R.}~\bibnamefont {Collins}}, \bibinfo {author} {\bibfnamefont
			{W.}~\bibnamefont {Courtney}}, \bibinfo {author} {\bibfnamefont
			{A.}~\bibnamefont {Dunsworth}}, \bibinfo {author} {\bibfnamefont
			{E.}~\bibnamefont {Farhi}}, \bibinfo {author} {\bibfnamefont
			{B.}~\bibnamefont {Foxen}}, \bibinfo {author} {\bibfnamefont
			{A.}~\bibnamefont {Fowler}}, \bibinfo {author} {\bibfnamefont
			{C.}~\bibnamefont {Gidney}}, \bibinfo {author} {\bibfnamefont
			{M.}~\bibnamefont {Giustina}}, \bibinfo {author} {\bibfnamefont
			{R.}~\bibnamefont {Graff}}, \bibinfo {author} {\bibfnamefont
			{K.}~\bibnamefont {Guerin}}, \bibinfo {author} {\bibfnamefont
			{S.}~\bibnamefont {Habegger}}, \bibinfo {author} {\bibfnamefont {M.~P.}\
			\bibnamefont {Harrigan}}, \bibinfo {author} {\bibfnamefont {M.~J.}\
			\bibnamefont {Hartmann}}, \bibinfo {author} {\bibfnamefont {A.}~\bibnamefont
			{Ho}}, \bibinfo {author} {\bibfnamefont {M.}~\bibnamefont {Hoffmann}},
		\bibinfo {author} {\bibfnamefont {T.}~\bibnamefont {Huang}}, \bibinfo
		{author} {\bibfnamefont {T.~S.}\ \bibnamefont {Humble}}, \bibinfo {author}
		{\bibfnamefont {S.~V.}\ \bibnamefont {Isakov}}, \bibinfo {author}
		{\bibfnamefont {E.}~\bibnamefont {Jeffrey}}, \bibinfo {author} {\bibfnamefont
			{Z.}~\bibnamefont {Jiang}}, \bibinfo {author} {\bibfnamefont
			{D.}~\bibnamefont {Kafri}}, \bibinfo {author} {\bibfnamefont
			{K.}~\bibnamefont {Kechedzhi}}, \bibinfo {author} {\bibfnamefont
			{J.}~\bibnamefont {Kelly}}, \bibinfo {author} {\bibfnamefont {P.~V.}\
			\bibnamefont {Klimov}}, \bibinfo {author} {\bibfnamefont {S.}~\bibnamefont
			{Knysh}}, \bibinfo {author} {\bibfnamefont {A.}~\bibnamefont {Korotkov}},
		\bibinfo {author} {\bibfnamefont {F.}~\bibnamefont {Kostritsa}}, \bibinfo
		{author} {\bibfnamefont {D.}~\bibnamefont {Landhuis}}, \bibinfo {author}
		{\bibfnamefont {M.}~\bibnamefont {Lindmark}}, \bibinfo {author}
		{\bibfnamefont {E.}~\bibnamefont {Lucero}}, \bibinfo {author} {\bibfnamefont
			{D.}~\bibnamefont {Lyakh}}, \bibinfo {author} {\bibfnamefont
			{S.}~\bibnamefont {Mandr{\`a}}}, \bibinfo {author} {\bibfnamefont {J.~R.}\
			\bibnamefont {McClean}}, \bibinfo {author} {\bibfnamefont {M.}~\bibnamefont
			{McEwen}}, \bibinfo {author} {\bibfnamefont {A.}~\bibnamefont {Megrant}},
		\bibinfo {author} {\bibfnamefont {X.}~\bibnamefont {Mi}}, \bibinfo {author}
		{\bibfnamefont {K.}~\bibnamefont {Michielsen}}, \bibinfo {author}
		{\bibfnamefont {M.}~\bibnamefont {Mohseni}}, \bibinfo {author} {\bibfnamefont
			{J.}~\bibnamefont {Mutus}}, \bibinfo {author} {\bibfnamefont
			{O.}~\bibnamefont {Naaman}}, \bibinfo {author} {\bibfnamefont
			{M.}~\bibnamefont {Neeley}}, \bibinfo {author} {\bibfnamefont
			{C.}~\bibnamefont {Neill}}, \bibinfo {author} {\bibfnamefont {M.~Y.}\
			\bibnamefont {Niu}}, \bibinfo {author} {\bibfnamefont {E.}~\bibnamefont
			{Ostby}}, \bibinfo {author} {\bibfnamefont {A.}~\bibnamefont {Petukhov}},
		\bibinfo {author} {\bibfnamefont {J.~C.}\ \bibnamefont {Platt}}, \bibinfo
		{author} {\bibfnamefont {C.}~\bibnamefont {Quintana}}, \bibinfo {author}
		{\bibfnamefont {E.~G.}\ \bibnamefont {Rieffel}}, \bibinfo {author}
		{\bibfnamefont {P.}~\bibnamefont {Roushan}}, \bibinfo {author} {\bibfnamefont
			{N.~C.}\ \bibnamefont {Rubin}}, \bibinfo {author} {\bibfnamefont
			{D.}~\bibnamefont {Sank}}, \bibinfo {author} {\bibfnamefont {K.~J.}\
			\bibnamefont {Satzinger}}, \bibinfo {author} {\bibfnamefont {V.}~\bibnamefont
			{Smelyanskiy}}, \bibinfo {author} {\bibfnamefont {K.~J.}\ \bibnamefont
			{Sung}}, \bibinfo {author} {\bibfnamefont {M.~D.}\ \bibnamefont
			{Trevithick}}, \bibinfo {author} {\bibfnamefont {A.}~\bibnamefont
			{Vainsencher}}, \bibinfo {author} {\bibfnamefont {B.}~\bibnamefont
			{Villalonga}}, \bibinfo {author} {\bibfnamefont {T.}~\bibnamefont {White}},
		\bibinfo {author} {\bibfnamefont {Z.~J.}\ \bibnamefont {Yao}}, \bibinfo
		{author} {\bibfnamefont {P.}~\bibnamefont {Yeh}}, \bibinfo {author}
		{\bibfnamefont {A.}~\bibnamefont {Zalcman}}, \bibinfo {author} {\bibfnamefont
			{H.}~\bibnamefont {Neven}},\ and\ \bibinfo {author} {\bibfnamefont {J.~M.}\
			\bibnamefont {Martinis}},\ }\bibfield  {title} {\bibinfo {title} {Quantum
			supremacy using a programmable superconducting processor},\ }\href
	{https://doi.org/10.1038/s41586-019-1666-5} {\bibfield  {journal} {\bibinfo
			{journal} {Nature}\ }\textbf {\bibinfo {volume} {574}},\ \bibinfo {pages}
		{505} (\bibinfo {year} {2019})}\BibitemShut {NoStop}%
	\bibitem [{\citenamefont {Bluvstein}\ \emph
		{et~al.}(2024{\natexlab{a}})\citenamefont {Bluvstein}, \citenamefont
		{Evered}, \citenamefont {Geim}, \citenamefont {Li}, \citenamefont {Zhou},
		\citenamefont {Manovitz}, \citenamefont {Ebadi}, \citenamefont {Cain},
		\citenamefont {Kalinowski}, \citenamefont {Hangleiter}, \citenamefont
		{Bonilla~Ataides}, \citenamefont {Maskara}, \citenamefont {Cong},
		\citenamefont {Gao}, \citenamefont {Sales~Rodriguez}, \citenamefont
		{Karolyshyn}, \citenamefont {Semeghini}, \citenamefont {Gullans},
		\citenamefont {Greiner}, \citenamefont {Vuleti{\'{c}}},\ and\ \citenamefont
		{Lukin}}]{Bluvstein2024rydberg}%
	\BibitemOpen
	\bibfield  {author} {\bibinfo {author} {\bibfnamefont {D.}~\bibnamefont
			{Bluvstein}}, \bibinfo {author} {\bibfnamefont {S.~J.}\ \bibnamefont
			{Evered}}, \bibinfo {author} {\bibfnamefont {A.~A.}\ \bibnamefont {Geim}},
		\bibinfo {author} {\bibfnamefont {S.~H.}\ \bibnamefont {Li}}, \bibinfo
		{author} {\bibfnamefont {H.}~\bibnamefont {Zhou}}, \bibinfo {author}
		{\bibfnamefont {T.}~\bibnamefont {Manovitz}}, \bibinfo {author}
		{\bibfnamefont {S.}~\bibnamefont {Ebadi}}, \bibinfo {author} {\bibfnamefont
			{M.}~\bibnamefont {Cain}}, \bibinfo {author} {\bibfnamefont {M.}~\bibnamefont
			{Kalinowski}}, \bibinfo {author} {\bibfnamefont {D.}~\bibnamefont
			{Hangleiter}}, \bibinfo {author} {\bibfnamefont {J.~P.}\ \bibnamefont
			{Bonilla~Ataides}}, \bibinfo {author} {\bibfnamefont {N.}~\bibnamefont
			{Maskara}}, \bibinfo {author} {\bibfnamefont {I.}~\bibnamefont {Cong}},
		\bibinfo {author} {\bibfnamefont {X.}~\bibnamefont {Gao}}, \bibinfo {author}
		{\bibfnamefont {P.}~\bibnamefont {Sales~Rodriguez}}, \bibinfo {author}
		{\bibfnamefont {T.}~\bibnamefont {Karolyshyn}}, \bibinfo {author}
		{\bibfnamefont {G.}~\bibnamefont {Semeghini}}, \bibinfo {author}
		{\bibfnamefont {M.~J.}\ \bibnamefont {Gullans}}, \bibinfo {author}
		{\bibfnamefont {M.}~\bibnamefont {Greiner}}, \bibinfo {author} {\bibfnamefont
			{V.}~\bibnamefont {Vuleti{\'{c}}}},\ and\ \bibinfo {author} {\bibfnamefont
			{M.~D.}\ \bibnamefont {Lukin}},\ }\bibfield  {title} {\bibinfo {title}
		{Logical quantum processor based on reconfigurable atom arrays},\ }\href
	{https://doi.org/10.1038/s41586-023-06927-3} {\bibfield  {journal} {\bibinfo
			{journal} {Nature}\ }\textbf {\bibinfo {volume} {626}},\ \bibinfo {pages}
		{58} (\bibinfo {year} {2024}{\natexlab{a}})}\BibitemShut {NoStop}%
	\bibitem [{\citenamefont {Anshu}\ and\ \citenamefont
		{Arunachalam}(2024)}]{anshu2024survey}%
	\BibitemOpen
	\bibfield  {author} {\bibinfo {author} {\bibfnamefont {A.}~\bibnamefont
			{Anshu}}\ and\ \bibinfo {author} {\bibfnamefont {S.}~\bibnamefont
			{Arunachalam}},\ }\bibfield  {title} {\bibinfo {title} {A survey on the
			complexity of learning quantum states},\ }\href
	{https://doi.org/10.1038/s42254-023-00662-4} {\bibfield  {journal} {\bibinfo
			{journal} {Nat. Rev. Phys.}\ }\textbf {\bibinfo {volume} {6}},\ \bibinfo
		{pages} {59} (\bibinfo {year} {2024})}\BibitemShut {NoStop}%
	\bibitem [{\citenamefont {Bubeck}\ \emph {et~al.}(2020)\citenamefont {Bubeck},
		\citenamefont {Chen},\ and\ \citenamefont {Li}}]{bubeck2020entanglement}%
	\BibitemOpen
	\bibfield  {author} {\bibinfo {author} {\bibfnamefont {S.}~\bibnamefont
			{Bubeck}}, \bibinfo {author} {\bibfnamefont {S.}~\bibnamefont {Chen}},\ and\
		\bibinfo {author} {\bibfnamefont {J.}~\bibnamefont {Li}},\ }\bibfield
	{title} {\bibinfo {title} {Entanglement is necessary for optimal quantum
			property testing},\ }in\ \href {https://doi.org/10.1109/FOCS46700.2020.00070}
	{\emph {\bibinfo {booktitle} {Proceedings of 2020 IEEE 61st Annual Symposium
				on Foundations of Computer Science (FOCS)}}}\ (\bibinfo {organization}
	{IEEE},\ \bibinfo {year} {2020})\ pp.\ \bibinfo {pages}
	{692--703}\BibitemShut {NoStop}%
	\bibitem [{\citenamefont {Chen}\ \emph
		{et~al.}(2022{\natexlab{a}})\citenamefont {Chen}, \citenamefont {Li},\ and\
		\citenamefont {O’Donnell}}]{chen2022toward}%
	\BibitemOpen
	\bibfield  {author} {\bibinfo {author} {\bibfnamefont {S.}~\bibnamefont
			{Chen}}, \bibinfo {author} {\bibfnamefont {J.}~\bibnamefont {Li}},\ and\
		\bibinfo {author} {\bibfnamefont {R.}~\bibnamefont {O’Donnell}},\
	}\bibfield  {title} {\bibinfo {title} {Toward instance-optimal state
			certification with incoherent measurements},\ }in\ \href
	{https://proceedings.mlr.press/v178/chen22b.html} {\emph {\bibinfo
			{booktitle} {Proceedings of Thirty-Fifth Conference on Learning Theory}}}\
	(\bibinfo {organization} {PMLR},\ \bibinfo {year} {2022})\ pp.\ \bibinfo
	{pages} {2541--2596}\BibitemShut {NoStop}%
	\bibitem [{\citenamefont {Chen}\ \emph
		{et~al.}(2022{\natexlab{b}})\citenamefont {Chen}, \citenamefont {Li},
		\citenamefont {Huang},\ and\ \citenamefont {Liu}}]{chen2022tight}%
	\BibitemOpen
	\bibfield  {author} {\bibinfo {author} {\bibfnamefont {S.}~\bibnamefont
			{Chen}}, \bibinfo {author} {\bibfnamefont {J.}~\bibnamefont {Li}}, \bibinfo
		{author} {\bibfnamefont {B.}~\bibnamefont {Huang}},\ and\ \bibinfo {author}
		{\bibfnamefont {A.}~\bibnamefont {Liu}},\ }\bibfield  {title} {\bibinfo
		{title} {Tight bounds for quantum state certification with incoherent
			measurements},\ }in\ \href {https://doi.org/10.1109/FOCS54457.2022.00118}
	{\emph {\bibinfo {booktitle} {Proceedings of 2022 IEEE 63rd Annual Symposium
				on Foundations of Computer Science (FOCS)}}}\ (\bibinfo {organization}
	{IEEE},\ \bibinfo {year} {2022})\ pp.\ \bibinfo {pages}
	{1205--1213}\BibitemShut {NoStop}%
	\bibitem [{\citenamefont {Chen}\ \emph
		{et~al.}(2024{\natexlab{a}})\citenamefont {Chen}, \citenamefont {Li},\ and\
		\citenamefont {Liu}}]{chen2024optimalstate}%
	\BibitemOpen
	\bibfield  {author} {\bibinfo {author} {\bibfnamefont {S.}~\bibnamefont
			{Chen}}, \bibinfo {author} {\bibfnamefont {J.}~\bibnamefont {Li}},\ and\
		\bibinfo {author} {\bibfnamefont {A.}~\bibnamefont {Liu}},\ }\bibfield
	{title} {\bibinfo {title} {An optimal tradeoff between entanglement and copy
			complexity for state tomography},\ }in\ \href
	{https://doi.org/10.1145/3618260.3649704} {\emph {\bibinfo {booktitle}
			{Proceedings of the 56th Annual ACM Symposium on Theory of Computing}}}\
	(\bibinfo {year} {2024})\ pp.\ \bibinfo {pages} {1331--1342}\BibitemShut
	{NoStop}%
	\bibitem [{\citenamefont {Aaronson}\ \emph {et~al.}(2024)\citenamefont
		{Aaronson}, \citenamefont {Rosenthal}, \citenamefont {Subramanian},
		\citenamefont {Datta},\ and\ \citenamefont {Gur}}]{aaronson2024quantum}%
	\BibitemOpen
	\bibfield  {author} {\bibinfo {author} {\bibfnamefont {H.}~\bibnamefont
			{Aaronson}}, \bibinfo {author} {\bibfnamefont {G.}~\bibnamefont {Rosenthal}},
		\bibinfo {author} {\bibfnamefont {S.}~\bibnamefont {Subramanian}}, \bibinfo
		{author} {\bibfnamefont {A.}~\bibnamefont {Datta}},\ and\ \bibinfo {author}
		{\bibfnamefont {T.}~\bibnamefont {Gur}},\ }\bibfield  {title} {\bibinfo
		{title} {Quantum channel testing in average-case distance},\ }\href
	{https://arxiv.org/abs/2409.12566} {\bibfield  {journal} {\bibinfo  {journal}
			{arXiv:2409.12566}\ } (\bibinfo {year} {2024})}\BibitemShut {NoStop}%
	\bibitem [{\citenamefont {Fawzi}\ \emph {et~al.}(2023)\citenamefont {Fawzi},
		\citenamefont {Flammarion}, \citenamefont {Garivier},\ and\ \citenamefont
		{Oufkir}}]{fawzi2023quantum}%
	\BibitemOpen
	\bibfield  {author} {\bibinfo {author} {\bibfnamefont {O.}~\bibnamefont
			{Fawzi}}, \bibinfo {author} {\bibfnamefont {N.}~\bibnamefont {Flammarion}},
		\bibinfo {author} {\bibfnamefont {A.}~\bibnamefont {Garivier}},\ and\
		\bibinfo {author} {\bibfnamefont {A.}~\bibnamefont {Oufkir}},\ }\bibfield
	{title} {\bibinfo {title} {Quantum channel certification with incoherent
			strategies},\ }in\ \href {https://proceedings.mlr.press/v195/fawzi23a.html}
	{\emph {\bibinfo {booktitle} {Proceedings of 36th Annual Conference on
				Learning Theory (COLT 2023)}}}\ (\bibinfo {year} {2023})\ pp.\ \bibinfo
	{pages} {1--58}\BibitemShut {NoStop}%
	\bibitem [{\citenamefont {Huang}\ \emph {et~al.}(2021)\citenamefont {Huang},
		\citenamefont {Kueng},\ and\ \citenamefont
		{Preskill}}]{huang2021information}%
	\BibitemOpen
	\bibfield  {author} {\bibinfo {author} {\bibfnamefont {H.-Y.}\ \bibnamefont
			{Huang}}, \bibinfo {author} {\bibfnamefont {R.}~\bibnamefont {Kueng}},\ and\
		\bibinfo {author} {\bibfnamefont {J.}~\bibnamefont {Preskill}},\ }\bibfield
	{title} {\bibinfo {title} {Information-theoretic bounds on quantum advantage
			in machine learning},\ }\href
	{https://doi.org/10.1103/PhysRevLett.126.190505} {\bibfield  {journal}
		{\bibinfo  {journal} {Phys. Rev. Lett.}\ }\textbf {\bibinfo {volume} {126}},\
		\bibinfo {pages} {190505} (\bibinfo {year} {2021})}\BibitemShut {NoStop}%
	\bibitem [{\citenamefont {Aharonov}\ \emph {et~al.}(2022)\citenamefont
		{Aharonov}, \citenamefont {Cotler},\ and\ \citenamefont
		{Qi}}]{aharonov2022quantum}%
	\BibitemOpen
	\bibfield  {author} {\bibinfo {author} {\bibfnamefont {D.}~\bibnamefont
			{Aharonov}}, \bibinfo {author} {\bibfnamefont {J.}~\bibnamefont {Cotler}},\
		and\ \bibinfo {author} {\bibfnamefont {X.-L.}\ \bibnamefont {Qi}},\
	}\bibfield  {title} {\bibinfo {title} {Quantum algorithmic measurement},\
	}\href {https://doi.org/10.1038/s41467-021-27922-0} {\bibfield  {journal}
		{\bibinfo  {journal} {Nat. Commun.}\ }\textbf {\bibinfo {volume} {13}},\
		\bibinfo {pages} {887} (\bibinfo {year} {2022})}\BibitemShut {NoStop}%
	\bibitem [{\citenamefont {Chen}\ \emph
		{et~al.}(2022{\natexlab{c}})\citenamefont {Chen}, \citenamefont {Cotler},
		\citenamefont {Huang},\ and\ \citenamefont {Li}}]{chen2022memory}%
	\BibitemOpen
	\bibfield  {author} {\bibinfo {author} {\bibfnamefont {S.}~\bibnamefont
			{Chen}}, \bibinfo {author} {\bibfnamefont {J.}~\bibnamefont {Cotler}},
		\bibinfo {author} {\bibfnamefont {H.-Y.}\ \bibnamefont {Huang}},\ and\
		\bibinfo {author} {\bibfnamefont {J.}~\bibnamefont {Li}},\ }\bibfield
	{title} {\bibinfo {title} {Exponential separations between learning with and
			without quantum memory},\ }in\ \href
	{https://doi.org/10.1109/FOCS52979.2021.00063} {\emph {\bibinfo {booktitle}
			{Proceedings of 2021 IEEE 62nd Annual Symposium on Foundations of Computer
				Science (FOCS)}}}\ (\bibinfo {year} {2022})\ pp.\ \bibinfo {pages}
	{574--585}\BibitemShut {NoStop}%
	\bibitem [{\citenamefont {Chen}\ \emph
		{et~al.}(2022{\natexlab{d}})\citenamefont {Chen}, \citenamefont {Zhou},
		\citenamefont {Seif},\ and\ \citenamefont {Jiang}}]{chen2022quantum}%
	\BibitemOpen
	\bibfield  {author} {\bibinfo {author} {\bibfnamefont {S.}~\bibnamefont
			{Chen}}, \bibinfo {author} {\bibfnamefont {S.}~\bibnamefont {Zhou}}, \bibinfo
		{author} {\bibfnamefont {A.}~\bibnamefont {Seif}},\ and\ \bibinfo {author}
		{\bibfnamefont {L.}~\bibnamefont {Jiang}},\ }\bibfield  {title} {\bibinfo
		{title} {Quantum advantages for {P}auli channel estimation},\ }\href
	{https://doi.org/10.1103/PhysRevA.105.032435} {\bibfield  {journal} {\bibinfo
			{journal} {Phys. Rev. A}\ }\textbf {\bibinfo {volume} {105}},\ \bibinfo
		{pages} {032435} (\bibinfo {year} {2022}{\natexlab{d}})}\BibitemShut
	{NoStop}%
	\bibitem [{\citenamefont {Oh}\ \emph {et~al.}(2024)\citenamefont {Oh},
		\citenamefont {Chen}, \citenamefont {Wong}, \citenamefont {Zhou},
		\citenamefont {Huang}, \citenamefont {Nielsen}, \citenamefont {Liu},
		\citenamefont {Neergaard-Nielsen}, \citenamefont {Andersen}, \citenamefont
		{Jiang} \emph {et~al.}}]{oh2024entanglement}%
	\BibitemOpen
	\bibfield  {author} {\bibinfo {author} {\bibfnamefont {C.}~\bibnamefont
			{Oh}}, \bibinfo {author} {\bibfnamefont {S.}~\bibnamefont {Chen}}, \bibinfo
		{author} {\bibfnamefont {Y.}~\bibnamefont {Wong}}, \bibinfo {author}
		{\bibfnamefont {S.}~\bibnamefont {Zhou}}, \bibinfo {author} {\bibfnamefont
			{H.-Y.}\ \bibnamefont {Huang}}, \bibinfo {author} {\bibfnamefont {J.~A.}\
			\bibnamefont {Nielsen}}, \bibinfo {author} {\bibfnamefont {Z.-H.}\
			\bibnamefont {Liu}}, \bibinfo {author} {\bibfnamefont {J.~S.}\ \bibnamefont
			{Neergaard-Nielsen}}, \bibinfo {author} {\bibfnamefont {U.~L.}\ \bibnamefont
			{Andersen}}, \bibinfo {author} {\bibfnamefont {L.}~\bibnamefont {Jiang}},
		\emph {et~al.},\ }\bibfield  {title} {\bibinfo {title} {Entanglement-enabled
			advantage for learning a {B}osonic random displacement channel},\ }\href
	{https://arxiv.org/abs/2402.18809} {\bibfield  {journal} {\bibinfo  {journal}
			{arXiv:2402.18809}\ } (\bibinfo {year} {2024})}\BibitemShut {NoStop}%
	\bibitem [{\citenamefont {Chen}\ \emph
		{et~al.}(2024{\natexlab{b}})\citenamefont {Chen}, \citenamefont {Gong},\ and\
		\citenamefont {Ye}}]{chen2024optimal}%
	\BibitemOpen
	\bibfield  {author} {\bibinfo {author} {\bibfnamefont {S.}~\bibnamefont
			{Chen}}, \bibinfo {author} {\bibfnamefont {W.}~\bibnamefont {Gong}},\ and\
		\bibinfo {author} {\bibfnamefont {Q.}~\bibnamefont {Ye}},\ }\bibfield
	{title} {\bibinfo {title} {Optimal tradeoffs for estimating {P}auli
			observables},\ }\href {https://arxiv.org/abs/2404.19105} {\bibfield
		{journal} {\bibinfo  {journal} {arXiv:2404.19105}\ } (\bibinfo {year}
		{2024}{\natexlab{b}})}\BibitemShut {NoStop}%
	\bibitem [{\citenamefont {Chen}\ and\ \citenamefont
		{Gong}(2023)}]{chen2023efficient}%
	\BibitemOpen
	\bibfield  {author} {\bibinfo {author} {\bibfnamefont {S.}~\bibnamefont
			{Chen}}\ and\ \bibinfo {author} {\bibfnamefont {W.}~\bibnamefont {Gong}},\
	}\bibfield  {title} {\bibinfo {title} {Efficient {P}auli channel estimation
			with logarithmic quantum memory},\ }\href {https://arxiv.org/abs/2309.14326}
	{\bibfield  {journal} {\bibinfo  {journal} {arXiv:2309.14326}\ } (\bibinfo
		{year} {2023})}\BibitemShut {NoStop}%
	\bibitem [{\citenamefont {Chen}\ \emph
		{et~al.}(2024{\natexlab{c}})\citenamefont {Chen}, \citenamefont {Oh},
		\citenamefont {Zhou}, \citenamefont {Huang},\ and\ \citenamefont
		{Jiang}}]{chen2024tight}%
	\BibitemOpen
	\bibfield  {author} {\bibinfo {author} {\bibfnamefont {S.}~\bibnamefont
			{Chen}}, \bibinfo {author} {\bibfnamefont {C.}~\bibnamefont {Oh}}, \bibinfo
		{author} {\bibfnamefont {S.}~\bibnamefont {Zhou}}, \bibinfo {author}
		{\bibfnamefont {H.-Y.}\ \bibnamefont {Huang}},\ and\ \bibinfo {author}
		{\bibfnamefont {L.}~\bibnamefont {Jiang}},\ }\bibfield  {title} {\bibinfo
		{title} {Tight bounds on pauli channel learning without entanglement},\
	}\href {https://doi.org/10.1103/PhysRevLett.132.180805} {\bibfield  {journal}
		{\bibinfo  {journal} {Phys. Rev. Lett.}\ }\textbf {\bibinfo {volume} {132}},\
		\bibinfo {pages} {180805} (\bibinfo {year} {2024}{\natexlab{c}})}\BibitemShut
	{NoStop}%
	\bibitem [{\citenamefont {Huang}\ \emph {et~al.}(2022)\citenamefont {Huang},
		\citenamefont {Broughton}, \citenamefont {Cotler}, \citenamefont {Chen},
		\citenamefont {Li}, \citenamefont {Mohseni}, \citenamefont {Neven},
		\citenamefont {Babbush}, \citenamefont {Kueng}, \citenamefont {Preskill},\
		and\ \citenamefont {McClean}}]{Huang_2022_quantum}%
	\BibitemOpen
	\bibfield  {author} {\bibinfo {author} {\bibfnamefont {H.-Y.}\ \bibnamefont
			{Huang}}, \bibinfo {author} {\bibfnamefont {M.}~\bibnamefont {Broughton}},
		\bibinfo {author} {\bibfnamefont {J.}~\bibnamefont {Cotler}}, \bibinfo
		{author} {\bibfnamefont {S.}~\bibnamefont {Chen}}, \bibinfo {author}
		{\bibfnamefont {J.}~\bibnamefont {Li}}, \bibinfo {author} {\bibfnamefont
			{M.}~\bibnamefont {Mohseni}}, \bibinfo {author} {\bibfnamefont
			{H.}~\bibnamefont {Neven}}, \bibinfo {author} {\bibfnamefont
			{R.}~\bibnamefont {Babbush}}, \bibinfo {author} {\bibfnamefont
			{R.}~\bibnamefont {Kueng}}, \bibinfo {author} {\bibfnamefont
			{J.}~\bibnamefont {Preskill}},\ and\ \bibinfo {author} {\bibfnamefont
			{J.~R.}\ \bibnamefont {McClean}},\ }\bibfield  {title} {\bibinfo {title}
		{Quantum advantage in learning from experiments},\ }\href
	{https://doi.org/10.1126/science.abn7293} {\bibfield  {journal} {\bibinfo
			{journal} {Science}\ }\textbf {\bibinfo {volume} {376}},\ \bibinfo {pages}
		{1182–1186} (\bibinfo {year} {2022})}\BibitemShut {NoStop}%
	\bibitem [{\citenamefont {Bluvstein}\ \emph
		{et~al.}(2024{\natexlab{b}})\citenamefont {Bluvstein}, \citenamefont
		{Evered}, \citenamefont {Geim}, \citenamefont {Li}, \citenamefont {Zhou},
		\citenamefont {Manovitz}, \citenamefont {Ebadi}, \citenamefont {Cain},
		\citenamefont {Kalinowski}, \citenamefont {Hangleiter} \emph
		{et~al.}}]{bluvstein2024logical}%
	\BibitemOpen
	\bibfield  {author} {\bibinfo {author} {\bibfnamefont {D.}~\bibnamefont
			{Bluvstein}}, \bibinfo {author} {\bibfnamefont {S.~J.}\ \bibnamefont
			{Evered}}, \bibinfo {author} {\bibfnamefont {A.~A.}\ \bibnamefont {Geim}},
		\bibinfo {author} {\bibfnamefont {S.~H.}\ \bibnamefont {Li}}, \bibinfo
		{author} {\bibfnamefont {H.}~\bibnamefont {Zhou}}, \bibinfo {author}
		{\bibfnamefont {T.}~\bibnamefont {Manovitz}}, \bibinfo {author}
		{\bibfnamefont {S.}~\bibnamefont {Ebadi}}, \bibinfo {author} {\bibfnamefont
			{M.}~\bibnamefont {Cain}}, \bibinfo {author} {\bibfnamefont {M.}~\bibnamefont
			{Kalinowski}}, \bibinfo {author} {\bibfnamefont {D.}~\bibnamefont
			{Hangleiter}}, \emph {et~al.},\ }\bibfield  {title} {\bibinfo {title}
		{Logical quantum processor based on reconfigurable atom arrays},\ }\href
	{https://doi.org/10.1038/s41586-023-06927-3} {\bibfield  {journal} {\bibinfo
			{journal} {Nature}\ }\textbf {\bibinfo {volume} {626}},\ \bibinfo {pages}
		{58} (\bibinfo {year} {2024}{\natexlab{b}})}\BibitemShut {NoStop}%
	\bibitem [{\citenamefont {Seif}\ \emph {et~al.}(2024)\citenamefont {Seif},
		\citenamefont {Chen}, \citenamefont {Majumder}, \citenamefont {Liao},
		\citenamefont {Wang}, \citenamefont {Malekakhlagh}, \citenamefont
		{Javadi-Abhari}, \citenamefont {Jiang},\ and\ \citenamefont
		{Minev}}]{seif2024entanglement}%
	\BibitemOpen
	\bibfield  {author} {\bibinfo {author} {\bibfnamefont {A.}~\bibnamefont
			{Seif}}, \bibinfo {author} {\bibfnamefont {S.}~\bibnamefont {Chen}}, \bibinfo
		{author} {\bibfnamefont {S.}~\bibnamefont {Majumder}}, \bibinfo {author}
		{\bibfnamefont {H.}~\bibnamefont {Liao}}, \bibinfo {author} {\bibfnamefont
			{D.~S.}\ \bibnamefont {Wang}}, \bibinfo {author} {\bibfnamefont
			{M.}~\bibnamefont {Malekakhlagh}}, \bibinfo {author} {\bibfnamefont
			{A.}~\bibnamefont {Javadi-Abhari}}, \bibinfo {author} {\bibfnamefont
			{L.}~\bibnamefont {Jiang}},\ and\ \bibinfo {author} {\bibfnamefont {Z.~K.}\
			\bibnamefont {Minev}},\ }\bibfield  {title} {\bibinfo {title}
		{Entanglement-enhanced learning of quantum processes at scale},\ }\href
	{https://arxiv.org/abs/2408.03376} {\bibfield  {journal} {\bibinfo  {journal}
			{arXiv:2408.03376}\ } (\bibinfo {year} {2024})}\BibitemShut {NoStop}%
	\bibitem [{\citenamefont {Chen}\ \emph {et~al.}(2021)\citenamefont {Chen},
		\citenamefont {Cotler}, \citenamefont {Huang},\ and\ \citenamefont
		{Li}}]{chen2021hierarchy}%
	\BibitemOpen
	\bibfield  {author} {\bibinfo {author} {\bibfnamefont {S.}~\bibnamefont
			{Chen}}, \bibinfo {author} {\bibfnamefont {J.}~\bibnamefont {Cotler}},
		\bibinfo {author} {\bibfnamefont {H.-Y.}\ \bibnamefont {Huang}},\ and\
		\bibinfo {author} {\bibfnamefont {J.}~\bibnamefont {Li}},\ }\bibfield
	{title} {\bibinfo {title} {A hierarchy for replica quantum advantage},\
	}\href {https://arxiv.org/abs/2111.05874} {\bibfield  {journal} {\bibinfo
			{journal} {arXiv:2111.05874}\ } (\bibinfo {year} {2021})}\BibitemShut
	{NoStop}%
	\bibitem [{\citenamefont {Uhlmann}(1976)}]{uhlmann1976transition}%
	\BibitemOpen
	\bibfield  {author} {\bibinfo {author} {\bibfnamefont {A.}~\bibnamefont
			{Uhlmann}},\ }\bibfield  {title} {\bibinfo {title} {The “transition
			probability” in the state space of {A}*-algebra},\ }\href
	{https://doi.org/10.1016/0034-4877(76)90060-4} {\bibfield  {journal}
		{\bibinfo  {journal} {Rep. Math. Phys.}\ }\textbf {\bibinfo {volume} {9}},\
		\bibinfo {pages} {273} (\bibinfo {year} {1976})}\BibitemShut {NoStop}%
	\bibitem [{\citenamefont {Bostanci}\ \emph {et~al.}(2023)\citenamefont
		{Bostanci}, \citenamefont {Efron}, \citenamefont {Metger}, \citenamefont
		{Poremba}, \citenamefont {Qian},\ and\ \citenamefont
		{Yuen}}]{bostanci2023unitary}%
	\BibitemOpen
	\bibfield  {author} {\bibinfo {author} {\bibfnamefont {J.}~\bibnamefont
			{Bostanci}}, \bibinfo {author} {\bibfnamefont {Y.}~\bibnamefont {Efron}},
		\bibinfo {author} {\bibfnamefont {T.}~\bibnamefont {Metger}}, \bibinfo
		{author} {\bibfnamefont {A.}~\bibnamefont {Poremba}}, \bibinfo {author}
		{\bibfnamefont {L.}~\bibnamefont {Qian}},\ and\ \bibinfo {author}
		{\bibfnamefont {H.}~\bibnamefont {Yuen}},\ }\bibfield  {title} {\bibinfo
		{title} {Unitary complexity and the {U}hlmann transformation problem},\
	}\href {https://arxiv.org/abs/2306.13073} {\bibfield  {journal} {\bibinfo
			{journal} {arXiv:2306.13073}\ } (\bibinfo {year} {2023})}\BibitemShut
	{NoStop}%
	\bibitem [{\citenamefont {Kokail}\ \emph {et~al.}(2021)\citenamefont {Kokail},
		\citenamefont {van Bijnen}, \citenamefont {Elben}, \citenamefont
		{Vermersch},\ and\ \citenamefont {Zoller}}]{Kokail2021entanglement}%
	\BibitemOpen
	\bibfield  {author} {\bibinfo {author} {\bibfnamefont {C.}~\bibnamefont
			{Kokail}}, \bibinfo {author} {\bibfnamefont {R.}~\bibnamefont {van Bijnen}},
		\bibinfo {author} {\bibfnamefont {A.}~\bibnamefont {Elben}}, \bibinfo
		{author} {\bibfnamefont {B.}~\bibnamefont {Vermersch}},\ and\ \bibinfo
		{author} {\bibfnamefont {P.}~\bibnamefont {Zoller}},\ }\bibfield  {title}
	{\bibinfo {title} {Entanglement hamiltonian tomography in quantum
			simulation},\ }\href {https://doi.org/10.1038/s41567-021-01260-w} {\bibfield
		{journal} {\bibinfo  {journal} {Nat. Phys.}\ }\textbf {\bibinfo {volume}
			{17}},\ \bibinfo {pages} {936} (\bibinfo {year} {2021})}\BibitemShut
	{NoStop}%
	\bibitem [{\citenamefont {Islam}\ \emph {et~al.}(2015)\citenamefont {Islam},
		\citenamefont {Ma}, \citenamefont {Preiss}, \citenamefont {Eric~Tai},
		\citenamefont {Lukin}, \citenamefont {Rispoli},\ and\ \citenamefont
		{Greiner}}]{Islam2015entropy}%
	\BibitemOpen
	\bibfield  {author} {\bibinfo {author} {\bibfnamefont {R.}~\bibnamefont
			{Islam}}, \bibinfo {author} {\bibfnamefont {R.}~\bibnamefont {Ma}}, \bibinfo
		{author} {\bibfnamefont {P.~M.}\ \bibnamefont {Preiss}}, \bibinfo {author}
		{\bibfnamefont {M.}~\bibnamefont {Eric~Tai}}, \bibinfo {author}
		{\bibfnamefont {A.}~\bibnamefont {Lukin}}, \bibinfo {author} {\bibfnamefont
			{M.}~\bibnamefont {Rispoli}},\ and\ \bibinfo {author} {\bibfnamefont
			{M.}~\bibnamefont {Greiner}},\ }\bibfield  {title} {\bibinfo {title}
		{Measuring entanglement entropy in a quantum many-body system},\ }\href
	{https://doi.org/10.1038/nature15750} {\bibfield  {journal} {\bibinfo
			{journal} {Nature}\ }\textbf {\bibinfo {volume} {528}},\ \bibinfo {pages}
		{77} (\bibinfo {year} {2015})}\BibitemShut {NoStop}%
	\bibitem [{\citenamefont {Kaufman}\ \emph {et~al.}(2016)\citenamefont
		{Kaufman}, \citenamefont {Tai}, \citenamefont {Lukin}, \citenamefont
		{Rispoli}, \citenamefont {Schittko}, \citenamefont {Preiss},\ and\
		\citenamefont {Greiner}}]{adam2016thermalization}%
	\BibitemOpen
	\bibfield  {author} {\bibinfo {author} {\bibfnamefont {A.~M.}\ \bibnamefont
			{Kaufman}}, \bibinfo {author} {\bibfnamefont {M.~E.}\ \bibnamefont {Tai}},
		\bibinfo {author} {\bibfnamefont {A.}~\bibnamefont {Lukin}}, \bibinfo
		{author} {\bibfnamefont {M.}~\bibnamefont {Rispoli}}, \bibinfo {author}
		{\bibfnamefont {R.}~\bibnamefont {Schittko}}, \bibinfo {author}
		{\bibfnamefont {P.~M.}\ \bibnamefont {Preiss}},\ and\ \bibinfo {author}
		{\bibfnamefont {M.}~\bibnamefont {Greiner}},\ }\bibfield  {title} {\bibinfo
		{title} {Quantum thermalization through entanglement in an isolated many-body
			system},\ }\href {https://doi.org/10.1126/science.aaf6725} {\bibfield
		{journal} {\bibinfo  {journal} {Science}\ }\textbf {\bibinfo {volume}
			{353}},\ \bibinfo {pages} {794} (\bibinfo {year} {2016})}\BibitemShut
	{NoStop}%
	\bibitem [{\citenamefont {Brydges}\ \emph {et~al.}(2019)\citenamefont
		{Brydges}, \citenamefont {Elben}, \citenamefont {Jurcevic}, \citenamefont
		{Vermersch}, \citenamefont {Maier}, \citenamefont {Lanyon}, \citenamefont
		{Zoller}, \citenamefont {Blatt},\ and\ \citenamefont
		{Roos}}]{brydges2019renyi}%
	\BibitemOpen
	\bibfield  {author} {\bibinfo {author} {\bibfnamefont {T.}~\bibnamefont
			{Brydges}}, \bibinfo {author} {\bibfnamefont {A.}~\bibnamefont {Elben}},
		\bibinfo {author} {\bibfnamefont {P.}~\bibnamefont {Jurcevic}}, \bibinfo
		{author} {\bibfnamefont {B.}~\bibnamefont {Vermersch}}, \bibinfo {author}
		{\bibfnamefont {C.}~\bibnamefont {Maier}}, \bibinfo {author} {\bibfnamefont
			{B.~P.}\ \bibnamefont {Lanyon}}, \bibinfo {author} {\bibfnamefont
			{P.}~\bibnamefont {Zoller}}, \bibinfo {author} {\bibfnamefont
			{R.}~\bibnamefont {Blatt}},\ and\ \bibinfo {author} {\bibfnamefont {C.~F.}\
			\bibnamefont {Roos}},\ }\bibfield  {title} {\bibinfo {title} {Probing rényi
			entanglement entropy via randomized measurements},\ }\href
	{https://doi.org/10.1126/science.aau4963} {\bibfield  {journal} {\bibinfo
			{journal} {Science}\ }\textbf {\bibinfo {volume} {364}},\ \bibinfo {pages}
		{260} (\bibinfo {year} {2019})}\BibitemShut {NoStop}%
	\bibitem [{\citenamefont {Shao}\ \emph {et~al.}(2024)\citenamefont {Shao},
		\citenamefont {Wang}, \citenamefont {Zhu}, \citenamefont {Zhu}, \citenamefont
		{Sun}, \citenamefont {Chen}, \citenamefont {Zhang}, \citenamefont {Fan},
		\citenamefont {Deng}, \citenamefont {Yao}, \citenamefont {Chen},\ and\
		\citenamefont {Pan}}]{Shao2024FHM}%
	\BibitemOpen
	\bibfield  {author} {\bibinfo {author} {\bibfnamefont {H.-J.}\ \bibnamefont
			{Shao}}, \bibinfo {author} {\bibfnamefont {Y.-X.}\ \bibnamefont {Wang}},
		\bibinfo {author} {\bibfnamefont {D.-Z.}\ \bibnamefont {Zhu}}, \bibinfo
		{author} {\bibfnamefont {Y.-S.}\ \bibnamefont {Zhu}}, \bibinfo {author}
		{\bibfnamefont {H.-N.}\ \bibnamefont {Sun}}, \bibinfo {author} {\bibfnamefont
			{S.-Y.}\ \bibnamefont {Chen}}, \bibinfo {author} {\bibfnamefont
			{C.}~\bibnamefont {Zhang}}, \bibinfo {author} {\bibfnamefont {Z.-J.}\
			\bibnamefont {Fan}}, \bibinfo {author} {\bibfnamefont {Y.}~\bibnamefont
			{Deng}}, \bibinfo {author} {\bibfnamefont {X.-C.}\ \bibnamefont {Yao}},
		\bibinfo {author} {\bibfnamefont {Y.-A.}\ \bibnamefont {Chen}},\ and\
		\bibinfo {author} {\bibfnamefont {J.-W.}\ \bibnamefont {Pan}},\ }\bibfield
	{title} {\bibinfo {title} {Antiferromagnetic phase transition in a 3d
			fermionic hubbard model},\ }\href
	{https://doi.org/10.1038/s41586-024-07689-2} {\bibfield  {journal} {\bibinfo
			{journal} {Nature}\ }\textbf {\bibinfo {volume} {632}},\ \bibinfo {pages}
		{267} (\bibinfo {year} {2024})}\BibitemShut {NoStop}%
	\bibitem [{\citenamefont {Bennett}\ and\ \citenamefont
		{Brassard}(1984)}]{bennett1984quantum}%
	\BibitemOpen
	\bibfield  {author} {\bibinfo {author} {\bibfnamefont {C.~H.}\ \bibnamefont
			{Bennett}}\ and\ \bibinfo {author} {\bibfnamefont {G.}~\bibnamefont
			{Brassard}},\ }\bibfield  {title} {\bibinfo {title} {Quantum cryptography:
			Public key distribution and coin tossing},\ }in\ \href@noop {} {\emph
		{\bibinfo {booktitle} {Proceedings of IEEE International Conference on
				Computers, Systems and Signal Processing}}}\ (\bibinfo {organization}
	{Bangalore, India},\ \bibinfo {year} {1984})\ pp.\ \bibinfo {pages}
	{175--179}\BibitemShut {NoStop}%
	\bibitem [{\citenamefont {Ekert}(1991)}]{ekert1991quantum}%
	\BibitemOpen
	\bibfield  {author} {\bibinfo {author} {\bibfnamefont {A.~K.}\ \bibnamefont
			{Ekert}},\ }\bibfield  {title} {\bibinfo {title} {Quantum cryptography based
			on bell's theorem},\ }\href {https://doi.org/10.1103/PhysRevLett.67.661}
	{\bibfield  {journal} {\bibinfo  {journal} {Phys. Rev. Lett.}\ }\textbf
		{\bibinfo {volume} {67}},\ \bibinfo {pages} {661} (\bibinfo {year}
		{1991})}\BibitemShut {NoStop}%
	\bibitem [{\citenamefont {Herrero-Collantes}\ and\ \citenamefont
		{Garcia-Escartin}(2017)}]{herrero2017qrng}%
	\BibitemOpen
	\bibfield  {author} {\bibinfo {author} {\bibfnamefont {M.}~\bibnamefont
			{Herrero-Collantes}}\ and\ \bibinfo {author} {\bibfnamefont {J.~C.}\
			\bibnamefont {Garcia-Escartin}},\ }\bibfield  {title} {\bibinfo {title}
		{Quantum random number generators},\ }\href
	{https://doi.org/10.1103/RevModPhys.89.015004} {\bibfield  {journal}
		{\bibinfo  {journal} {Rev. Mod. Phys.}\ }\textbf {\bibinfo {volume} {89}},\
		\bibinfo {pages} {015004} (\bibinfo {year} {2017})}\BibitemShut {NoStop}%
	\bibitem [{\citenamefont {Low}\ and\ \citenamefont
		{Chuang}(2019)}]{Low2019hamiltonian}%
	\BibitemOpen
	\bibfield  {author} {\bibinfo {author} {\bibfnamefont {G.~H.}\ \bibnamefont
			{Low}}\ and\ \bibinfo {author} {\bibfnamefont {I.~L.}\ \bibnamefont
			{Chuang}},\ }\bibfield  {title} {\bibinfo {title} {Hamiltonian {S}imulation
			by {Q}ubitization},\ }\href {https://doi.org/10.22331/q-2019-07-12-163}
	{\bibfield  {journal} {\bibinfo  {journal} {{Quantum}}\ }\textbf {\bibinfo
			{volume} {3}},\ \bibinfo {pages} {163} (\bibinfo {year} {2019})}\BibitemShut
	{NoStop}%
	\bibitem [{\citenamefont {Wang}\ \emph {et~al.}(2024)\citenamefont {Wang},
		\citenamefont {Guan}, \citenamefont {Liu}, \citenamefont {Zhang},\ and\
		\citenamefont {Ying}}]{wang2024entropy}%
	\BibitemOpen
	\bibfield  {author} {\bibinfo {author} {\bibfnamefont {Q.}~\bibnamefont
			{Wang}}, \bibinfo {author} {\bibfnamefont {J.}~\bibnamefont {Guan}}, \bibinfo
		{author} {\bibfnamefont {J.}~\bibnamefont {Liu}}, \bibinfo {author}
		{\bibfnamefont {Z.}~\bibnamefont {Zhang}},\ and\ \bibinfo {author}
		{\bibfnamefont {M.}~\bibnamefont {Ying}},\ }\bibfield  {title} {\bibinfo
		{title} {New quantum algorithms for computing quantum entropies and
			distances},\ }\href {https://doi.org/10.1109/TIT.2024.3399014} {\bibfield
		{journal} {\bibinfo  {journal} {IEEE Transactions on Information Theory}\
		}\textbf {\bibinfo {volume} {70}},\ \bibinfo {pages} {5653} (\bibinfo {year}
		{2024})}\BibitemShut {NoStop}%
	\bibitem [{\citenamefont {Liu}\ and\ \citenamefont
		{Wang}(2024)}]{liu2024estimatingtracequantumstate}%
	\BibitemOpen
	\bibfield  {author} {\bibinfo {author} {\bibfnamefont {Y.}~\bibnamefont
			{Liu}}\ and\ \bibinfo {author} {\bibfnamefont {Q.}~\bibnamefont {Wang}},\
	}\href {https://arxiv.org/abs/2410.13559} {\bibinfo {title} {On estimating
			the trace of quantum state powers}} (\bibinfo {year} {2024}),\ \Eprint
	{https://arxiv.org/abs/2410.13559} {arXiv:2410.13559 [quant-ph]} \BibitemShut
	{NoStop}%
	\bibitem [{\citenamefont {DiVincenzo}\ \emph {et~al.}(1999)\citenamefont
		{DiVincenzo}, \citenamefont {Fuchs}, \citenamefont {Mabuchi}, \citenamefont
		{Smolin}, \citenamefont {Thapliyal},\ and\ \citenamefont
		{Uhlmann}}]{divincenzo1999entanglement}%
	\BibitemOpen
	\bibfield  {author} {\bibinfo {author} {\bibfnamefont {D.~P.}\ \bibnamefont
			{DiVincenzo}}, \bibinfo {author} {\bibfnamefont {C.~A.}\ \bibnamefont
			{Fuchs}}, \bibinfo {author} {\bibfnamefont {H.}~\bibnamefont {Mabuchi}},
		\bibinfo {author} {\bibfnamefont {J.~A.}\ \bibnamefont {Smolin}}, \bibinfo
		{author} {\bibfnamefont {A.}~\bibnamefont {Thapliyal}},\ and\ \bibinfo
		{author} {\bibfnamefont {A.}~\bibnamefont {Uhlmann}},\ }\bibfield  {title}
	{\bibinfo {title} {Entanglement of assistance},\ }in\ \href
	{https://link.springer.com/chapter/10.1007/3-540-49208-9_21} {\emph {\bibinfo
			{booktitle} {Quantum Computing and Quantum Communications: First NASA
				International Conference, QCQC’98 Palm Springs, California, USA February
				17--20, 1998 Selected Papers}}}\ (\bibinfo {organization} {Springer},\
	\bibinfo {year} {1999})\ pp.\ \bibinfo {pages} {247--257}\BibitemShut
	{NoStop}%
	\bibitem [{\citenamefont {Ezzell}\ \emph {et~al.}(2023)\citenamefont {Ezzell},
		\citenamefont {Ball}, \citenamefont {Siddiqui}, \citenamefont {Wilde},
		\citenamefont {Sornborger}, \citenamefont {Coles},\ and\ \citenamefont
		{Holmes}}]{ezzell2023quantum}%
	\BibitemOpen
	\bibfield  {author} {\bibinfo {author} {\bibfnamefont {N.}~\bibnamefont
			{Ezzell}}, \bibinfo {author} {\bibfnamefont {E.~M.}\ \bibnamefont {Ball}},
		\bibinfo {author} {\bibfnamefont {A.~U.}\ \bibnamefont {Siddiqui}}, \bibinfo
		{author} {\bibfnamefont {M.~M.}\ \bibnamefont {Wilde}}, \bibinfo {author}
		{\bibfnamefont {A.~T.}\ \bibnamefont {Sornborger}}, \bibinfo {author}
		{\bibfnamefont {P.~J.}\ \bibnamefont {Coles}},\ and\ \bibinfo {author}
		{\bibfnamefont {Z.}~\bibnamefont {Holmes}},\ }\bibfield  {title} {\bibinfo
		{title} {Quantum mixed state compiling},\ }\href
	{https://doi.org/10.1088/2058-9565/acc4e3} {\bibfield  {journal} {\bibinfo
			{journal} {Quantum Sci. Technol.}\ }\textbf {\bibinfo {volume} {8}},\
		\bibinfo {pages} {035001} (\bibinfo {year} {2023})}\BibitemShut {NoStop}%
	\bibitem [{\citenamefont {Zhang}\ \emph {et~al.}(2021)\citenamefont {Zhang},
		\citenamefont {Sun}, \citenamefont {Fang}, \citenamefont {Zhang},
		\citenamefont {Yuan},\ and\ \citenamefont {Lu}}]{zhang2021shadow}%
	\BibitemOpen
	\bibfield  {author} {\bibinfo {author} {\bibfnamefont {T.}~\bibnamefont
			{Zhang}}, \bibinfo {author} {\bibfnamefont {J.}~\bibnamefont {Sun}}, \bibinfo
		{author} {\bibfnamefont {X.-X.}\ \bibnamefont {Fang}}, \bibinfo {author}
		{\bibfnamefont {X.-M.}\ \bibnamefont {Zhang}}, \bibinfo {author}
		{\bibfnamefont {X.}~\bibnamefont {Yuan}},\ and\ \bibinfo {author}
		{\bibfnamefont {H.}~\bibnamefont {Lu}},\ }\bibfield  {title} {\bibinfo
		{title} {Experimental quantum state measurement with classical shadows},\
	}\href {https://doi.org/10.1103/PhysRevLett.127.200501} {\bibfield  {journal}
		{\bibinfo  {journal} {Phys. Rev. Lett.}\ }\textbf {\bibinfo {volume} {127}},\
		\bibinfo {pages} {200501} (\bibinfo {year} {2021})}\BibitemShut {NoStop}%
	\bibitem [{\citenamefont {Shaw}\ \emph {et~al.}(2024)\citenamefont {Shaw},
		\citenamefont {Chen}, \citenamefont {Choi}, \citenamefont {Mark},
		\citenamefont {Scholl}, \citenamefont {Finkelstein}, \citenamefont {Elben},
		\citenamefont {Choi},\ and\ \citenamefont {Endres}}]{Shaw2024benchmarking}%
	\BibitemOpen
	\bibfield  {author} {\bibinfo {author} {\bibfnamefont {A.~L.}\ \bibnamefont
			{Shaw}}, \bibinfo {author} {\bibfnamefont {Z.}~\bibnamefont {Chen}}, \bibinfo
		{author} {\bibfnamefont {J.}~\bibnamefont {Choi}}, \bibinfo {author}
		{\bibfnamefont {D.~K.}\ \bibnamefont {Mark}}, \bibinfo {author}
		{\bibfnamefont {P.}~\bibnamefont {Scholl}}, \bibinfo {author} {\bibfnamefont
			{R.}~\bibnamefont {Finkelstein}}, \bibinfo {author} {\bibfnamefont
			{A.}~\bibnamefont {Elben}}, \bibinfo {author} {\bibfnamefont
			{S.}~\bibnamefont {Choi}},\ and\ \bibinfo {author} {\bibfnamefont
			{M.}~\bibnamefont {Endres}},\ }\bibfield  {title} {\bibinfo {title}
		{Benchmarking highly entangled states on a 60-atom analogue quantum
			simulator},\ }\href {https://doi.org/10.1038/s41586-024-07173-x} {\bibfield
		{journal} {\bibinfo  {journal} {Nature}\ }\textbf {\bibinfo {volume} {628}},\
		\bibinfo {pages} {71} (\bibinfo {year} {2024})}\BibitemShut {NoStop}%
	\bibitem [{\citenamefont {Gong}\ \emph {et~al.}(2024)\citenamefont {Gong},
		\citenamefont {Haferkamp}, \citenamefont {Ye},\ and\ \citenamefont
		{Zhang}}]{gong2024sample}%
	\BibitemOpen
	\bibfield  {author} {\bibinfo {author} {\bibfnamefont {W.}~\bibnamefont
			{Gong}}, \bibinfo {author} {\bibfnamefont {J.}~\bibnamefont {Haferkamp}},
		\bibinfo {author} {\bibfnamefont {Q.}~\bibnamefont {Ye}},\ and\ \bibinfo
		{author} {\bibfnamefont {Z.}~\bibnamefont {Zhang}},\ }\bibfield  {title}
	{\bibinfo {title} {On the sample complexity of purity and inner product
			estimation},\ }\href {https://arxiv.org/abs/2410.12712} {\bibfield  {journal}
		{\bibinfo  {journal} {arXiv:2410.12712}\ } (\bibinfo {year}
		{2024})}\BibitemShut {NoStop}%
	\bibitem [{\citenamefont {Cotler}\ \emph {et~al.}(2019)\citenamefont {Cotler},
		\citenamefont {Choi}, \citenamefont {Lukin}, \citenamefont {Gharibyan},
		\citenamefont {Grover}, \citenamefont {Tai}, \citenamefont {Rispoli},
		\citenamefont {Schittko}, \citenamefont {Preiss}, \citenamefont {Kaufman},
		\citenamefont {Greiner}, \citenamefont {Pichler},\ and\ \citenamefont
		{Hayden}}]{cotler2019cooling}%
	\BibitemOpen
	\bibfield  {author} {\bibinfo {author} {\bibfnamefont {J.}~\bibnamefont
			{Cotler}}, \bibinfo {author} {\bibfnamefont {S.}~\bibnamefont {Choi}},
		\bibinfo {author} {\bibfnamefont {A.}~\bibnamefont {Lukin}}, \bibinfo
		{author} {\bibfnamefont {H.}~\bibnamefont {Gharibyan}}, \bibinfo {author}
		{\bibfnamefont {T.}~\bibnamefont {Grover}}, \bibinfo {author} {\bibfnamefont
			{M.~E.}\ \bibnamefont {Tai}}, \bibinfo {author} {\bibfnamefont
			{M.}~\bibnamefont {Rispoli}}, \bibinfo {author} {\bibfnamefont
			{R.}~\bibnamefont {Schittko}}, \bibinfo {author} {\bibfnamefont {P.~M.}\
			\bibnamefont {Preiss}}, \bibinfo {author} {\bibfnamefont {A.~M.}\
			\bibnamefont {Kaufman}}, \bibinfo {author} {\bibfnamefont {M.}~\bibnamefont
			{Greiner}}, \bibinfo {author} {\bibfnamefont {H.}~\bibnamefont {Pichler}},\
		and\ \bibinfo {author} {\bibfnamefont {P.}~\bibnamefont {Hayden}},\
	}\bibfield  {title} {\bibinfo {title} {Quantum virtual cooling},\ }\href
	{https://doi.org/10.1103/PhysRevX.9.031013} {\bibfield  {journal} {\bibinfo
			{journal} {Phys. Rev. X}\ }\textbf {\bibinfo {volume} {9}},\ \bibinfo {pages}
		{031013} (\bibinfo {year} {2019})}\BibitemShut {NoStop}%
	\bibitem [{\citenamefont {Cai}\ \emph {et~al.}(2023)\citenamefont {Cai},
		\citenamefont {Babbush}, \citenamefont {Benjamin}, \citenamefont {Endo},
		\citenamefont {Huggins}, \citenamefont {Li}, \citenamefont {McClean},\ and\
		\citenamefont {O'Brien}}]{cai2023qem}%
	\BibitemOpen
	\bibfield  {author} {\bibinfo {author} {\bibfnamefont {Z.}~\bibnamefont
			{Cai}}, \bibinfo {author} {\bibfnamefont {R.}~\bibnamefont {Babbush}},
		\bibinfo {author} {\bibfnamefont {S.~C.}\ \bibnamefont {Benjamin}}, \bibinfo
		{author} {\bibfnamefont {S.}~\bibnamefont {Endo}}, \bibinfo {author}
		{\bibfnamefont {W.~J.}\ \bibnamefont {Huggins}}, \bibinfo {author}
		{\bibfnamefont {Y.}~\bibnamefont {Li}}, \bibinfo {author} {\bibfnamefont
			{J.~R.}\ \bibnamefont {McClean}},\ and\ \bibinfo {author} {\bibfnamefont
			{T.~E.}\ \bibnamefont {O'Brien}},\ }\bibfield  {title} {\bibinfo {title}
		{Quantum error mitigation},\ }\href
	{https://doi.org/10.1103/RevModPhys.95.045005} {\bibfield  {journal}
		{\bibinfo  {journal} {Rev. Mod. Phys.}\ }\textbf {\bibinfo {volume} {95}},\
		\bibinfo {pages} {045005} (\bibinfo {year} {2023})}\BibitemShut {NoStop}%
	\bibitem [{\citenamefont
		{Koczor}(2021)}]{koczorExponentialErrorSuppression2021}%
	\BibitemOpen
	\bibfield  {author} {\bibinfo {author} {\bibfnamefont {B.}~\bibnamefont
			{Koczor}},\ }\bibfield  {title} {\bibinfo {title} {Exponential error
			suppression for near-term quantum devices},\ }\href
	{https://doi.org/10.1103/PhysRevX.11.031057} {\bibfield  {journal} {\bibinfo
			{journal} {Phys. Rev. X}\ }\textbf {\bibinfo {volume} {11}},\ \bibinfo
		{pages} {031057} (\bibinfo {year} {2021})}\BibitemShut {NoStop}%
	\bibitem [{\citenamefont {Huggins}\ \emph {et~al.}(2021)\citenamefont
		{Huggins}, \citenamefont {McArdle}, \citenamefont {O'Brien}, \citenamefont
		{Lee}, \citenamefont {Rubin}, \citenamefont {Boixo}, \citenamefont {Whaley},
		\citenamefont {Babbush},\ and\ \citenamefont
		{McClean}}]{hugginsVirtualDistillationQuantum2021}%
	\BibitemOpen
	\bibfield  {author} {\bibinfo {author} {\bibfnamefont {W.~J.}\ \bibnamefont
			{Huggins}}, \bibinfo {author} {\bibfnamefont {S.}~\bibnamefont {McArdle}},
		\bibinfo {author} {\bibfnamefont {T.~E.}\ \bibnamefont {O'Brien}}, \bibinfo
		{author} {\bibfnamefont {J.}~\bibnamefont {Lee}}, \bibinfo {author}
		{\bibfnamefont {N.~C.}\ \bibnamefont {Rubin}}, \bibinfo {author}
		{\bibfnamefont {S.}~\bibnamefont {Boixo}}, \bibinfo {author} {\bibfnamefont
			{K.~B.}\ \bibnamefont {Whaley}}, \bibinfo {author} {\bibfnamefont
			{R.}~\bibnamefont {Babbush}},\ and\ \bibinfo {author} {\bibfnamefont {J.~R.}\
			\bibnamefont {McClean}},\ }\bibfield  {title} {\bibinfo {title} {Virtual
			{{Distillation}} for {{Quantum Error Mitigation}}},\ }\href
	{https://doi.org/10.1103/PhysRevX.11.041036} {\bibfield  {journal} {\bibinfo
			{journal} {Phys. Rev. X}\ }\textbf {\bibinfo {volume} {11}},\ \bibinfo
		{pages} {041036} (\bibinfo {year} {2021})}\BibitemShut {NoStop}%
	\bibitem [{\citenamefont {Lloyd}\ \emph {et~al.}(2014)\citenamefont {Lloyd},
		\citenamefont {Mohseni},\ and\ \citenamefont {Rebentrost}}]{Lloyd2014qpca}%
	\BibitemOpen
	\bibfield  {author} {\bibinfo {author} {\bibfnamefont {S.}~\bibnamefont
			{Lloyd}}, \bibinfo {author} {\bibfnamefont {M.}~\bibnamefont {Mohseni}},\
		and\ \bibinfo {author} {\bibfnamefont {P.}~\bibnamefont {Rebentrost}},\
	}\bibfield  {title} {\bibinfo {title} {Quantum principal component
			analysis},\ }\href {https://doi.org/10.1038/nphys3029} {\bibfield  {journal}
		{\bibinfo  {journal} {Nat. Phys.}\ }\textbf {\bibinfo {volume} {10}},\
		\bibinfo {pages} {631} (\bibinfo {year} {2014})}\BibitemShut {NoStop}%
	\bibitem [{\citenamefont {Kimmel}\ \emph {et~al.}(2017)\citenamefont {Kimmel},
		\citenamefont {Lin}, \citenamefont {Low}, \citenamefont {Ozols},\ and\
		\citenamefont {Yoder}}]{Kimmel2017hamiltonian}%
	\BibitemOpen
	\bibfield  {author} {\bibinfo {author} {\bibfnamefont {S.}~\bibnamefont
			{Kimmel}}, \bibinfo {author} {\bibfnamefont {C.~Y.-Y.}\ \bibnamefont {Lin}},
		\bibinfo {author} {\bibfnamefont {G.~H.}\ \bibnamefont {Low}}, \bibinfo
		{author} {\bibfnamefont {M.}~\bibnamefont {Ozols}},\ and\ \bibinfo {author}
		{\bibfnamefont {T.~J.}\ \bibnamefont {Yoder}},\ }\bibfield  {title} {\bibinfo
		{title} {{H}amiltonian simulation with optimal sample complexity},\ }\href
	{https://doi.org/10.1038/s41534-017-0013-7} {\bibfield  {journal} {\bibinfo
			{journal} {npj Quantum Inf.}\ }\textbf {\bibinfo {volume} {3}},\ \bibinfo
		{pages} {13} (\bibinfo {year} {2017})}\BibitemShut {NoStop}%
	\bibitem [{\citenamefont {Braunstein}\ and\ \citenamefont
		{Caves}(1994)}]{braunstein1994fisher}%
	\BibitemOpen
	\bibfield  {author} {\bibinfo {author} {\bibfnamefont {S.~L.}\ \bibnamefont
			{Braunstein}}\ and\ \bibinfo {author} {\bibfnamefont {C.~M.}\ \bibnamefont
			{Caves}},\ }\bibfield  {title} {\bibinfo {title} {Statistical distance and
			the geometry of quantum states},\ }\href
	{https://doi.org/10.1103/PhysRevLett.72.3439} {\bibfield  {journal} {\bibinfo
			{journal} {Phys. Rev. Lett.}\ }\textbf {\bibinfo {volume} {72}},\ \bibinfo
		{pages} {3439} (\bibinfo {year} {1994})}\BibitemShut {NoStop}%
	\bibitem [{\citenamefont {Giovannetti}\ \emph {et~al.}(2011)\citenamefont
		{Giovannetti}, \citenamefont {Lloyd},\ and\ \citenamefont
		{Maccone}}]{Giovannetti2011metrology}%
	\BibitemOpen
	\bibfield  {author} {\bibinfo {author} {\bibfnamefont {V.}~\bibnamefont
			{Giovannetti}}, \bibinfo {author} {\bibfnamefont {S.}~\bibnamefont {Lloyd}},\
		and\ \bibinfo {author} {\bibfnamefont {L.}~\bibnamefont {Maccone}},\
	}\bibfield  {title} {\bibinfo {title} {Advances in quantum metrology},\
	}\href {https://doi.org/10.1038/nphoton.2011.35} {\bibfield  {journal}
		{\bibinfo  {journal} {Nat. Photonics}\ }\textbf {\bibinfo {volume} {5}},\
		\bibinfo {pages} {222} (\bibinfo {year} {2011})}\BibitemShut {NoStop}%
	\bibitem [{\citenamefont {Montanaro}\ and\ \citenamefont
		{de~Wolf}(2013)}]{montanaro2013survey}%
	\BibitemOpen
	\bibfield  {author} {\bibinfo {author} {\bibfnamefont {A.}~\bibnamefont
			{Montanaro}}\ and\ \bibinfo {author} {\bibfnamefont {R.}~\bibnamefont
			{de~Wolf}},\ }\bibfield  {title} {\bibinfo {title} {A survey of quantum
			property testing},\ }\href {https://arxiv.org/abs/1310.2035} {\bibfield
		{journal} {\bibinfo  {journal} {arXiv:1310.2035}\ } (\bibinfo {year}
		{2013})}\BibitemShut {NoStop}%
	\bibitem [{\citenamefont {Chen}\ \emph
		{et~al.}(2023{\natexlab{a}})\citenamefont {Chen}, \citenamefont {Wang},
		\citenamefont {Long},\ and\ \citenamefont {Ying}}]{chen2023unitarity}%
	\BibitemOpen
	\bibfield  {author} {\bibinfo {author} {\bibfnamefont {K.}~\bibnamefont
			{Chen}}, \bibinfo {author} {\bibfnamefont {Q.}~\bibnamefont {Wang}}, \bibinfo
		{author} {\bibfnamefont {P.}~\bibnamefont {Long}},\ and\ \bibinfo {author}
		{\bibfnamefont {M.}~\bibnamefont {Ying}},\ }\bibfield  {title} {\bibinfo
		{title} {Unitarity estimation for quantum channels},\ }\href
	{https://doi.org/10.1109/TIT.2023.3263645} {\bibfield  {journal} {\bibinfo
			{journal} {IEEE Trans. Inf. Theory}\ }\textbf {\bibinfo {volume} {69}},\
		\bibinfo {pages} {5116} (\bibinfo {year} {2023}{\natexlab{a}})}\BibitemShut
	{NoStop}%
	\bibitem [{\citenamefont {Liu}\ \emph {et~al.}(2024)\citenamefont {Liu},
		\citenamefont {Zhang}, \citenamefont {Fei},\ and\ \citenamefont
		{Cai}}]{liu2024virtual}%
	\BibitemOpen
	\bibfield  {author} {\bibinfo {author} {\bibfnamefont {Z.}~\bibnamefont
			{Liu}}, \bibinfo {author} {\bibfnamefont {X.}~\bibnamefont {Zhang}}, \bibinfo
		{author} {\bibfnamefont {Y.-Y.}\ \bibnamefont {Fei}},\ and\ \bibinfo {author}
		{\bibfnamefont {Z.}~\bibnamefont {Cai}},\ }\bibfield  {title} {\bibinfo
		{title} {Virtual channel purification},\ }\href
	{https://arxiv.org/abs/2402.07866} {\bibfield  {journal} {\bibinfo  {journal}
			{arXiv:2402.07866}\ } (\bibinfo {year} {2024})}\BibitemShut {NoStop}%
	\bibitem [{\citenamefont {Elben}\ \emph {et~al.}(2019)\citenamefont {Elben},
		\citenamefont {Vermersch}, \citenamefont {Roos},\ and\ \citenamefont
		{Zoller}}]{elben2019statistical}%
	\BibitemOpen
	\bibfield  {author} {\bibinfo {author} {\bibfnamefont {A.}~\bibnamefont
			{Elben}}, \bibinfo {author} {\bibfnamefont {B.}~\bibnamefont {Vermersch}},
		\bibinfo {author} {\bibfnamefont {C.~F.}\ \bibnamefont {Roos}},\ and\
		\bibinfo {author} {\bibfnamefont {P.}~\bibnamefont {Zoller}},\ }\bibfield
	{title} {\bibinfo {title} {Statistical correlations between locally
			randomized measurements: A toolbox for probing entanglement in many-body
			quantum states},\ }\href {https://doi.org/10.1103/PhysRevA.99.052323}
	{\bibfield  {journal} {\bibinfo  {journal} {Phys. Rev. A}\ }\textbf {\bibinfo
			{volume} {99}},\ \bibinfo {pages} {052323} (\bibinfo {year}
		{2019})}\BibitemShut {NoStop}%
	\bibitem [{\citenamefont {Kueng}\ \emph {et~al.}(2017)\citenamefont {Kueng},
		\citenamefont {Rauhut},\ and\ \citenamefont {Terstiege}}]{KUENG2017tomo}%
	\BibitemOpen
	\bibfield  {author} {\bibinfo {author} {\bibfnamefont {R.}~\bibnamefont
			{Kueng}}, \bibinfo {author} {\bibfnamefont {H.}~\bibnamefont {Rauhut}},\ and\
		\bibinfo {author} {\bibfnamefont {U.}~\bibnamefont {Terstiege}},\ }\bibfield
	{title} {\bibinfo {title} {Low rank matrix recovery from rank one
			measurements},\ }\href
	{https://doi.org/https://doi.org/10.1016/j.acha.2015.07.007} {\bibfield
		{journal} {\bibinfo  {journal} {Appl. Comput. Harmon. Anal.}\ }\textbf
		{\bibinfo {volume} {42}},\ \bibinfo {pages} {88} (\bibinfo {year}
		{2017})}\BibitemShut {NoStop}%
	\bibitem [{\citenamefont {Chen}\ \emph
		{et~al.}(2023{\natexlab{b}})\citenamefont {Chen}, \citenamefont {Huang},
		\citenamefont {Li}, \citenamefont {Liu},\ and\ \citenamefont
		{Sellke}}]{chen2023adaptivity}%
	\BibitemOpen
	\bibfield  {author} {\bibinfo {author} {\bibfnamefont {S.}~\bibnamefont
			{Chen}}, \bibinfo {author} {\bibfnamefont {B.}~\bibnamefont {Huang}},
		\bibinfo {author} {\bibfnamefont {J.}~\bibnamefont {Li}}, \bibinfo {author}
		{\bibfnamefont {A.}~\bibnamefont {Liu}},\ and\ \bibinfo {author}
		{\bibfnamefont {M.}~\bibnamefont {Sellke}},\ }\bibfield  {title} {\bibinfo
		{title} {When does adaptivity help for quantum state learning?},\ }in\ \href
	{https://doi.org/10.1109/FOCS57990.2023.00029} {\emph {\bibinfo {booktitle}
			{Proceedings of 2023 IEEE 64th Annual Symposium on Foundations of Computer
				Science (FOCS)}}}\ (\bibinfo {year} {2023})\ pp.\ \bibinfo {pages}
	{391--404}\BibitemShut {NoStop}%
	\bibitem [{\citenamefont {Zhou}\ and\ \citenamefont
		{Liu}(2024)}]{Zhou2024hybrid}%
	\BibitemOpen
	\bibfield  {author} {\bibinfo {author} {\bibfnamefont {Y.}~\bibnamefont
			{Zhou}}\ and\ \bibinfo {author} {\bibfnamefont {Z.}~\bibnamefont {Liu}},\
	}\bibfield  {title} {\bibinfo {title} {A hybrid framework for estimating
			nonlinear functions of quantum states},\ }\href
	{https://doi.org/10.1038/s41534-024-00846-5} {\bibfield  {journal} {\bibinfo
			{journal} {npj Quantum Inf.}\ }\textbf {\bibinfo {volume} {10}},\ \bibinfo
		{pages} {62} (\bibinfo {year} {2024})}\BibitemShut {NoStop}%
	\bibitem [{\citenamefont {Huang}\ \emph {et~al.}(2020)\citenamefont {Huang},
		\citenamefont {Kueng},\ and\ \citenamefont {Preskill}}]{Huang2020predicting}%
	\BibitemOpen
	\bibfield  {author} {\bibinfo {author} {\bibfnamefont {H.-Y.}\ \bibnamefont
			{Huang}}, \bibinfo {author} {\bibfnamefont {R.}~\bibnamefont {Kueng}},\ and\
		\bibinfo {author} {\bibfnamefont {J.}~\bibnamefont {Preskill}},\ }\bibfield
	{title} {\bibinfo {title} {Predicting many properties of a quantum system
			from very few measurements},\ }\href
	{https://doi.org/10.1038/s41567-020-0932-7} {\bibfield  {journal} {\bibinfo
			{journal} {Nat. Phys.}\ }\textbf {\bibinfo {volume} {16}},\ \bibinfo {pages}
		{1050} (\bibinfo {year} {2020})}\BibitemShut {NoStop}%
	\bibitem [{\citenamefont {Hu}\ \emph {et~al.}(2022)\citenamefont {Hu},
		\citenamefont {LaRose}, \citenamefont {You}, \citenamefont {Rieffel},\ and\
		\citenamefont {Wang}}]{hu2022logical}%
	\BibitemOpen
	\bibfield  {author} {\bibinfo {author} {\bibfnamefont {H.-Y.}\ \bibnamefont
			{Hu}}, \bibinfo {author} {\bibfnamefont {R.}~\bibnamefont {LaRose}}, \bibinfo
		{author} {\bibfnamefont {Y.-Z.}\ \bibnamefont {You}}, \bibinfo {author}
		{\bibfnamefont {E.}~\bibnamefont {Rieffel}},\ and\ \bibinfo {author}
		{\bibfnamefont {Z.}~\bibnamefont {Wang}},\ }\bibfield  {title} {\bibinfo
		{title} {Logical shadow tomography: Efficient estimation of error-mitigated
			observables},\ }\href {https://arxiv.org/abs/2203.07263} {\bibfield
		{journal} {\bibinfo  {journal} {arXiv:2203.07263}\ } (\bibinfo {year}
		{2022})}\BibitemShut {NoStop}%
	\bibitem [{\citenamefont {Seif}\ \emph {et~al.}(2023)\citenamefont {Seif},
		\citenamefont {Cian}, \citenamefont {Zhou}, \citenamefont {Chen},\ and\
		\citenamefont {Jiang}}]{seif2023shadow_dist}%
	\BibitemOpen
	\bibfield  {author} {\bibinfo {author} {\bibfnamefont {A.}~\bibnamefont
			{Seif}}, \bibinfo {author} {\bibfnamefont {Z.-P.}\ \bibnamefont {Cian}},
		\bibinfo {author} {\bibfnamefont {S.}~\bibnamefont {Zhou}}, \bibinfo {author}
		{\bibfnamefont {S.}~\bibnamefont {Chen}},\ and\ \bibinfo {author}
		{\bibfnamefont {L.}~\bibnamefont {Jiang}},\ }\bibfield  {title} {\bibinfo
		{title} {Shadow distillation: Quantum error mitigation with classical shadows
			for near-term quantum processors},\ }\href
	{https://doi.org/10.1103/PRXQuantum.4.010303} {\bibfield  {journal} {\bibinfo
			{journal} {PRX Quantum}\ }\textbf {\bibinfo {volume} {4}},\ \bibinfo {pages}
		{010303} (\bibinfo {year} {2023})}\BibitemShut {NoStop}%
	\bibitem [{\citenamefont {Wei}\ \emph {et~al.}(2023)\citenamefont {Wei},
		\citenamefont {Liu}, \citenamefont {Liu}, \citenamefont {Han}, \citenamefont
		{Ma}, \citenamefont {Deng},\ and\ \citenamefont {Liu}}]{wei2023realizing}%
	\BibitemOpen
	\bibfield  {author} {\bibinfo {author} {\bibfnamefont {F.}~\bibnamefont
			{Wei}}, \bibinfo {author} {\bibfnamefont {Z.}~\bibnamefont {Liu}}, \bibinfo
		{author} {\bibfnamefont {G.}~\bibnamefont {Liu}}, \bibinfo {author}
		{\bibfnamefont {Z.}~\bibnamefont {Han}}, \bibinfo {author} {\bibfnamefont
			{X.}~\bibnamefont {Ma}}, \bibinfo {author} {\bibfnamefont {D.-L.}\
			\bibnamefont {Deng}},\ and\ \bibinfo {author} {\bibfnamefont
			{Z.}~\bibnamefont {Liu}},\ }\bibfield  {title} {\bibinfo {title} {Realizing
			non-physical actions through hermitian-preserving map exponentiation},\
	}\href {https://arxiv.org/abs/2308.07956} {\bibfield  {journal} {\bibinfo
			{journal} {arXiv:2308.07956}\ } (\bibinfo {year} {2023})}\BibitemShut
	{NoStop}%
	\bibitem [{\citenamefont {Gordon}\ \emph {et~al.}(2022)\citenamefont {Gordon},
		\citenamefont {Cerezo}, \citenamefont {Cincio},\ and\ \citenamefont
		{Coles}}]{gordon2022covariance}%
	\BibitemOpen
	\bibfield  {author} {\bibinfo {author} {\bibfnamefont {M.~H.}\ \bibnamefont
			{Gordon}}, \bibinfo {author} {\bibfnamefont {M.}~\bibnamefont {Cerezo}},
		\bibinfo {author} {\bibfnamefont {L.}~\bibnamefont {Cincio}},\ and\ \bibinfo
		{author} {\bibfnamefont {P.~J.}\ \bibnamefont {Coles}},\ }\bibfield  {title}
	{\bibinfo {title} {Covariance matrix preparation for quantum principal
			component analysis},\ }\href {https://doi.org/10.1103/PRXQuantum.3.030334}
	{\bibfield  {journal} {\bibinfo  {journal} {PRX Quantum}\ }\textbf {\bibinfo
			{volume} {3}},\ \bibinfo {pages} {030334} (\bibinfo {year}
		{2022})}\BibitemShut {NoStop}%
	\bibitem [{\citenamefont {Uola}\ \emph {et~al.}(2020)\citenamefont {Uola},
		\citenamefont {Costa}, \citenamefont {Nguyen},\ and\ \citenamefont
		{G\"uhne}}]{uola2020steering}%
	\BibitemOpen
	\bibfield  {author} {\bibinfo {author} {\bibfnamefont {R.}~\bibnamefont
			{Uola}}, \bibinfo {author} {\bibfnamefont {A.~C.~S.}\ \bibnamefont {Costa}},
		\bibinfo {author} {\bibfnamefont {H.~C.}\ \bibnamefont {Nguyen}},\ and\
		\bibinfo {author} {\bibfnamefont {O.}~\bibnamefont {G\"uhne}},\ }\bibfield
	{title} {\bibinfo {title} {Quantum steering},\ }\href
	{https://doi.org/10.1103/RevModPhys.92.015001} {\bibfield  {journal}
		{\bibinfo  {journal} {Rev. Mod. Phys.}\ }\textbf {\bibinfo {volume} {92}},\
		\bibinfo {pages} {015001} (\bibinfo {year} {2020})}\BibitemShut {NoStop}%
	\bibitem [{\citenamefont {Hyllus}\ \emph {et~al.}(2012)\citenamefont {Hyllus},
		\citenamefont {Laskowski}, \citenamefont {Krischek}, \citenamefont
		{Schwemmer}, \citenamefont {Wieczorek}, \citenamefont {Weinfurter},
		\citenamefont {Pezz\'e},\ and\ \citenamefont {Smerzi}}]{hyllus2012fisher}%
	\BibitemOpen
	\bibfield  {author} {\bibinfo {author} {\bibfnamefont {P.}~\bibnamefont
			{Hyllus}}, \bibinfo {author} {\bibfnamefont {W.}~\bibnamefont {Laskowski}},
		\bibinfo {author} {\bibfnamefont {R.}~\bibnamefont {Krischek}}, \bibinfo
		{author} {\bibfnamefont {C.}~\bibnamefont {Schwemmer}}, \bibinfo {author}
		{\bibfnamefont {W.}~\bibnamefont {Wieczorek}}, \bibinfo {author}
		{\bibfnamefont {H.}~\bibnamefont {Weinfurter}}, \bibinfo {author}
		{\bibfnamefont {L.}~\bibnamefont {Pezz\'e}},\ and\ \bibinfo {author}
		{\bibfnamefont {A.}~\bibnamefont {Smerzi}},\ }\bibfield  {title} {\bibinfo
		{title} {Fisher information and multiparticle entanglement},\ }\href
	{https://doi.org/10.1103/PhysRevA.85.022321} {\bibfield  {journal} {\bibinfo
			{journal} {Phys. Rev. A}\ }\textbf {\bibinfo {volume} {85}},\ \bibinfo
		{pages} {022321} (\bibinfo {year} {2012})}\BibitemShut {NoStop}%
	\bibitem [{\citenamefont {T\'oth}(2012)}]{toth2012fisher}%
	\BibitemOpen
	\bibfield  {author} {\bibinfo {author} {\bibfnamefont {G.}~\bibnamefont
			{T\'oth}},\ }\bibfield  {title} {\bibinfo {title} {Multipartite entanglement
			and high-precision metrology},\ }\href
	{https://doi.org/10.1103/PhysRevA.85.022322} {\bibfield  {journal} {\bibinfo
			{journal} {Phys. Rev. A}\ }\textbf {\bibinfo {volume} {85}},\ \bibinfo
		{pages} {022322} (\bibinfo {year} {2012})}\BibitemShut {NoStop}%
	\bibitem [{\citenamefont {Li}\ and\ \citenamefont {Luo}(2013)}]{li2013fisher}%
	\BibitemOpen
	\bibfield  {author} {\bibinfo {author} {\bibfnamefont {N.}~\bibnamefont
			{Li}}\ and\ \bibinfo {author} {\bibfnamefont {S.}~\bibnamefont {Luo}},\
	}\bibfield  {title} {\bibinfo {title} {Entanglement detection via quantum
			fisher information},\ }\href {https://doi.org/10.1103/PhysRevA.88.014301}
	{\bibfield  {journal} {\bibinfo  {journal} {Phys. Rev. A}\ }\textbf {\bibinfo
			{volume} {88}},\ \bibinfo {pages} {014301} (\bibinfo {year}
		{2013})}\BibitemShut {NoStop}%
	\bibitem [{\citenamefont {Tóth}\ and\ \citenamefont
		{Apellaniz}(2014)}]{toth2014fisher}%
	\BibitemOpen
	\bibfield  {author} {\bibinfo {author} {\bibfnamefont {G.}~\bibnamefont
			{Tóth}}\ and\ \bibinfo {author} {\bibfnamefont {I.}~\bibnamefont
			{Apellaniz}},\ }\bibfield  {title} {\bibinfo {title} {Quantum metrology from
			a quantum information science perspective},\ }\href
	{https://doi.org/10.1088/1751-8113/47/42/424006} {\bibfield  {journal}
		{\bibinfo  {journal} {Journal of Physics A: Mathematical and Theoretical}\
		}\textbf {\bibinfo {volume} {47}},\ \bibinfo {pages} {424006} (\bibinfo
		{year} {2014})}\BibitemShut {NoStop}%
	\bibitem [{\citenamefont {Ren}\ \emph {et~al.}(2021)\citenamefont {Ren},
		\citenamefont {Li}, \citenamefont {Smerzi},\ and\ \citenamefont
		{Gessner}}]{ren2021metrology}%
	\BibitemOpen
	\bibfield  {author} {\bibinfo {author} {\bibfnamefont {Z.}~\bibnamefont
			{Ren}}, \bibinfo {author} {\bibfnamefont {W.}~\bibnamefont {Li}}, \bibinfo
		{author} {\bibfnamefont {A.}~\bibnamefont {Smerzi}},\ and\ \bibinfo {author}
		{\bibfnamefont {M.}~\bibnamefont {Gessner}},\ }\bibfield  {title} {\bibinfo
		{title} {Metrological detection of multipartite entanglement from young
			diagrams},\ }\href {https://doi.org/10.1103/PhysRevLett.126.080502}
	{\bibfield  {journal} {\bibinfo  {journal} {Phys. Rev. Lett.}\ }\textbf
		{\bibinfo {volume} {126}},\ \bibinfo {pages} {080502} (\bibinfo {year}
		{2021})}\BibitemShut {NoStop}%
	\bibitem [{\citenamefont {Yu}\ \emph {et~al.}(2021)\citenamefont {Yu},
		\citenamefont {Li}, \citenamefont {Wang}, \citenamefont {Chu}, \citenamefont
		{Yang}, \citenamefont {Gong}, \citenamefont {Goldman},\ and\ \citenamefont
		{Cai}}]{yu2021fisher}%
	\BibitemOpen
	\bibfield  {author} {\bibinfo {author} {\bibfnamefont {M.}~\bibnamefont
			{Yu}}, \bibinfo {author} {\bibfnamefont {D.}~\bibnamefont {Li}}, \bibinfo
		{author} {\bibfnamefont {J.}~\bibnamefont {Wang}}, \bibinfo {author}
		{\bibfnamefont {Y.}~\bibnamefont {Chu}}, \bibinfo {author} {\bibfnamefont
			{P.}~\bibnamefont {Yang}}, \bibinfo {author} {\bibfnamefont {M.}~\bibnamefont
			{Gong}}, \bibinfo {author} {\bibfnamefont {N.}~\bibnamefont {Goldman}},\ and\
		\bibinfo {author} {\bibfnamefont {J.}~\bibnamefont {Cai}},\ }\bibfield
	{title} {\bibinfo {title} {Experimental estimation of the quantum fisher
			information from randomized measurements},\ }\href
	{https://doi.org/10.1103/PhysRevResearch.3.043122} {\bibfield  {journal}
		{\bibinfo  {journal} {Phys. Rev. Res.}\ }\textbf {\bibinfo {volume} {3}},\
		\bibinfo {pages} {043122} (\bibinfo {year} {2021})}\BibitemShut {NoStop}%
	\bibitem [{\citenamefont {Rath}\ \emph {et~al.}(2021)\citenamefont {Rath},
		\citenamefont {Branciard}, \citenamefont {Minguzzi},\ and\ \citenamefont
		{Vermersch}}]{rath2021fisher}%
	\BibitemOpen
	\bibfield  {author} {\bibinfo {author} {\bibfnamefont {A.}~\bibnamefont
			{Rath}}, \bibinfo {author} {\bibfnamefont {C.}~\bibnamefont {Branciard}},
		\bibinfo {author} {\bibfnamefont {A.}~\bibnamefont {Minguzzi}},\ and\
		\bibinfo {author} {\bibfnamefont {B.}~\bibnamefont {Vermersch}},\ }\bibfield
	{title} {\bibinfo {title} {Quantum fisher information from randomized
			measurements},\ }\href {https://doi.org/10.1103/PhysRevLett.127.260501}
	{\bibfield  {journal} {\bibinfo  {journal} {Phys. Rev. Lett.}\ }\textbf
		{\bibinfo {volume} {127}},\ \bibinfo {pages} {260501} (\bibinfo {year}
		{2021})}\BibitemShut {NoStop}%
	\bibitem [{\citenamefont {Ji}\ \emph {et~al.}(2018)\citenamefont {Ji},
		\citenamefont {Liu},\ and\ \citenamefont {Song}}]{ji2018pseudo}%
	\BibitemOpen
	\bibfield  {author} {\bibinfo {author} {\bibfnamefont {Z.}~\bibnamefont
			{Ji}}, \bibinfo {author} {\bibfnamefont {Y.-K.}\ \bibnamefont {Liu}},\ and\
		\bibinfo {author} {\bibfnamefont {F.}~\bibnamefont {Song}},\ }\bibfield
	{title} {\bibinfo {title} {Pseudorandom quantum states},\ }in\ \href@noop {}
	{\emph {\bibinfo {booktitle} {Advances in Cryptology -- CRYPTO 2018}}},\
	\bibinfo {editor} {edited by\ \bibinfo {editor} {\bibfnamefont
			{H.}~\bibnamefont {Shacham}}\ and\ \bibinfo {editor} {\bibfnamefont
			{A.}~\bibnamefont {Boldyreva}}}\ (\bibinfo  {publisher} {Springer
		International Publishing},\ \bibinfo {address} {Cham},\ \bibinfo {year}
	{2018})\ pp.\ \bibinfo {pages} {126--152}\BibitemShut {NoStop}%
	\bibitem [{\citenamefont {Schuster}\ \emph {et~al.}(2024)\citenamefont
		{Schuster}, \citenamefont {Haferkamp},\ and\ \citenamefont
		{Huang}}]{schuster2024random}%
	\BibitemOpen
	\bibfield  {author} {\bibinfo {author} {\bibfnamefont {T.}~\bibnamefont
			{Schuster}}, \bibinfo {author} {\bibfnamefont {J.}~\bibnamefont
			{Haferkamp}},\ and\ \bibinfo {author} {\bibfnamefont {H.-Y.}\ \bibnamefont
			{Huang}},\ }\bibfield  {title} {\bibinfo {title} {Random unitaries in
			extremely low depth},\ }\href {https://arxiv.org/abs/2407.07754} {\bibfield
		{journal} {\bibinfo  {journal} {arXiv:2407.07754}\ } (\bibinfo {year}
		{2024})}\BibitemShut {NoStop}%
	\bibitem [{\citenamefont {Broadbent}\ \emph {et~al.}(2009)\citenamefont
		{Broadbent}, \citenamefont {Fitzsimons},\ and\ \citenamefont
		{Kashefi}}]{broadbent2009bqc}%
	\BibitemOpen
	\bibfield  {author} {\bibinfo {author} {\bibfnamefont {A.}~\bibnamefont
			{Broadbent}}, \bibinfo {author} {\bibfnamefont {J.}~\bibnamefont
			{Fitzsimons}},\ and\ \bibinfo {author} {\bibfnamefont {E.}~\bibnamefont
			{Kashefi}},\ }\bibfield  {title} {\bibinfo {title} {Universal blind quantum
			computation},\ }in\ \href {https://doi.org/10.1109/FOCS.2009.36} {\emph
		{\bibinfo {booktitle} {Proceedings of 2009 50th Annual IEEE Symposium on
				Foundations of Computer Science}}}\ (\bibinfo {year} {2009})\ pp.\ \bibinfo
	{pages} {517--526}\BibitemShut {NoStop}%
	\bibitem [{\citenamefont {Barz}\ \emph {et~al.}(2012)\citenamefont {Barz},
		\citenamefont {Kashefi}, \citenamefont {Broadbent}, \citenamefont
		{Fitzsimons}, \citenamefont {Zeilinger},\ and\ \citenamefont
		{Walther}}]{stefanie2012BQC}%
	\BibitemOpen
	\bibfield  {author} {\bibinfo {author} {\bibfnamefont {S.}~\bibnamefont
			{Barz}}, \bibinfo {author} {\bibfnamefont {E.}~\bibnamefont {Kashefi}},
		\bibinfo {author} {\bibfnamefont {A.}~\bibnamefont {Broadbent}}, \bibinfo
		{author} {\bibfnamefont {J.~F.}\ \bibnamefont {Fitzsimons}}, \bibinfo
		{author} {\bibfnamefont {A.}~\bibnamefont {Zeilinger}},\ and\ \bibinfo
		{author} {\bibfnamefont {P.}~\bibnamefont {Walther}},\ }\bibfield  {title}
	{\bibinfo {title} {Demonstration of blind quantum computing},\ }\href
	{https://doi.org/10.1126/science.1214707} {\bibfield  {journal} {\bibinfo
			{journal} {Science}\ }\textbf {\bibinfo {volume} {335}},\ \bibinfo {pages}
		{303} (\bibinfo {year} {2012})}\BibitemShut {NoStop}%
	\bibitem [{\citenamefont {Fitzsimons}(2017)}]{Fitzsimons2017bqc}%
	\BibitemOpen
	\bibfield  {author} {\bibinfo {author} {\bibfnamefont {J.~F.}\ \bibnamefont
			{Fitzsimons}},\ }\bibfield  {title} {\bibinfo {title} {Private quantum
			computation: an introduction to blind quantum computing and related
			protocols},\ }\href {https://doi.org/10.1038/s41534-017-0025-3} {\bibfield
		{journal} {\bibinfo  {journal} {npj Quantum Inf.}\ }\textbf {\bibinfo
			{volume} {3}},\ \bibinfo {pages} {23} (\bibinfo {year} {2017})}\BibitemShut
	{NoStop}%
	\bibitem [{\citenamefont {Tran}\ \emph {et~al.}(2023)\citenamefont {Tran},
		\citenamefont {Mark}, \citenamefont {Ho},\ and\ \citenamefont
		{Choi}}]{tran2023measuring}%
	\BibitemOpen
	\bibfield  {author} {\bibinfo {author} {\bibfnamefont {M.~C.}\ \bibnamefont
			{Tran}}, \bibinfo {author} {\bibfnamefont {D.~K.}\ \bibnamefont {Mark}},
		\bibinfo {author} {\bibfnamefont {W.~W.}\ \bibnamefont {Ho}},\ and\ \bibinfo
		{author} {\bibfnamefont {S.}~\bibnamefont {Choi}},\ }\bibfield  {title}
	{\bibinfo {title} {Measuring arbitrary physical properties in analog quantum
			simulation},\ }\href {https://doi.org/10.1103/PhysRevX.13.011049} {\bibfield
		{journal} {\bibinfo  {journal} {Phys. Rev. X}\ }\textbf {\bibinfo {volume}
			{13}},\ \bibinfo {pages} {011049} (\bibinfo {year} {2023})}\BibitemShut
	{NoStop}%
	\bibitem [{\citenamefont {Kjaergaard}\ \emph {et~al.}(2022)\citenamefont
		{Kjaergaard}, \citenamefont {Schwartz}, \citenamefont {Greene}, \citenamefont
		{Samach}, \citenamefont {Bengtsson}, \citenamefont {O'Keeffe}, \citenamefont
		{McNally}, \citenamefont {Braum\"uller}, \citenamefont {Kim}, \citenamefont
		{Krantz}, \citenamefont {Marvian}, \citenamefont {Melville}, \citenamefont
		{Niedzielski}, \citenamefont {Sung}, \citenamefont {Winik}, \citenamefont
		{Yoder}, \citenamefont {Rosenberg}, \citenamefont {Obenland}, \citenamefont
		{Lloyd}, \citenamefont {Orlando}, \citenamefont {Marvian}, \citenamefont
		{Gustavsson},\ and\ \citenamefont {Oliver}}]{kjaergaard2022DME}%
	\BibitemOpen
	\bibfield  {author} {\bibinfo {author} {\bibfnamefont {M.}~\bibnamefont
			{Kjaergaard}}, \bibinfo {author} {\bibfnamefont {M.~E.}\ \bibnamefont
			{Schwartz}}, \bibinfo {author} {\bibfnamefont {A.}~\bibnamefont {Greene}},
		\bibinfo {author} {\bibfnamefont {G.~O.}\ \bibnamefont {Samach}}, \bibinfo
		{author} {\bibfnamefont {A.}~\bibnamefont {Bengtsson}}, \bibinfo {author}
		{\bibfnamefont {M.}~\bibnamefont {O'Keeffe}}, \bibinfo {author}
		{\bibfnamefont {C.~M.}\ \bibnamefont {McNally}}, \bibinfo {author}
		{\bibfnamefont {J.}~\bibnamefont {Braum\"uller}}, \bibinfo {author}
		{\bibfnamefont {D.~K.}\ \bibnamefont {Kim}}, \bibinfo {author} {\bibfnamefont
			{P.}~\bibnamefont {Krantz}}, \bibinfo {author} {\bibfnamefont
			{M.}~\bibnamefont {Marvian}}, \bibinfo {author} {\bibfnamefont
			{A.}~\bibnamefont {Melville}}, \bibinfo {author} {\bibfnamefont {B.~M.}\
			\bibnamefont {Niedzielski}}, \bibinfo {author} {\bibfnamefont
			{Y.}~\bibnamefont {Sung}}, \bibinfo {author} {\bibfnamefont {R.}~\bibnamefont
			{Winik}}, \bibinfo {author} {\bibfnamefont {J.}~\bibnamefont {Yoder}},
		\bibinfo {author} {\bibfnamefont {D.}~\bibnamefont {Rosenberg}}, \bibinfo
		{author} {\bibfnamefont {K.}~\bibnamefont {Obenland}}, \bibinfo {author}
		{\bibfnamefont {S.}~\bibnamefont {Lloyd}}, \bibinfo {author} {\bibfnamefont
			{T.~P.}\ \bibnamefont {Orlando}}, \bibinfo {author} {\bibfnamefont
			{I.}~\bibnamefont {Marvian}}, \bibinfo {author} {\bibfnamefont
			{S.}~\bibnamefont {Gustavsson}},\ and\ \bibinfo {author} {\bibfnamefont
			{W.~D.}\ \bibnamefont {Oliver}},\ }\bibfield  {title} {\bibinfo {title}
		{Demonstration of density matrix exponentiation using a superconducting
			quantum processor},\ }\href {https://doi.org/10.1103/PhysRevX.12.011005}
	{\bibfield  {journal} {\bibinfo  {journal} {Phys. Rev. X}\ }\textbf {\bibinfo
			{volume} {12}},\ \bibinfo {pages} {011005} (\bibinfo {year}
		{2022})}\BibitemShut {NoStop}%
	\bibitem [{\citenamefont {Chen}\ \emph
		{et~al.}(2024{\natexlab{d}})\citenamefont {Chen}, \citenamefont {Wang},\ and\
		\citenamefont {Zhang}}]{chen2024localtestunitarilyinvariant}%
	\BibitemOpen
	\bibfield  {author} {\bibinfo {author} {\bibfnamefont {K.}~\bibnamefont
			{Chen}}, \bibinfo {author} {\bibfnamefont {Q.}~\bibnamefont {Wang}},\ and\
		\bibinfo {author} {\bibfnamefont {Z.}~\bibnamefont {Zhang}},\ }\href
	{https://arxiv.org/abs/2404.04599} {\bibinfo {title} {Local test for
			unitarily invariant properties of bipartite quantum states}} (\bibinfo {year}
	{2024}{\natexlab{d}}),\ \Eprint {https://arxiv.org/abs/2404.04599}
	{arXiv:2404.04599 [quant-ph]} \BibitemShut {NoStop}%
	\bibitem [{\citenamefont {Blais}\ \emph {et~al.}(2021)\citenamefont {Blais},
		\citenamefont {Grimsmo}, \citenamefont {Girvin},\ and\ \citenamefont
		{Wallraff}}]{blais2021cqed}%
	\BibitemOpen
	\bibfield  {author} {\bibinfo {author} {\bibfnamefont {A.}~\bibnamefont
			{Blais}}, \bibinfo {author} {\bibfnamefont {A.~L.}\ \bibnamefont {Grimsmo}},
		\bibinfo {author} {\bibfnamefont {S.~M.}\ \bibnamefont {Girvin}},\ and\
		\bibinfo {author} {\bibfnamefont {A.}~\bibnamefont {Wallraff}},\ }\bibfield
	{title} {\bibinfo {title} {Circuit quantum electrodynamics},\ }\href
	{https://doi.org/10.1103/RevModPhys.93.025005} {\bibfield  {journal}
		{\bibinfo  {journal} {Rev. Mod. Phys.}\ }\textbf {\bibinfo {volume} {93}},\
		\bibinfo {pages} {025005} (\bibinfo {year} {2021})}\BibitemShut {NoStop}%
	\bibitem [{\citenamefont {Wang}\ \emph {et~al.}(2011)\citenamefont {Wang},
		\citenamefont {Ashhab},\ and\ \citenamefont {Nori}}]{wang2011open}%
	\BibitemOpen
	\bibfield  {author} {\bibinfo {author} {\bibfnamefont {H.}~\bibnamefont
			{Wang}}, \bibinfo {author} {\bibfnamefont {S.}~\bibnamefont {Ashhab}},\ and\
		\bibinfo {author} {\bibfnamefont {F.}~\bibnamefont {Nori}},\ }\bibfield
	{title} {\bibinfo {title} {Quantum algorithm for simulating the dynamics of
			an open quantum system},\ }\href {https://doi.org/10.1103/PhysRevA.83.062317}
	{\bibfield  {journal} {\bibinfo  {journal} {Phys. Rev. A}\ }\textbf {\bibinfo
			{volume} {83}},\ \bibinfo {pages} {062317} (\bibinfo {year}
		{2011})}\BibitemShut {NoStop}%
	\bibitem [{\citenamefont {Su}\ and\ \citenamefont {Li}(2020)}]{su2020open}%
	\BibitemOpen
	\bibfield  {author} {\bibinfo {author} {\bibfnamefont {H.-Y.}\ \bibnamefont
			{Su}}\ and\ \bibinfo {author} {\bibfnamefont {Y.}~\bibnamefont {Li}},\
	}\bibfield  {title} {\bibinfo {title} {Quantum algorithm for the simulation
			of open-system dynamics and thermalization},\ }\href
	{https://doi.org/10.1103/PhysRevA.101.012328} {\bibfield  {journal} {\bibinfo
			{journal} {Phys. Rev. A}\ }\textbf {\bibinfo {volume} {101}},\ \bibinfo
		{pages} {012328} (\bibinfo {year} {2020})}\BibitemShut {NoStop}%
	\bibitem [{\citenamefont {Cattaneo}\ \emph {et~al.}(2021)\citenamefont
		{Cattaneo}, \citenamefont {De~Chiara}, \citenamefont {Maniscalco},
		\citenamefont {Zambrini},\ and\ \citenamefont
		{Giorgi}}]{cattaneo2021collision}%
	\BibitemOpen
	\bibfield  {author} {\bibinfo {author} {\bibfnamefont {M.}~\bibnamefont
			{Cattaneo}}, \bibinfo {author} {\bibfnamefont {G.}~\bibnamefont {De~Chiara}},
		\bibinfo {author} {\bibfnamefont {S.}~\bibnamefont {Maniscalco}}, \bibinfo
		{author} {\bibfnamefont {R.}~\bibnamefont {Zambrini}},\ and\ \bibinfo
		{author} {\bibfnamefont {G.~L.}\ \bibnamefont {Giorgi}},\ }\bibfield  {title}
	{\bibinfo {title} {Collision models can efficiently simulate any multipartite
			markovian quantum dynamics},\ }\href
	{https://doi.org/10.1103/PhysRevLett.126.130403} {\bibfield  {journal}
		{\bibinfo  {journal} {Phys. Rev. Lett.}\ }\textbf {\bibinfo {volume} {126}},\
		\bibinfo {pages} {130403} (\bibinfo {year} {2021})}\BibitemShut {NoStop}%
	\bibitem [{\citenamefont {Belovs}(2019)}]{belovs2019quantum}%
	\BibitemOpen
	\bibfield  {author} {\bibinfo {author} {\bibfnamefont {A.}~\bibnamefont
			{Belovs}},\ }\bibfield  {title} {\bibinfo {title} {Quantum algorithms for
			classical probability distributions},\ }\href
	{https://arxiv.org/abs/1904.02192} {\bibfield  {journal} {\bibinfo  {journal}
			{arXiv preprint arXiv:1904.02192}\ } (\bibinfo {year} {2019})}\BibitemShut
	{NoStop}%
	\bibitem [{\citenamefont {Gily{\'e}n}\ and\ \citenamefont
		{Li}(2019)}]{gilyen2019distributional}%
	\BibitemOpen
	\bibfield  {author} {\bibinfo {author} {\bibfnamefont {A.}~\bibnamefont
			{Gily{\'e}n}}\ and\ \bibinfo {author} {\bibfnamefont {T.}~\bibnamefont
			{Li}},\ }\bibfield  {title} {\bibinfo {title} {Distributional property
			testing in a quantum world},\ }\href {https://arxiv.org/abs/1902.00814}
	{\bibfield  {journal} {\bibinfo  {journal} {arXiv preprint arXiv:1902.00814}\
		} (\bibinfo {year} {2019})}\BibitemShut {NoStop}%
	\bibitem [{\citenamefont {Gur}\ \emph {et~al.}(2021)\citenamefont {Gur},
		\citenamefont {Hsieh},\ and\ \citenamefont {Subramanian}}]{gur2021sublinear}%
	\BibitemOpen
	\bibfield  {author} {\bibinfo {author} {\bibfnamefont {T.}~\bibnamefont
			{Gur}}, \bibinfo {author} {\bibfnamefont {M.-H.}\ \bibnamefont {Hsieh}},\
		and\ \bibinfo {author} {\bibfnamefont {S.}~\bibnamefont {Subramanian}},\
	}\bibfield  {title} {\bibinfo {title} {Sublinear quantum algorithms for
			estimating von neumann entropy},\ }\href {https://arxiv.org/abs/2111.11139}
	{\bibfield  {journal} {\bibinfo  {journal} {arXiv preprint arXiv:2111.11139}\
		} (\bibinfo {year} {2021})}\BibitemShut {NoStop}%
	\bibitem [{\citenamefont {Kim}\ \emph {et~al.}(2020)\citenamefont {Kim},
		\citenamefont {Tang},\ and\ \citenamefont {Preskill}}]{Kim2020gravity}%
	\BibitemOpen
	\bibfield  {author} {\bibinfo {author} {\bibfnamefont {I.}~\bibnamefont
			{Kim}}, \bibinfo {author} {\bibfnamefont {E.}~\bibnamefont {Tang}},\ and\
		\bibinfo {author} {\bibfnamefont {J.}~\bibnamefont {Preskill}},\ }\bibfield
	{title} {\bibinfo {title} {The ghost in the radiation: robust encodings of
			the black hole interior},\ }\href {https://doi.org/10.1007/JHEP06(2020)031}
	{\bibfield  {journal} {\bibinfo  {journal} {J. High Energy Phys.}\ }\textbf
		{\bibinfo {volume} {2020}}\bibinfo  {number} { (6)},\ \bibinfo {pages}
		{31}}\BibitemShut {NoStop}%
	\bibitem [{\citenamefont {Huang}\ \emph {et~al.}(2023)\citenamefont {Huang},
		\citenamefont {Tong}, \citenamefont {Fang},\ and\ \citenamefont
		{Su}}]{huang2023Hamiltonian}%
	\BibitemOpen
	\bibfield  {number} {  }\bibfield  {author} {\bibinfo {author} {\bibfnamefont
			{H.-Y.}\ \bibnamefont {Huang}}, \bibinfo {author} {\bibfnamefont
			{Y.}~\bibnamefont {Tong}}, \bibinfo {author} {\bibfnamefont {D.}~\bibnamefont
			{Fang}},\ and\ \bibinfo {author} {\bibfnamefont {Y.}~\bibnamefont {Su}},\
	}\bibfield  {title} {\bibinfo {title} {Learning many-body hamiltonians with
			heisenberg-limited scaling},\ }\href
	{https://doi.org/10.1103/PhysRevLett.130.200403} {\bibfield  {journal}
		{\bibinfo  {journal} {Phys. Rev. Lett.}\ }\textbf {\bibinfo {volume} {130}},\
		\bibinfo {pages} {200403} (\bibinfo {year} {2023})}\BibitemShut {NoStop}%
	\bibitem [{\citenamefont {Li}\ \emph {et~al.}(2024)\citenamefont {Li},
		\citenamefont {Tong}, \citenamefont {Gefen}, \citenamefont {Ni},\ and\
		\citenamefont {Ying}}]{Li2024hamiltonian}%
	\BibitemOpen
	\bibfield  {author} {\bibinfo {author} {\bibfnamefont {H.}~\bibnamefont
			{Li}}, \bibinfo {author} {\bibfnamefont {Y.}~\bibnamefont {Tong}}, \bibinfo
		{author} {\bibfnamefont {T.}~\bibnamefont {Gefen}}, \bibinfo {author}
		{\bibfnamefont {H.}~\bibnamefont {Ni}},\ and\ \bibinfo {author}
		{\bibfnamefont {L.}~\bibnamefont {Ying}},\ }\bibfield  {title} {\bibinfo
		{title} {Heisenberg-limited hamiltonian learning for interacting bosons},\
	}\href {https://doi.org/10.1038/s41534-024-00881-2} {\bibfield  {journal}
		{\bibinfo  {journal} {npj Quantum Information}\ }\textbf {\bibinfo {volume}
			{10}},\ \bibinfo {pages} {83} (\bibinfo {year} {2024})}\BibitemShut {NoStop}%
	\bibitem [{\citenamefont {Yu}\ \emph {et~al.}(2023)\citenamefont {Yu},
		\citenamefont {Sun}, \citenamefont {Han},\ and\ \citenamefont
		{Yuan}}]{Yu2023hamiltonian}%
	\BibitemOpen
	\bibfield  {author} {\bibinfo {author} {\bibfnamefont {W.}~\bibnamefont
			{Yu}}, \bibinfo {author} {\bibfnamefont {J.}~\bibnamefont {Sun}}, \bibinfo
		{author} {\bibfnamefont {Z.}~\bibnamefont {Han}},\ and\ \bibinfo {author}
		{\bibfnamefont {X.}~\bibnamefont {Yuan}},\ }\bibfield  {title} {\bibinfo
		{title} {Robust and {E}fficient {H}amiltonian {L}earning},\ }\href
	{https://doi.org/10.22331/q-2023-06-29-1045} {\bibfield  {journal} {\bibinfo
			{journal} {{Quantum}}\ }\textbf {\bibinfo {volume} {7}},\ \bibinfo {pages}
		{1045} (\bibinfo {year} {2023})}\BibitemShut {NoStop}%
	\bibitem [{\citenamefont {Bluhm}\ \emph {et~al.}(2024)\citenamefont {Bluhm},
		\citenamefont {Caro},\ and\ \citenamefont {Oufkir}}]{bluhm2024hamiltonian}%
	\BibitemOpen
	\bibfield  {author} {\bibinfo {author} {\bibfnamefont {A.}~\bibnamefont
			{Bluhm}}, \bibinfo {author} {\bibfnamefont {M.~C.}\ \bibnamefont {Caro}},\
		and\ \bibinfo {author} {\bibfnamefont {A.}~\bibnamefont {Oufkir}},\
	}\bibfield  {title} {\bibinfo {title} {Hamiltonian property testing},\ }\href
	{https://arxiv.org/abs/2403.02968} {\bibfield  {journal} {\bibinfo  {journal}
			{arXiv preprint arXiv:2403.02968}\ } (\bibinfo {year} {2024})}\BibitemShut
	{NoStop}%
	\bibitem [{\citenamefont {Dutkiewicz}\ \emph {et~al.}(2024)\citenamefont
		{Dutkiewicz}, \citenamefont {O'Brien},\ and\ \citenamefont
		{Schuster}}]{Dutkiewicz2024hamiltonian}%
	\BibitemOpen
	\bibfield  {author} {\bibinfo {author} {\bibfnamefont {A.}~\bibnamefont
			{Dutkiewicz}}, \bibinfo {author} {\bibfnamefont {T.~E.}\ \bibnamefont
			{O'Brien}},\ and\ \bibinfo {author} {\bibfnamefont {T.}~\bibnamefont
			{Schuster}},\ }\bibfield  {title} {\bibinfo {title} {The advantage of quantum
			control in many-body {H}amiltonian learning},\ }\href
	{https://doi.org/10.22331/q-2024-11-26-1537} {\bibfield  {journal} {\bibinfo
			{journal} {{Quantum}}\ }\textbf {\bibinfo {volume} {8}},\ \bibinfo {pages}
		{1537} (\bibinfo {year} {2024})}\BibitemShut {NoStop}%
	\bibitem [{\citenamefont {Elben}\ \emph {et~al.}(2020)\citenamefont {Elben},
		\citenamefont {Kueng}, \citenamefont {Huang}, \citenamefont {van Bijnen},
		\citenamefont {Kokail}, \citenamefont {Dalmonte}, \citenamefont {Calabrese},
		\citenamefont {Kraus}, \citenamefont {Preskill}, \citenamefont {Zoller},\
		and\ \citenamefont {Vermersch}}]{elben2020mixed}%
	\BibitemOpen
	\bibfield  {author} {\bibinfo {author} {\bibfnamefont {A.}~\bibnamefont
			{Elben}}, \bibinfo {author} {\bibfnamefont {R.}~\bibnamefont {Kueng}},
		\bibinfo {author} {\bibfnamefont {H.-Y.~R.}\ \bibnamefont {Huang}}, \bibinfo
		{author} {\bibfnamefont {R.}~\bibnamefont {van Bijnen}}, \bibinfo {author}
		{\bibfnamefont {C.}~\bibnamefont {Kokail}}, \bibinfo {author} {\bibfnamefont
			{M.}~\bibnamefont {Dalmonte}}, \bibinfo {author} {\bibfnamefont
			{P.}~\bibnamefont {Calabrese}}, \bibinfo {author} {\bibfnamefont
			{B.}~\bibnamefont {Kraus}}, \bibinfo {author} {\bibfnamefont
			{J.}~\bibnamefont {Preskill}}, \bibinfo {author} {\bibfnamefont
			{P.}~\bibnamefont {Zoller}},\ and\ \bibinfo {author} {\bibfnamefont
			{B.}~\bibnamefont {Vermersch}},\ }\bibfield  {title} {\bibinfo {title}
		{Mixed-state entanglement from local randomized measurements},\ }\href
	{https://doi.org/10.1103/PhysRevLett.125.200501} {\bibfield  {journal}
		{\bibinfo  {journal} {Phys. Rev. Lett.}\ }\textbf {\bibinfo {volume} {125}},\
		\bibinfo {pages} {200501} (\bibinfo {year} {2020})}\BibitemShut {NoStop}%
	\bibitem [{\citenamefont {Zhou}\ \emph {et~al.}(2020)\citenamefont {Zhou},
		\citenamefont {Zeng},\ and\ \citenamefont {Liu}}]{zhou2020single}%
	\BibitemOpen
	\bibfield  {author} {\bibinfo {author} {\bibfnamefont {Y.}~\bibnamefont
			{Zhou}}, \bibinfo {author} {\bibfnamefont {P.}~\bibnamefont {Zeng}},\ and\
		\bibinfo {author} {\bibfnamefont {Z.}~\bibnamefont {Liu}},\ }\bibfield
	{title} {\bibinfo {title} {Single-copies estimation of entanglement
			negativity},\ }\href {https://doi.org/10.1103/PhysRevLett.125.200502}
	{\bibfield  {journal} {\bibinfo  {journal} {Phys. Rev. Lett.}\ }\textbf
		{\bibinfo {volume} {125}},\ \bibinfo {pages} {200502} (\bibinfo {year}
		{2020})}\BibitemShut {NoStop}%
	\bibitem [{\citenamefont {Shapourian}\ \emph {et~al.}(2021)\citenamefont
		{Shapourian}, \citenamefont {Liu}, \citenamefont {Kudler-Flam},\ and\
		\citenamefont {Vishwanath}}]{liu2021negativity}%
	\BibitemOpen
	\bibfield  {author} {\bibinfo {author} {\bibfnamefont {H.}~\bibnamefont
			{Shapourian}}, \bibinfo {author} {\bibfnamefont {S.}~\bibnamefont {Liu}},
		\bibinfo {author} {\bibfnamefont {J.}~\bibnamefont {Kudler-Flam}},\ and\
		\bibinfo {author} {\bibfnamefont {A.}~\bibnamefont {Vishwanath}},\ }\bibfield
	{title} {\bibinfo {title} {Entanglement negativity spectrum of random mixed
			states: A diagrammatic approach},\ }\href
	{https://doi.org/10.1103/PRXQuantum.2.030347} {\bibfield  {journal} {\bibinfo
			{journal} {PRX Quantum}\ }\textbf {\bibinfo {volume} {2}},\ \bibinfo {pages}
		{030347} (\bibinfo {year} {2021})}\BibitemShut {NoStop}%
	\bibitem [{\citenamefont {King}\ \emph {et~al.}()\citenamefont {King},
		\citenamefont {Gosset}, \citenamefont {Kothari},\ and\ \citenamefont
		{Babbush}}]{rovvie2025triple}%
	\BibitemOpen
	\bibfield  {author} {\bibinfo {author} {\bibfnamefont {R.}~\bibnamefont
			{King}}, \bibinfo {author} {\bibfnamefont {D.}~\bibnamefont {Gosset}},
		\bibinfo {author} {\bibfnamefont {R.}~\bibnamefont {Kothari}},\ and\ \bibinfo
		{author} {\bibfnamefont {R.}~\bibnamefont {Babbush}},\ }\bibinfo {title}
	{Triply efficient shadow tomography},\ in\ \href
	{https://doi.org/10.1137/1.9781611978322.27} {\emph {\bibinfo {booktitle}
			{Proceedings of the 2025 Annual ACM-SIAM Symposium on Discrete Algorithms
				(SODA)}}},\ pp.\ \bibinfo {pages} {914--946}\BibitemShut {NoStop}%
	\bibitem [{\citenamefont {King}\ \emph {et~al.}(2024)\citenamefont {King},
		\citenamefont {Wan},\ and\ \citenamefont {McClean}}]{robbie2024conjugate}%
	\BibitemOpen
	\bibfield  {author} {\bibinfo {author} {\bibfnamefont {R.}~\bibnamefont
			{King}}, \bibinfo {author} {\bibfnamefont {K.}~\bibnamefont {Wan}},\ and\
		\bibinfo {author} {\bibfnamefont {J.~R.}\ \bibnamefont {McClean}},\
	}\bibfield  {title} {\bibinfo {title} {Exponential learning advantages with
			conjugate states and minimal quantum memory},\ }\href
	{https://doi.org/10.1103/PRXQuantum.5.040301} {\bibfield  {journal} {\bibinfo
			{journal} {PRX Quantum}\ }\textbf {\bibinfo {volume} {5}},\ \bibinfo {pages}
		{040301} (\bibinfo {year} {2024})}\BibitemShut {NoStop}%
	\bibitem [{\citenamefont {Watrous}(2018)}]{watrous2018theory}%
	\BibitemOpen
	\bibfield  {author} {\bibinfo {author} {\bibfnamefont {J.}~\bibnamefont
			{Watrous}},\ }\href@noop {} {\emph {\bibinfo {title} {The theory of quantum
				information}}}\ (\bibinfo  {publisher} {Cambridge university press},\
	\bibinfo {year} {2018})\BibitemShut {NoStop}%
	\bibitem [{\citenamefont {Aharonov}\ and\ \citenamefont
		{Ben-Or}(1997)}]{aharonov1997fault}%
	\BibitemOpen
	\bibfield  {author} {\bibinfo {author} {\bibfnamefont {D.}~\bibnamefont
			{Aharonov}}\ and\ \bibinfo {author} {\bibfnamefont {M.}~\bibnamefont
			{Ben-Or}},\ }\bibfield  {title} {\bibinfo {title} {Fault-tolerant quantum
			computation with constant error},\ }in\ \href
	{https://doi.org/10.1145/258533.258579} {\emph {\bibinfo {booktitle}
			{Proceedings of the Twenty-ninth Annual ACM Symposium on Theory of
				Computing}}}\ (\bibinfo {year} {1997})\ pp.\ \bibinfo {pages}
	{176--188}\BibitemShut {NoStop}%
	\bibitem [{\citenamefont {Chen}\ \emph
		{et~al.}(2023{\natexlab{c}})\citenamefont {Chen}, \citenamefont {Cotler},
		\citenamefont {Huang},\ and\ \citenamefont {Li}}]{chen2023complexity}%
	\BibitemOpen
	\bibfield  {author} {\bibinfo {author} {\bibfnamefont {S.}~\bibnamefont
			{Chen}}, \bibinfo {author} {\bibfnamefont {J.}~\bibnamefont {Cotler}},
		\bibinfo {author} {\bibfnamefont {H.-Y.}\ \bibnamefont {Huang}},\ and\
		\bibinfo {author} {\bibfnamefont {J.}~\bibnamefont {Li}},\ }\bibfield
	{title} {\bibinfo {title} {The complexity of {NISQ}},\ }\href
	{https://doi.org/10.1038/s41467-023-41217-6} {\bibfield  {journal} {\bibinfo
			{journal} {Nat. Commun.}\ }\textbf {\bibinfo {volume} {14}},\ \bibinfo
		{pages} {6001} (\bibinfo {year} {2023}{\natexlab{c}})}\BibitemShut {NoStop}%
	\bibitem [{\citenamefont {Yu}(1997)}]{yu1997assouad}%
	\BibitemOpen
	\bibfield  {author} {\bibinfo {author} {\bibfnamefont {B.}~\bibnamefont
			{Yu}},\ }\bibfield  {title} {\bibinfo {title} {Assouad, {F}ano, and {L}e
			{C}am},\ }in\ \href@noop {} {\emph {\bibinfo {booktitle} {Festschrift for
				Lucien Le Cam: research papers in probability and statistics}}}\ (\bibinfo
	{publisher} {Springer},\ \bibinfo {year} {1997})\ pp.\ \bibinfo {pages}
	{423--435}\BibitemShut {NoStop}%
\end{thebibliography}

%

\end{document}